%% file: main.tex
\documentclass{arxiv}
\usepackage{graphicx,amsmath,amssymb,amsfonts,amsthm,mathrsfs}
\numberwithin{equation}{section}
\usepackage{appendix}
\usepackage{booktabs,array,tabularx,enumitem,tikz,capt-of}
\usepackage[section]{placeins}
\usepackage{needspace}
\usepackage{microtype}
\usepackage[loadonly]{titlesec}
\titleclass{\subsection}[2]{straight}
\titleformat{\subsection}[block]
  {\fontsize{11}{13}\selectfont\bfseries}{\mdseries\thesubsection.}{0.5em}{}
\titlespacing{\subsection}{0pt}
  {0.5\linespacing plus 0.25\linespacing}{0.3\linespacing}
\usepackage[colorlinks=true,allcolors=blue]{hyperref}
\usetikzlibrary{arrows.meta,positioning,calc}
\input{diagrams/styles}
\theoremstyle{plain}
\newtheorem{theorem}[thrm]{Theorem}
\newtheorem{proposition}[thrm]{Proposition}
\newtheorem{lemma}[thrm]{Lemma}

\theoremstyle{definition}
\newtheorem{definition}[thrm]{Definition}

\theoremstyle{definition}

\newcommand{\R}{\mathbb R}
\newcommand{\Z}{\mathbb Z}

\newcommand{\C}{\mathbb C}
\newcommand{\LL}{L^2(0,1)}
\newcommand{\ip}[2]{\langle #1,#2\rangle}
\newcommand{\hp}[2]{\langle #1,#2\rangle_{\mathcal H}}
\newcommand{\Ap}{\mathcal A'_{\lambda}}
\newcommand{\Aop}{A_{\lambda}}
\newcommand{\dd}{\,\mathrm d}
\newcommand{\qd}[1]{#1^{[3]}}
\allowdisplaybreaks
\begin{document}
\runningtitle{Rapid Fredholm stabilization of the Kuramoto--Sivashinsky equation with unrestricted, spatially-varying anti-diffusion.}
\runningauthors{L. Bhan, M. Krstic and Y. Shi}
\title{Rapid Fredholm stabilization of the Kuramoto--Sivashinsky equation with unrestricted, spatially-varying anti-diffusion.}
\author{Luke Bhan}
\address{University of California, San Diego, 3180 Voigt Drive,
La Jolla, CA 92093, USA.\\
\email{lbhan@ucsd.edu, mkrstic@ucsd.edu, yyshi@ucsd.edu}\\
Corresponding author: Luke Bhan (\texttt{lbhan@ucsd.edu}).}
\author{Miroslav Krstic}\sameaddress{1}
\author{Yuanyuan Shi}\sameaddress{1}
\date{}
\begin{abstract}
In this work, we develop the first feedback design for rapid
stabilization of the Kuramoto--Sivashinsky equation with a spatially varying
anti-diffusion coefficient. For constant coefficients, the single-input Fredholm
design of Coron and L\"u (2015) excludes a discrete set of values at which
repeated unstable eigenvalues cause a loss of controllability. We overcome this
obstruction by introducing a second boundary input and assigning the two inputs
distinct roles. The key idea, inspired by Heymann's Lemma, is to use the
boundary value $u(0,t)$ entirely for a pre-feedback that renders the modified
plant controllable through the curvature input $u_{xx}(0,t)$. The latter input
then stabilizes the plant through a Fredholm backstepping transformation.
We show that two inputs suffice for controllability and are necessary
when the plant has an unstable double eigenvalue.
However, the Fredholm kernel still must be approximated for implementation. Hence,
to enable kernel and gain approximation, we prove continuity of the
coefficient-to-gain design map on compact admissible design classes.
Unlike Volterra-based continuity proofs using successive approximations,
our proof uses the modal representation to control the spectral data,
the inverse coefficient system, and the tails of the kernel and gain
series. This yields a single neural operator approximation of the gain to
any prescribed $L^2$ accuracy across the class.
Finally, we establish rapid local stabilization of the nonlinear
closed-loop system under both the exact gains and sufficiently accurate
approximations.
We conclude with numerical results that illustrate prescribed decay rates
and the accuracy and computational cost of the approximations. In particular, we train a Fourier neural operator that
achieves typical relative gain errors of approximately $0.1\%$ and
stabilizes all held-out cases tested, including a plant with an unstable
double eigenvalue.
\end{abstract}
\maketitle


\section{Introduction}\label{sec:intro}

\subsection{Summary of results}\label{ssec:summary}

We study the Kuramoto--Sivashinsky (KS) equation with a spatially varying
anti-diffusion coefficient and two boundary inputs:
\begin{subequations}\label{eq:plant}
\begin{align}
u_t+u_{xxxx}+\big(\lambda(x)\,u_x\big)_x+u\,u_x&=0,\label{eq:plant-a}\\
u(1,t)=u_{xx}(1,t)&=0,\\
u(0,t)&=U_1(t),\\
u_{xx}(0,t)&=U_2(t),
\end{align}
\end{subequations}
on $(x,t)\in(0,1)\times(0,\infty)$, with $u(\cdot,0)=u_0\in\LL$.
The inputs $U_1$ and $U_2$ act on boundary displacement and curvature, respectively.
For constant $\lambda$ and
$U_1\equiv0$, the KS system in \eqref{eq:plant} is equivalent to the system in
\cite{CoronLu2015}, whose design requires
\begin{align}
\lambda\notin\mathscr N_\lambda:=\{(i^2+j^2)\pi^2:\ i,j\in\Z^+,\ i\neq j\},\label{eq:CL-restriction}
\end{align}
where, at these values, the linearized plant has an unstable double eigenvalue that no linear feedback through $U_2$ alone can move (Proposition~\ref{prop:obstruction}).

We showcase that, with two inputs, the restriction of \eqref{eq:CL-restriction} is removed. Moreover, we upgrade $\lambda$ to be any function in $W^{1,\infty}(0,1)$ by assigning the two inputs distinct roles. The displacement input $U_1$ is used, in the sense of Heymann's Lemma \cite{Heymann1968}, entirely for a pre-feedback that makes the plant controllable through the curvature input $U_2$, and $U_2$ carries a Fredholm backstepping feedback (Figure~\ref{fig:design-overview}). The first of our two main results is the following.

\begin{theorem}[Two boundary feedbacks stabilize the Kuramoto--Sivashinsky equation at any rate for every anti-diffusion coefficient]\label{thm:intro}
Consider the plant \eqref{eq:plant} with $\lambda\in W^{1,\infty}(0,1)$ under the feedback laws
\begin{align}
U_1(t)=\varepsilon\int_0^1e^{-\theta y}\,u(y,t)\dd y,\qquad
U_2(t)=\int_0^1h(y)\,u(y,t)\dd y,\label{eq:intro-laws}
\end{align}
with the gain
\begin{align}
h(y)=\sum_{n\ge1}\frac{a+\tilde\kappa_n}{\tilde\chi_n'(0)}\,\tilde\chi_n(y),\label{eq:intro-gain}
\end{align}
where $\tilde\mu_n$, $n\ge1$, are the eigenvalues, with eigenfunctions $\tilde\chi_n$, of
\begin{align}
-\tilde\chi''''-(\lambda\tilde\chi')'-\varepsilon\big(\tilde\chi'''(0)+\lambda(0)\tilde\chi'(0)\big)e^{-\theta x}=\tilde\mu\,\tilde\chi,\qquad
\tilde\chi(0)=\tilde\chi''(0)=\tilde\chi(1)=\tilde\chi''(1)=0,\label{eq:intro-eig}
\end{align}
and $(\tilde\kappa_n)$ is the bounded solution of
\begin{align}
\tilde\kappa_m=-a\sum_{n\ne m}\frac{a+\tilde\kappa_n}{\tilde\mu_n-\tilde\mu_m+a},\qquad m\ge1.\label{eq:intro-kappa}
\end{align}
There exist a discrete set $\mathscr T_\lambda\subset(0,\infty)$ and a positive function $\varepsilon_1$ on $(0,\infty)\setminus\mathscr T_\lambda$ such that, for all $\theta\notin\mathscr T_\lambda$, $0<|\varepsilon|<\varepsilon_1(\theta)$, and $a>\max\{0,\max_n\operatorname{Re}\tilde\mu_n\}$ with $a\notin\{\tilde\mu_i-\tilde\mu_j:\ i,j\ge1\}$, every $\tilde\mu_n$ is algebraically simple, $\tilde\chi_n'(0)\ne0$, the bounded solution $(\tilde\kappa_n)$ of \eqref{eq:intro-kappa} exists and is unique, the series \eqref{eq:intro-gain} converges in $\LL$ to a real-valued function, and, for every $\omega\in\big(0,\,a-\max_n\operatorname{Re}\tilde\mu_n\big)$, there exist $\rho_0,C,\epsilon^*>0$ such that, for every $u_0\in\LL$ with $\|u_0\|\le\rho_0$, the closed-loop system \eqref{eq:plant}, \eqref{eq:intro-laws}, with $h$ in $U_2$ defined by \eqref{eq:intro-gain} or replaced by any $\hat h\in\LL$ satisfying $\|\hat h-h\|\le\epsilon^*$, has a unique mild solution $u\in C([0,\infty);\LL)\cap L^2_{\rm loc}([0,\infty);H^2(0,1))$ satisfying
\begin{align}
\|u(t)\|\le C\,e^{-\omega t}\,\|u_0\|,\qquad t\ge0.\label{eq:intro-decay}
\end{align}
\end{theorem}

Theorem~\ref{thm:intro} concerns the controller \eqref{eq:intro-laws}--\eqref{eq:intro-kappa} with the gain $h$ computed exactly or approximated within $\epsilon^*$. Since $a$ is restricted only from below and away from a discrete set, the decay rate $\omega$ can be prescribed arbitrarily, for every $\lambda\in W^{1,\infty}(0,1)$. The assertions on the eigenvalues, the gain and the coefficients $\tilde\kappa_n$ are proved in Sections~\ref{sec:existence}--\ref{sec:fredholm} (Theorems~\ref{thm:generic} and~\ref{thm:kernel}), the closed-loop assertion in Section~\ref{sec:target} (Theorem~\ref{thm:nominal} for $h$, Theorem~\ref{thm:main} for $\hat h$), and well-posedness in Appendix~\ref{app:well-posedness}.

The gain \eqref{eq:intro-gain} is defined through the spectral data of \eqref{eq:intro-eig} and the infinite system \eqref{eq:intro-kappa}, and it must be recomputed for each $\lambda$ and each choice of $(a,\varepsilon,\theta)$. Our second main result replaces this computation by the evaluation of a single neural operator, trained once on a compact class of coefficients and design parameters.

\begin{theorem}[A single neural operator stabilizes every plant in a compact class at a uniform rate]\label{thm:intro-NO}
For $B>0$ and $d\in(0,1]$, let
\begin{align}
\mathscr U_{B,d}:=\Big\{&\mathbf p=(\lambda,a,\varepsilon,\theta){}\in W^{2,\infty}(0,1)\times(0,\infty)\times\R\times(0,\infty):\ 
\|\lambda\|_{W^{2,\infty}}\le B,\ \ a+|\varepsilon|+\theta+\theta^{-1}\le B,\notag\\
&\min_{i\ne j}\frac{|\tilde\mu_i-\tilde\mu_j|}{\max(i,j)^3}\ge d,\ \ 
\min_n\frac{|\tilde\chi_n'(0)|}{n\,\|\tilde\chi_n\|}\ge d,\ \ 
\operatorname{dist}\big(a,\{\tilde\mu_i-\tilde\mu_j\}\big)\ge d,\ \ 
a-\max_n\operatorname{Re}\tilde\mu_n\ge d\Big\},\label{eq:intro-class}
\end{align}
where $\tilde\mu_n$, repeated according to algebraic multiplicity and indexed so that $\operatorname{Re}\tilde\mu_1\ge\operatorname{Re}\tilde\mu_2\ge\cdots$, are the eigenvalues of \eqref{eq:intro-eig}, with eigenfunctions $\tilde\chi_n$, let $\mathcal G(\mathbf p)$ denote the gain $h$ in \eqref{eq:intro-gain}, and consider, for a map $\hat{\mathcal G}:\mathscr U_{B,d}\to\LL$ and $\mathbf p\in\mathscr U_{B,d}$, the plant \eqref{eq:plant} with coefficient $\lambda$ under the feedback laws
\begin{align}
U_1(t)=\varepsilon\int_0^1e^{-\theta y}\,u(y,t)\dd y,\qquad
U_2(t)=\int_0^1\hat{\mathcal G}(\mathbf p)(y)\,u(y,t)\dd y.\label{eq:intro-NO-laws}
\end{align}
For every $\epsilon>0$ there exists a neural operator $\hat{\mathcal G}$, with function input $\lambda$, scalar inputs $(a,\varepsilon,\theta)$, and output in $\LL$, satisfying $\sup_{\mathbf p\in\mathscr U_{B,d}}\|\hat{\mathcal G}(\mathbf p)-\mathcal G(\mathbf p)\|\le\epsilon$, and there exist $\epsilon^*_{B,d},\rho_{B,d},C_{B,d}>0$ such that, for every map $\hat{\mathcal G}$ with $\sup_{\mathbf p\in\mathscr U_{B,d}}\|\hat{\mathcal G}(\mathbf p)-\mathcal G(\mathbf p)\|\le\epsilon^*_{B,d}$, every $\mathbf p\in\mathscr U_{B,d}$, and every $u_0\in\LL$ with $\|u_0\|\le\rho_{B,d}$, the closed-loop system \eqref{eq:plant}, \eqref{eq:intro-NO-laws} has a unique mild solution $u\in C([0,\infty);\LL)\cap L^2_{\rm loc}([0,\infty);H^2(0,1))$ satisfying
\begin{align}
\|u(t)\|\le C_{B,d}\,e^{-dt/4}\,\|u_0\|,\qquad t\ge0.\label{eq:intro-NO-decay}
\end{align}
\end{theorem}

The class $\mathscr U_{B,d}$ bounds $\lambda$ in $W^{2,\infty}$, one derivative above Theorem~\ref{thm:intro}, so that it is compact in $W^{1,\infty}(0,1)\times\R^3$. By Theorem~\ref{thm:intro}, every $\lambda\in W^{2,\infty}(0,1)$ admits designs in $\mathscr U_{B,d}$ for suitable $B$ and $d$. The rate is uniform over the class, and $d$ caps it, whereas in Theorem~\ref{thm:intro} it is arbitrary for a single plant. The tolerance applies to any approximate gain, whether produced by a neural operator or by a truncated modal computation. Theorem~\ref{thm:intro-NO} is proved in Section~\ref{sec:approx-stability}, from continuity of $\mathcal G$ and uniform bounds on the design constants (Theorem~\ref{thm:lipschitz}), neural-operator approximation (Theorem~\ref{thm:NO}), and stability under gain error (Theorem~\ref{thm:main}). Section~\ref{sec:numerics} illustrates both results on a plant with an unstable double eigenvalue, with a Fourier neural operator trained on a family of coefficients. We now review the related literature.

\subsection{Related work}\label{ssec:related}

\subsubsection*{A. Fredholm backstepping}
The Fredholm-based backstepping design originates in \cite{SGK2009}, where an
Euler--Bernoulli beam is stabilized at an arbitrary decay rate from displacement
and moment actuation at one end. Fredholm transformations were used again in
\cite{BribiescaKrstic2015} for first-order hyperbolic equations, where the
non-strict-feedback structure rules out a Volterra transformation and the kernel
is obtained by a fixed point on the square. Concurrently, and by a different
technique, Coron and coauthors developed Fredholm designs for the
Korteweg--de Vries equation \cite{CoronLu2014}, the Kuramoto--Sivashinsky
equation \cite{CoronLu2015}, integro-differential hyperbolic equations
\cite{CoronHuOlive2016}, a linearized bilinear Schr\"odinger equation
\cite{CoronGagnonMorancey2018} and the linearized water tank
\cite{CoronHayatXiangZhang2022}. The design relies on a transformation built from the spectral data, together with an
equivalence of the closed-loop operator with respect to a shifted copy of the plant
operator. The method has since been abstracted to general linear systems
with a Riesz basis of eigenvectors \cite{HayatLoko2026} and to the parabolic
case \cite{BoulardHayat2025}. Gagnon et al.\ \cite{GHXZ2025WaterWaves}
extended the spectral construction to skew-adjoint operators of order
greater than one through a compactness/duality argument, with an
application to linearized water waves. More recently,
\cite{GHMXZ2026Quantitative} gives explicit Fredholm constructions and
quantitative bounds on the transformation and its inverse as the
prescribed decay rate varies, for self-adjoint and skew-adjoint operators.
Here we study parameter dependence and gain approximation for a generally
non-self-adjoint pre-compensated operator. It is the design of
\cite{CoronLu2015}, in the abstract form of \cite{HayatLoko2026}, that
Section~\ref{sec:kernel} applies.

\subsubsection*{B. Boundary control of Kuramoto--Sivashinsky}
 For the KS PDE, the available set of stabilization results depend on both the anti-diffusion coefficient and the boundary conditions. Liu and Krstic \cite{LiuKrstic2001} established nonlinear
boundary stabilization for constant $\lambda<4\pi^2$. Their comparison
with the uncontrolled system uses clamped conditions $u=u_x=0$ at both
ends, which differ from the hinged conditions in \eqref{eq:plant}.
For the linear equation, Cerpa \cite{Cerpa2010} identified a countable
set of anti-diffusion parameters at which a single boundary input loses
null controllability. Two suitably chosen controls at the same end remove
this restriction \cite{Cerpa2010,CerpaGuzmanMercado2017}, with an earlier
result for $\lambda=0$ in \cite{LinGuo2002}. For the hinged configuration
considered here, Coron and L\"u \cite{CoronLu2015} constructed a Fredholm
feedback achieving local rapid stabilization with a single curvature input
when the constant coefficient satisfies \eqref{eq:CL-restriction}.
Our objective is to extend this design to spatially varying coefficients
while accommodating the repeated eigenvalues that obstruct single-input
stabilization.

\subsubsection*{C. Using pre-feedback to obtain controllability}
In finite dimensions, Heymann's Lemma \cite{Heymann1968}
states that any controllable multi-input system can be converted into a single-input controllable system using appropriate state feedback. A subsequent feedback through that input can then assign the closed-loop poles
\cite{Wonham1967,Wonham1985,Sontag1998}. Along this conceptual direction, there have been PDE designs that assign
different roles to multiple boundary inputs, as in the beam, Burgers,
and flow-control settings of \cite{KGBS2008,KMV2008,VazquezKrstic}.
For the periodic heat equation, Gagnon et al.\ \cite{GHXZ2022} proved
necessity and sufficiency of two scalar controls and constructed Fredholm
feedbacks on the odd and even subspaces, with an extension to local rapid
stabilization of the viscous Burgers equation. Our construction instead
uses one boundary feedback to establish controllability through the other
input before designing its Fredholm feedback.
Here we establish a corresponding construction for the KS equation:
feedback through the displacement input separates repeated eigenvalues
and makes every mode accessible to the curvature input. The latter
then implements the Fredholm stabilization law. Thus, the first stage
establishes the controllability needed by the second stage.

\subsubsection*{D. Numerical and neural operator gain approximation}

In PDE backstepping, the controller gain is itself obtained by solving a
kernel PDE \cite{KrsticSmyshlyaev2008}. Hence, numerical implementation of the gain has been a
significant subject in the backstepping literature, with various
approximations introduced. For example, Smyshlyaev and Krstic
\cite{SmyshlyaevKrstic2004} presented early results on approximate gains for
controlling parabolic PDEs using a finite-difference scheme for the kernel.
The gain accuracy affects closed-loop stability, so a sufficiently fine
resolution is needed to preserve stability. Subsequently, Vazquez et al.
\cite{VazquezChenQiaoKrstic2023} developed a power-series method that
represents the kernel by a truncated Taylor expansion and determines its
coefficients from the kernel equations to approximate the gain. Both
approaches compute gains for a specified plant coefficient.

Recently, \cite{BhanShiKrstic2024} introduced neural operators for learning
the mapping from plant coefficients to gains. The neural network learns
this mapping over a compact class of plant coefficients, allowing the
trained operator to be reused across the class. This direction has since
been applied to multiple PDE backstepping problems, including adaptive
control of hyperbolic \cite{LamarqueBhanShiKrstic2025} and parabolic PDEs
\cite{BhanShiKrstic2025}, adaptive freeway traffic control
\cite{LvWangZhangYu2025}, and $2\times2$ hyperbolic PDE systems
\cite{WangDiagneKrstic2025}.

A common structural restriction of these approaches is their reliance on
Volterra-type integral constructions for the gain operator. This differs
fundamentally from the Fredholm construction studied here: Lipschitz
continuity of Volterra gain operators is typically established using
successive approximations \cite{KrsticBhanShi2024}, whereas here we use a
modal representation. Establishing continuity through this representation
is therefore a necessary step toward neural-operator approximation of the
Fredholm gain.

\subsubsection*{E. Companion papers and distinctions}

The controllability and Fredholm design in
Sections~\ref{sec:control}--\ref{sec:fredholm} enable the companion papers
on stabilizing Poiseuille channel flow \cite{Krstic2026NavierStokesHeymann} and Yih's viscoisty jump between two fluids \cite{Krstic2026YihInterface}.
Compared with classical Poiseuille flow, which has at most one unstable
eigenmode per streamwise wavenumber \cite{RubanGajjarWalton2023},
the KS equation studied here must accommodate an arbitrarily large
unstable subspace and unstable double eigenvalues.
Moreover, this paper differs from \cite{Krstic2026YihInterface} in the
pre-feedback design. The two-fluid interface paper constructs pre-feedback using the plant's low mode spectral profile, allowing direct application of Heymann's Lemma to the problematic
modes \cite[Lemma~4]{Krstic2026YihInterface}.
Here, the pre-feedback is completely specified by an exponential profile with static scalar parameters, proving that this restricted family suffices for every admissible
plant coefficient and provides a consistent parameterization for continuity
and gain approximation.
Lastly, this paper differs from both companions by studying approximation
of the gain operator and establishing corresponding stability results
under the approximate gain.

\subsection{Contributions}\label{ssec:contributions}
This paper makes two principal contributions.
\begin{enumerate}
\item \emph{Rapid Fredholm stabilization of the Kuramoto–Sivashinsky equation with spatially varying anti-diffusion coefficient.} The existing Fredholm design \cite{CoronLu2015} is limited to constant anti-diffusion coefficients outside the critical set \eqref{eq:CL-restriction}, where the single controlled boundary trace loses controllability. We remove this restriction: the coefficient can be any $\lambda\in W^{1,\infty}(0,1)$. This is achieved through two control inputs. The first is a Heymann pre-feedback that makes every mode controllable. This enables the second curvature input which stabilizes the system via Fredholm backstepping. Although the Heymann--Fredholm architecture also appears in the companion Navier--Stokes and Yih designs \cite{Krstic2026NavierStokesHeymann,Krstic2026YihInterface}, its realization here is substantially more explicit: the complete controller consists of one exponential pre-feedback and one scalar Fredholm gain, stated together with its nonlinear rapid-stabilization guarantee in Theorem~\ref{thm:intro}.
\item \emph{The first neural-operator approximation of a Fredholm backstepping gain.} Existing neural-operator backstepping results approximate gains generated by Volterra transformations \cite{BhanShiKrstic2024,BhanShiKrstic2025,LamarqueBhanShiKrstic2025,LvWangZhangYu2025,WangDiagneKrstic2025}. A Fredholm gain instead depends globally on the spectral data of the plant and on an infinite coupled coefficient system, so the continuity arguments used for Volterra kernel maps do not apply. We prove continuity of the coefficient-and-design-to-gain map on compact admissible classes, obtain uniform approximation by a single neural operator, and prove that every sufficiently accurate approximate gain preserves local exponential stability uniformly across the class. This introduces neural-operator gain approximation into Fredholm backstepping and supplies the stability theorem required for its use in feedback. Neither this approximation result nor its robustness guarantee appears in the companion Fredholm designs \cite{Krstic2026NavierStokesHeymann,Krstic2026YihInterface}.
\end{enumerate}

The following technical innovations support these contributions.
\begin{enumerate}[label=(\roman*)]
\item We prove that two boundary inputs are sufficient for every spatially varying coefficient and necessary when the plant has an unstable double eigenvalue. We then show that the restricted exponential pre-feedback family suffices to split all repeated eigenvalues and make every mode controllable from the curvature input. This KS-specific construction is what gives Theorem~\ref{thm:intro} its unusually explicit form.
\item We establish continuity of the Fredholm gain despite the non-self-adjoint pre-compensated operator, its parameter-dependent domain, and eigenvalue collisions in the original plant. The proof combines resolvent and spectral-projection continuity with uniform control of the infinite coefficient system and the modal tails.
\item We prove nonlinear closed-loop robustness to $L^2$ error in the Fredholm gain, uniformly over compact admissible design classes. This estimate is the step that converts neural-operator approximation of the gain into a uniform stabilization result.
\end{enumerate}

\subsection{Organization of the paper}
Sections~\ref{sec:control}--\ref{sec:fredholm} establish controllability,
construct the feedback, and prove its existence and invertibility.
Section~\ref{sec:continuity} establishes continuity of the gain and kernel
and neural-operator approximation. Section~\ref{sec:target}
establishes stability under exact and approximate gains, followed by numerical results in Section~\ref{sec:numerics}
and conclusions in Section~\ref{sec:conclusion}.
Table~\ref{tab:main-results} summarizes the main results.
The principal proofs accompany their results in the main text.
The appendices contain the kernel-equation verification, Sobolev bounds
for the modal coordinate map, and closed-loop well-posedness arguments.

\begin{table}[htbp]
\caption{Main results and their roles in the design.}\label{tab:main-results}
\centering\normalsize
\renewcommand{\arraystretch}{1.15}
\begin{tabularx}{\textwidth}{@{}p{0.29\textwidth}Xp{0.11\textwidth}@{}}
\toprule
Result & Statement & Section\\
\midrule
Theorem~\ref{thm:control}, Proposition~\ref{prop:obstruction} & Two inputs suffice for every coefficient. A nonnegative double eigenvalue obstructs stabilization through either input alone. & \ref{sec:control}\\
Theorem~\ref{thm:generic} & Pre-feedback ensures simple spectrum and controllability through the curvature input. & \ref{sec:existence}\\
Theorem~\ref{thm:kernel} & The kernel and gain exist, with explicit modal formulas and an invertible Fredholm transformation. & \ref{sec:fredholm}\\
Theorem~\ref{thm:lipschitz} & The kernel and gain depend continuously on a compact admissible design class, with uniformly bounded design and stability constants. & \ref{sec:continuity}\\
Theorem~\ref{thm:NO} & One neural operator approximates the gain uniformly to any prescribed accuracy. & \ref{sec:neural-approximation}\\
Theorem~\ref{thm:intro}, Theorem~\ref{thm:nominal} & Exact feedback achieves local exponential stability at any prescribed decay rate. & \ref{sec:target}\\
Theorems~\ref{thm:intro} and~\ref{thm:intro-NO}, Theorem~\ref{thm:main} & Accurate gains preserve local exponential stability, with uniform guarantees for one neural operator across the class. & \ref{sec:approx-stability}\\
\bottomrule
\end{tabularx}
\end{table}

\subsection{Notation}\label{ssec:notation}
$\LL$ is the real Lebesgue space, complexified for spectral arguments,
with norm $\|\cdot\|$ and inner product
$\ip{f}{g}=\int_0^1 f\bar g\dd x$, linear in $f$.
The bracket also denotes the induced duality pairing. States and implemented
gains are real, with feedback functionals extended complex linearly.
$H^s(0,1)$ and $W^{k,\infty}(0,1)$ are Sobolev spaces.

For Banach spaces $X,Y$, $\mathcal L(X,Y)$ denotes bounded linear
operators, with $\mathcal L(X)=\mathcal L(X,X)$ and space subscripts on
norms as needed. $D(A)$, $\sigma(A)$, and $A^*$ denote domain, spectrum,
and adjoint, with the adjoint taken in $L^2$ unless specified otherwise.
We use $\Z^+=\{1,2,\ldots\}$ and the square-summable sequence space
$\ell^2$. The weighted space $\ell^2_s$, $s\in\R$, has norm
$\|z\|_{\ell^2_s}^2=\sum_{n\ge1}n^{2s}|z_n|^2$.
$\operatorname{Re}z$, $\bar z$, $\delta_{mn}$, and $\delta$ denote real
part, complex conjugate, Kronecker symbol, and Dirac distribution.
Positive constants $c,C$ may vary between estimates and are independent
of the estimated indices and variables. The notation $C(B,d)$ indicates
dependence only on $B,d$. Other permitted parameter dependencies are stated locally.

For an admissible design, let $\{\chi_n\},\{\tilde\chi_n\}$ be the
biorthogonal Riesz bases of Theorem~\ref{thm:generic}(ii) and write
$v_n=\ip{v}{\tilde\chi_n}$. The modal inner product and spectral norms are
\begin{align}
\hp{v}{w}:=\sum_{n\ge1}v_n\overline{w_n},\qquad
\|w\|_{\mathcal H^s}^2:=\sum_{n\ge1}n^{2s}|w_n|^2.
\label{eq:modal-inner-product}
\end{align}
For $s\ge0$, $\mathcal H^s$ consists of the functions in $\LL$ with finite
spectral norm. For $s<0$, it is the completion of finite sums of $\chi_n$
in that norm. We write $\mathcal H=\mathcal H^0$, which is $\LL$ with
the modal inner product, and distinguish these spectral spaces from the
Sobolev spaces $H^s(0,1)$.

\section{Two inputs: necessary and sufficient}\label{sec:control}

We will begin by establishing that two inputs are not only sufficient, but necessary for controlling \eqref{eq:plant}. We first establish that its operator is self adjoint, and then present the spectrum of \eqref{eq:plant} which leads to the controllability result. 

\subsection{Plant operator}

Fix a real coefficient $\lambda\in W^{1,\infty}(0,1)$. Define the plant
operator $\Aop$ by
\begin{subequations}\label{eq:A}
\begin{align}
\Aop\phi:=&\, -\phi''''-(\lambda\phi')'\,, \label{eq:abstract-A}
\end{align}
with domain
\begin{align}
    D(\Aop):=&\, \{\phi\in H^4(0,1):\ \phi(0)=\phi(1)=0,\ \phi''(0)=\phi''(1)=0\}. \label{eq:A-dom}
\end{align}
\end{subequations}
For any test function
$\phi\in D(\Aop)$, the plant \eqref{eq:plant} can be written in the abstract form
\begin{subequations}\label{eq:abstract}
\begin{align}
\frac{d}{dt}\ip{u}{\phi}&=\ip{u}{\Aop\phi}+U_1\ip{B_1}{\phi}+U_2\ip{B_2}{\phi}
   -\ip{uu_x}{\phi}\,,\label{eq:abstract-plant}\\
\ip{B_1}{\phi}&=-\phi'''(0)-\lambda(0)\phi'(0)\,,\label{eq:b1}\\
\ip{B_2}{\phi}&=-\phi'(0)\,.\label{eq:b2}
\end{align}
\end{subequations}

First, we will show that $\Aop$ is self-adjoint with compact resolvent 
so it has a real
orthonormal basis of eigenfunctions $\{\varphi_n\}$, with eigenvalues
$\mu_1\ge\mu_2\ge\cdots\to-\infty$ repeated according to multiplicity. For
constant $\lambda$ the eigenpairs are known in closed form
(see \cite{CoronLu2015}):
\begin{align}
\varphi_n(x)=\sqrt2\sin(n\pi x),\qquad \mu_n=-n^4\pi^4+\lambda n^2\pi^2 .\label{eq:sine}
\end{align}
Let $\Aop$ be split into its spatial derivative and anti-diffusion components as
\begin{align}
\Aop=A_0+P,\qquad A_0\phi:=-\phi'''',\qquad P\phi:=-(\lambda\phi')'\,.\label{eq:A0P}
\end{align}
where we ensure $A_0$ has the
sine basis: $A_0\big(\sqrt2\sin n\pi x\big)=-n^4\pi^4\sqrt2\sin n\pi x$. 
Moreover, by integration by parts, one can show the perturbation $P$ is symmetric on $D(\Aop)$ and that $P$ is small relative to $A_0$:
\begin{align}
\|\phi'\|^2\le\|\phi''\|\,\|\phi\|,\qquad
\|\phi''\|^2\le\|\phi''''\|\,\|\phi\|=\|A_0\phi\|\,\|\phi\|.
\label{eq:phi-bounds}
\end{align}
These bounds and Young's inequality \cite[App.~B.2]{Evans} give, for every $\epsilon>0$,
\begin{align}
\|P\phi\|\le\epsilon\|A_0\phi\|+C_\epsilon\|\phi\|,\quad
C_\epsilon:=\frac{M^2}{4\epsilon}+\frac{M}{2},\quad
M:=\|\lambda\|_{L^\infty}+\|\lambda'\|_{L^\infty}.
\label{eq:relbound}
\end{align}
Thus the Kato--Rellich Theorem \cite[Thm.~V.4.3]{Kato} guarantees that $\Aop$
is self-adjoint on $D(\Aop)$. Its spectrum is bounded above and its resolvent is compact because
$D(\Aop)\hookrightarrow\LL$ compactly. 

To showcase the effect of the inputs, we compute the following identity via integration by parts
\begin{lemma}(Integration by parts identity) \label{lem:ibp}
Let $\psi\in H^4(0,1)$ satisfy $\psi(1)=\psi''(1)=0$ and let $\phi\in D(\Aop)$. Define 
\begin{align}
\qd\phi:=\phi'''+\lambda\phi'\,.\label{eq:quasi}
\end{align}
The following boundary identity holds:
\begin{align}\label{eq:ibp}
\int_0^1\big(\psi''''+(\lambda\psi')'\big)\phi\dd x
={}&\psi''(0)\,\phi'(0)+\psi(0)\,\qd\phi(0)
-\int_0^1\psi\,\Aop\phi\dd x .
\end{align}
\end{lemma}
 Applied with $\psi=u(\cdot,t)$, \eqref{eq:ibp} produces exactly the two boundary
terms of \eqref{eq:b1}--\eqref{eq:b2}. 

Lastly, throughout thsi work, we will use the modal state and input coefficients repeatedly. Hence, we define the shorthand notation:
\begin{align}
u_n:=\ip{u(\cdot,t)}{\varphi_n},\quad
\alpha_n:=-\ip{B_1}{\varphi_n}=\qd{\varphi_n}(0),\quad
\beta_n:=-\ip{B_2}{\varphi_n}=\varphi_n'(0).
\label{eq:modal-coefficients}
\end{align}
Then \eqref{eq:abstract} gives
$\dot u_n=\mu_n u_n-\alpha_n U_1-\beta_n U_2-\ip{uu_x}{\varphi_n}$.

\subsection{Plant operator has finitely many unstable eigenvalues}\label{ssec:spectrum-plant}

In this section, we showcase that the corresponding eigenvalue structure associated with the 
operator $A_\lambda$ yields only \emph{finitely} many unstable eigenvalues which are widely separated. 
To do so, we consider the upper and lower bounds of $\lambda$ and invoke the min-max principle. 

First, consider a constant $\bar{\lambda} = \|\lambda \|_{L^\infty}$ and correspondingly define the operator
$A_{\bar{\lambda}} = -\phi'''' - \bar{\lambda}\phi''$ for $\phi \in D(A_\lambda)$. Then, since the eigenpairs are 
known in closed form for constant $\lambda$ as in \eqref{eq:sine}, we have that 
\begin{align}
\ip{A_{-\bar{ \lambda}}\phi}{\phi}\le\ip{\Aop\phi}{\phi}\le\ip{A_{ \bar \lambda}\phi}{\phi}
\qquad\text{for all }\phi\in D(\Aop).\label{eq:sandwich}
\end{align}
The operators $A\in\{A_{-\bar\lambda},A_\lambda,A_{\bar\lambda}\}$ are
self-adjoint and bounded above, with compact resolvent on $D(\Aop)$, so
their decreasingly ordered eigenvalues satisfy the min--max formula
\begin{align}
\mu_n(A)=\max_{\substack{S\subset D(\Aop)\\ \dim S=n}}\ \
\min_{\substack{\phi \in S\\ \phi \neq 0 }}\ \frac{\ip{A\phi}{\phi}}{\|\phi\|^2} .\label{eq:minmax}
\end{align}
Applying \eqref{eq:minmax} to both the span of the first $n$ sine modes and the
orthogonal complement of the first $n-1$ modes, along with the fact that
$-k^4\pi^4+\bar\lambda k^2\pi^2$ decreases for
$k^2\ge\bar\lambda/(2\pi^2)$, gives
\begin{align}
\big|\mu_n+n^4\pi^4\big|\le\|\lambda\|_{L^\infty}\,n^2\pi^2,\qquad n\ge\max\left\{1,\left\lceil\sqrt{\frac{\|\lambda\|_{L^\infty}}{2\pi^2}}\right\rceil\right\}.\label{eq:mu-asymp}
\end{align}
Two important caveats are worth mentioning regarding \eqref{eq:mu-asymp}. First, $n^4$ grows faster than $n^2$, so only
finitely many $\mu_n$ are nonnegative. This implies that, no matter how large the anti-diffusion, the
instability is finite dimensional. Second, the gaps $|\mu_m-\mu_n|$ grow like
$n^3$, so high modes are widely separated and only finitely many eigenvalues
can be repeated. Now, we show that the boundary control with a single input, cannot
compensate for eigenvalues with multiplicity two, and hence a second input is not only sufficient, but
necessary. 

\subsection{Two inputs are necessary and sufficient}

For each distinct eigenvalue $\mu$ of $\Aop$, define its \emph{eigenspace}
\begin{align}
E_\mu:=\{\phi\in D(\Aop):\ \Aop\phi=\mu\phi\}.
\end{align}
Additionally, let $P_\mu$ denote the orthogonal projection in $\LL$ onto $E_\mu$:
\begin{align}
P_\mu v=\sum_{n:\,\mu_n=\mu}\ip{v}{\varphi_n}\varphi_n,
\qquad v\in\LL.
\end{align}
For the linearization of \eqref{eq:plant}, an unstable double eigenvalue
requires two inputs for stabilization. With zero initial state and
$(U_1,U_2)\in C_c^\infty((0,T_c);\R^2)$, variation of constants gives
\begin{align}
u(T_c)=\int_0^{T_c}e^{(T_c-s)\Aop}\big(B_1U_1(s)+B_2U_2(s)\big)\dd s .\label{eq:reachable}
\end{align}

\begin{theorem}[Two inputs suffice for every $\lambda$]\label{thm:control}
Let $\lambda\in W^{1,\infty}(0,1)$ and $T_c>0$. For the linearization of \eqref{eq:plant} with the two boundary inputs $U_1=u(0,\cdot)$ and $U_2=u_{xx}(0,\cdot)$, started from $u(\cdot,0)=0$, the set of terminal states \eqref{eq:reachable} reached by controls $(U_1,U_2)\in C_c^\infty((0,T_c);\R^2)$ is dense in $\LL$.
\end{theorem}

\begin{proof}
By duality, it suffices to show that an adjoint solution with both boundary observations identically zero must have zero initial state.
Let $\mathcal R(T_c)$ be the linear space of terminal states in
\eqref{eq:reachable}. Take $z_0\in\mathcal R(T_c)^\perp$ and set
$z(t)=e^{t\Aop}z_0$. Self-adjointness and \eqref{eq:reachable} give
\begin{align}
0=\ip{u(T_c)}{z_0}
&=\;\int_0^{T_c}\big[U_1(s)\ip{B_1}{z(T_c-s)}
+U_2(s)\ip{B_2}{z(T_c-s)}\big]\dd s.\label{eq:boundary-linear-proof}
\end{align}
For $t>0$, parabolic smoothing gives $z(t)\in D(\Aop)$, so both boundary
observations are well defined and continuous. Since the controls are
arbitrary, \eqref{eq:b1}--\eqref{eq:b2} imply
\begin{align}
z_x(0,t)&=\;\qd z(0,t)=0,\qquad 0<t<T_c.\label{eq:obs}
\end{align}

To recover the traces of each spectral projection from \eqref{eq:obs},
we first bound their growth with the eigenvalue. The trace inequality
\cite[Sec.~5.5, Thm.~1]{Evans} and the elliptic estimate give,
for $\phi\in D(\Aop)$,
\begin{align}
|\phi'(0)|+|\qd\phi(0)|&\le\;C_\lambda\|\phi\|_{H^4},\label{eq:control-trace}\\
\|\phi\|_{H^4}&\le\;C_\lambda\big(\|\Aop\phi\|+\|\phi\|\big),
\label{eq:control-elliptic}
\end{align}
with $C_\lambda$ depending only on $\|\lambda\|_{W^{1,\infty}}$.
For the second estimate, the sine expansion gives
$\|\phi\|_{H^4}\le C\|A_0\phi\|$. Writing $A_0=\Aop-P$ and
using \eqref{eq:relbound} with $\epsilon=1/2$ absorbs the perturbation
term and proves \eqref{eq:control-elliptic}.
Applying these estimates to $P_\mu z_0$ yields, after enlarging $C_\lambda$,
\begin{align}
|(P_\mu z_0)'(0)|+|\qd{(P_\mu z_0)}(0)|&\le\;C_\lambda(1+|\mu|)\|P_\mu z_0\|.\label{eq:app-b-display-016}
\end{align}

Now expand $z(t)=\sum_\mu e^{\mu t}P_\mu z_0$ over the distinct
eigenvalues. Repeated eigenvalues are grouped in a single projection,
so no simplicity assumption is needed. Equation~\eqref{eq:obs} becomes
\begin{align}
\sum_\mu e^{\mu t}(P_\mu z_0)'(0)=0,\qquad
\sum_\mu e^{\mu t}\qd{(P_\mu z_0)}(0)=0,
\quad 0<t<T_c.\label{eq:boundary-series}
\end{align}
The coefficient bound \eqref{eq:app-b-display-016} and quartic spectral
decay \eqref{eq:mu-asymp} give locally uniform convergence for
$\operatorname{Re}t>0$. These series are therefore analytic, and their
vanishing on $(0,T_c)$ extends to all $t>0$.
Dividing by $e^{\mu^\star t}$ for the largest eigenvalue $\mu^\star$
and letting $t\to\infty$ isolates the two traces of $P_{\mu^\star}z_0$.
Dominated convergence shows that both vanish. Removing this term and
repeating gives $(P_\mu z_0)'(0)=\qd{(P_\mu z_0)}(0)=0$ for every $\mu$.

Each $\phi=P_\mu z_0$ thus has four zero Cauchy data at the left endpoint:
the hinged conditions give $\phi(0)=\phi''(0)=0$, and the two observations
give $\phi'(0)=\qd\phi(0)=0$.
Writing $y=(\phi,\phi',\phi'',\qd\phi)^{\mathsf T}$ turns its
eigenvalue equation into
\begin{align}
y_1'&=\;y_2,\qquad y_2'=y_3,\qquad
y_3'=y_4-\lambda y_2,\qquad y_4'=-\mu y_1,
\qquad y(0)=0.\label{eq:app-b-display-002}
\end{align}
Uniqueness for this linear system with $L^\infty$ coefficients gives
$y\equiv0$. Hence every spectral projection $P_\mu z_0$ vanishes,
so $z_0=0$ and $\overline{\mathcal R(T_c)}=\LL$.
\end{proof}

\begin{proposition}[One input does not suffice]\label{prop:obstruction}
Let $\lambda\in W^{1,\infty}(0,1)$ and let $\mu$ be a double eigenvalue of
$\Aop$. Consider the linearization of \eqref{eq:plant} actuated through the curvature input only,
$U_1\equiv0$, with any linear feedback $U_2=F(u)$ where $F$ is defined on a subspace of
$\LL$ containing $E_\mu$. Then:
\begin{enumerate}[label=(\roman*)]
\item $\mu$ is an eigenvalue of the closed loop: there is
$\psi\in E_\mu\setminus\{0\}$ with $u(t)=e^{\mu t}\psi$ solving the closed-loop
equations. Thus, no such feedback yields a decay rate larger than $-\mu$, and
none is stabilizing if $\mu\ge0$. The same holds with $U_1$ and $U_2$
exchanged.
\item For constant $\lambda$, $\Aop$ has a double eigenvalue exactly when
$\lambda=(i^2+j^2)\pi^2\in\mathscr N_\lambda$, and the corresponding double
eigenvalue is $\mu=i^2j^2\pi^4>0$. For every $\lambda\in\mathscr N_\lambda$ the
curvature input alone cannot stabilize the linearized plant.
\end{enumerate}
\end{proposition}

\begin{proof}
We first prove (i). Since $\dim E_\mu=2$ and $F$ is scalar, choose
$0\ne\psi\in E_\mu\cap\ker F$. Its hinged boundary conditions and
$F(\psi)=0$ satisfy the feedback boundary conditions, so
$u(t)=e^{\mu t}\psi$ remains a closed-loop solution. Thus no decay rate
exceeding $-\mu$ is possible, and $\mu\ge0$ precludes exponential
stability. The argument similarly applies with the inputs interchanged,
completing (i).

For (ii), take constant $\lambda$. Equality of the eigenvalues of distinct sine
modes $i\ne j$ gives $\lambda=(i^2+j^2)\pi^2$ and the positive repeated
eigenvalue $\mu=i^2j^2\pi^4$. Applying (i) proves (ii).
\end{proof}

Give two inputs are needed, we now showcase how to design the feedback law.

\section{Controller design}\label{sec:kernel}

Figure~\ref{fig:design-overview} summarizes the two-stage design. The displacement feedback's job is to enable controllability by splitting repeated eigenvalues
into simple ones while
the curvature input applies the one-input design of  \cite{CoronLu2015}, in the abstract form of \cite{HayatLoko2026}, to
the pre-compensated plant. In finite dimensions, this is Heymann's Lemma
\cite{Heymann1968,Wonham1967}.

\begin{figure}[htbp]
\centering
\resizebox{\textwidth}{!}{\input{diagrams/part1-overview}}
\caption{Two feedbacks with distinct roles. (a) Displacement pre-feedback
(blue) makes every mode controllable from the curvature input. The
Fredholm design then constructs the stabilizing gain $h$ and transformation
$T=I-K$. (b) Both feedback laws act simultaneously on the same plant state,
with $g_\theta(y)=e^{-\theta y}$. The shift $a$ sets the achievable decay rate.}
\label{fig:design-overview}
\end{figure}
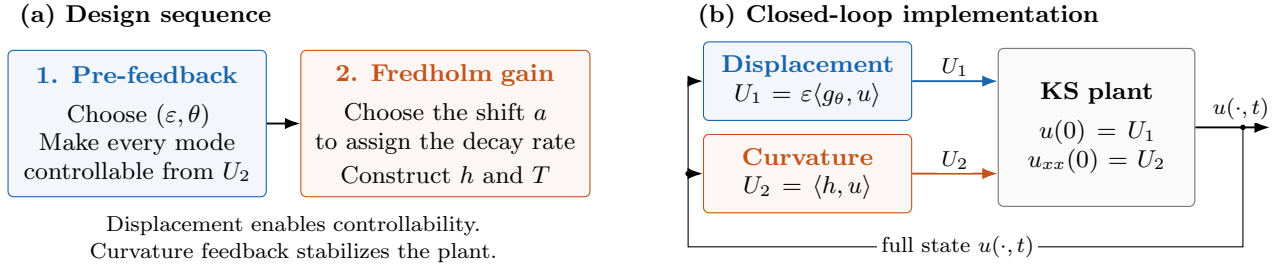

\subsection{Pre-feedback and pre-compensated operator}
We start with the design of $U_1$ which we call the pre-feedback as it enables stabilization,
and does nothing more. The displacement input 
is given by 
\begin{align}
U_1(t)=\varepsilon\ip{g_\theta}{u(\cdot,t)}
=\varepsilon\int_0^1g_\theta(y)u(y,t)\dd y,\qquad
 g_\theta(y):=e^{-\theta y},\qquad y\in[0,1],\label{eq:U1}
\end{align}
where $\theta>0$ and $\varepsilon\neq0$ control the spatial weighting and magnitude of the displacement pre-feedback, respectively.
In particular, for $\theta$ outside a discrete set and all small
$\varepsilon\ne0$ it splits every double eigenvalue of $\Aop$ and makes every
resulting mode reachable from the curvature input, while leaving the
high-frequency structure of $\Aop$ untouched. With the input $U_1$ applied, the \emph{pre-compensated
operator} then becomes
\begin{align}
D(\Ap)=\{\phi\in H^4:\ \phi(0)=\varepsilon\ip{g_\theta}{\phi},\
\phi''(0)=\phi(1)=\phi''(1)=0\},\quad
\Ap\phi=-\phi''''-(\lambda\phi')'.\label{eq:Aprime}
\end{align}
where $\Ap=\Aop+\varepsilon B_1\ip{g_\theta}{\cdot}$, plays the role of
$A+b_1k_1^{\mathsf T}$ in Heymann's Lemma. For $\varepsilon\ne0$, $\Ap$ is not self-adjoint. 
Nevertheless, under the parameter conditions established below, its eigenvectors form a Riesz basis of $\LL$ (cf.~Theorem~\ref{thm:generic}(ii)), allowing stabilization to be shown.

\subsection{Transformation, target, kernel equations}
We now give the stabilizing feedback $U_2$. To do so, we begin by introducing the Fredholm transformation:
\begin{align}
w(x,t)&=u(x,t)-\int_0^1k(x,y)u(y,t)\dd y=:(Tu)(x,t),\qquad T:=I-K,\label{eq:transform}\\
(Ku)(x)&:=\int_0^1k(x,y)u(y)\dd y,\label{eq:kernel-operator}
\end{align}
and the curvature feedback
\begin{align}
U_2(t)=\ip{h}{u(\cdot,t)}=\int_0^1h(y)u(y,t)\dd y .\label{eq:U2}
\end{align}
where $(T, h)$ are to be designed. The kernel $k(x,y)$ defines the transformation throughout the spatial domain,
whereas $h(y)$ is the gain used at the controlled boundary. For smooth $k$,
these are related by $
h(y)=k_{xx}(0,y)$.
However, $k$ is not necessarily globally smooth as we will see below. 

Given this design, consider the following target system with parameter $a > 0$ 
\begin{subequations}\label{eq:target}
\begin{align}
w_t+w_{xxxx}+(\lambda w_x)_x+a\,w&=-(I-K)(uu_x),\label{eq:target-a}\\
w(0,t)&=\varepsilon\ip{g_\theta}{w},\label{eq:target-b}\\
w_{xx}(0,t)=w(1,t)=w_{xx}(1,t)&=0,\label{eq:target-c}
\end{align}
\end{subequations}

To ensure the corresponding Fredholm transform \eqref{eq:transform} maps the plant \eqref{eq:plant} into the target system \eqref{eq:target}, we design $(k,h)$ through the operator conditions
\begin{align}
T\big(\Ap+B_2\ip{h}{\cdot}\big)=(\Ap-a)T,\qquad TB_2=B_2,
\label{eq:Fequiv-intro}
\end{align}
where $T=I-K$ and $K$ has kernel $k$.
To obtain the first identity, substitute the feedback \eqref{eq:U2}
into the pre-compensated plant and apply the time-independent transformation
$w=Tu$:
\begin{align}
w_t&=T\big(\Ap+B_2\ip{h}{\cdot}\big)u-T(uu_x).
\end{align}
The target equation requires $w_t=(\Ap-a)Tu-T(uu_x)$, so matching
the linear terms gives the first identity in \eqref{eq:Fequiv-intro}.
We additionally impose the normalization $TB_2=B_2$, which preserves
the curvature input direction.
The second identity is a design condition, rather than a consequence of
the first. Integrating \eqref{eq:Fequiv-intro} by parts where traces exist gives the following kernel equations.

\begin{proposition}[Kernel equations map plant to target]\label{prop:kernel}
Let $k\in C^4([0,1]^2\setminus\{x=y\})$ be such that $k$ and its derivatives up to order three extend continuously to each closed triangle $\{y\le x\}$, $\{y\ge x\}$ (so that $k_{yyy}(x,0)$ and $k(0,y)$, $k_{xx}(0,y)$ are well defined), and suppose $k$ solves
\begin{subequations}\label{eq:kpde}
\begin{align}
k_{xxxx}+(\lambda(x)k_x)_x-k_{yyyy}-(\lambda(y)k_y)_y+a\,k
 &=a\,\delta(x-y)+\varepsilon\,k_{yyy}(x,0)\,g_\theta(y),\label{eq:kpde-a}\\
k(x,0)&=k(x,1)=k_{yy}(x,0)=k_{yy}(x,1)=0,\label{eq:kpde-b}\\
k_y(x,0)&=0,\label{eq:kpde-c}\\
k(1,y)&=k_{xx}(1,y)=0,\label{eq:kpde-d}\\
k(0,y)&=\varepsilon\int_0^1k(x,y)\,g_\theta(x)\dd x,\label{eq:kpde-e}
\end{align}
\end{subequations}
where \eqref{eq:kpde-a} is understood in the sense of distributions on $(0,1)^2$, \eqref{eq:kpde-b}--\eqref{eq:kpde-c} hold for $x\in(0,1)$, and \eqref{eq:kpde-d}--\eqref{eq:kpde-e} hold for $y\in(0,1)$. Then, for every sufficiently smooth solution $u$ of \eqref{eq:plant} with the feedback \eqref{eq:U1} and $U_2=\int_0^1k_{xx}(0,y)u(y,t)\dd y$, the function $w$ defined by \eqref{eq:transform} satisfies the target system \eqref{eq:target}.
\end{proposition}

The proof is given in Appendix~\ref{app:kernel-transformation-proof}.

\subsection{Modal form of the kernel equations}

Given the representation of the kernel and gain, it will be much easier to work in modal coordinates specified in this section. Recall that $\{\varphi_n\}$ is the orthonormal eigenbasis of $\Aop$,
with eigenvalues $\mu_n$.

\begin{proposition}[Modal form of kernel equations]\label{prop:modal}
Suppose $k$ satisfies the hypotheses of Proposition~\ref{prop:kernel},
and set $h(y):=k_{xx}(0,y)$. Define the modal coefficients by
\begin{equation}
h_j:=\ip{h}{\varphi_j},\qquad
g_j:=\ip{g_\theta}{\varphi_j},\qquad
K_{mj}:=\int_0^1\int_0^1 k(x,y)\varphi_m(x)\varphi_j(y)\dd y\dd x.
\label{eq:kernel-modal-coefficients}
\end{equation}
Then these coefficients satisfy, for all $m,j\ge1$,
\begin{subequations}\label{eq:modal}
\begin{align}
(\mu_j-\mu_m+a)K_{mj}+\beta_m h_j
  +\varepsilon\,\alpha_m\textstyle\sum_{p}g_pK_{pj}
 -\varepsilon\,g_j\textstyle\sum_{i}K_{mi}\alpha_i&=a\,\delta_{mj},
   \label{eq:modal-a}\\
\textstyle\sum_{j}K_{mj}\,\beta_j&=0 .\label{eq:modal-b}
\end{align}
\end{subequations}
\end{proposition}

The proof is given in Appendix~\ref{app:modal-equivalence-proof}. The condition \eqref{eq:modal-b} is the normalization $TB_2=B_2$. It makes
\eqref{eq:modal} linear in $(K,h)$ which is needed for the analysis. For $\varepsilon=0$, this is the system
behind \cite[(2.53)--(2.55)]{CoronLu2015}.

Now, the PDE mapping the original plant to the target system is established, but not necessarily well-posed. The next section gives 
conditions on $a, \theta, \varepsilon$ such that this occurs.

\section{Pre-feedback enables controllability}\label{sec:existence}

\subsection{Admissible design and pre-feedback}
We specify the spectral properties that the displacement pre-feedback
must provide for the curvature input to implement the Fredholm design.
We then prove that suitable parameters exist for every plant coefficient
$\lambda\in W^{1,\infty}(0,1)$.
Let $\mu_n'$ denote the eigenvalues of the pre-compensated operator $\Ap$.
For each algebraically simple eigenvalue, let $\chi_n$ and $\tilde\chi_n$
be corresponding eigenvectors of $\Ap$ and its adjoint $\Ap{}^*$:
\begin{align}
\Ap\chi_n=\mu_n'\chi_n,\qquad
\Ap{}^*\tilde\chi_n=\overline{\mu_n'}\,\tilde\chi_n,\qquad
\ip{\chi_n}{\tilde\chi_n}=1.\label{eq:adjoint-eigenvector-normalization}
\end{align}
We use $\chi_n,\tilde\chi_n$ for the eigenvectors associated with mode $n$,
and $\chi_z,\tilde\chi_z$ for those of $\Ap$ and its adjoint associated
with $z,\bar z$. Since $\Ap$ has real coefficients, nonreal eigenvalues
occur in complex-conjugate pairs; we take the eigenvectors real when the
eigenvalue is real and conjugate-symmetric otherwise, so that
$\chi_{\bar z}=\overline{\chi_z}$ and $\tilde\chi_{\bar z}=\overline{\tilde\chi_z}$.
The design conditions are as follows.

\begin{definition}[Admissible design]\label{def:admissible}
Fix $\lambda\in W^{1,\infty}(0,1)$. A triple
$(a,\varepsilon,\theta)\in(0,\infty)\times\R\times(0,\infty)$ is
\emph{admissible for} $\lambda$ if the pre-compensated operator $\Ap$
in \eqref{eq:Aprime}, with eigenvalues $\mu_n'$, satisfies:
\begin{enumerate}[label=(H\arabic*)]
\item every eigenvalue of $\Ap$, and hence of the plant as modified by $(\varepsilon,\theta)$, is algebraically simple;
\item The curvature input acts on every mode, again a property of
$(\varepsilon,\theta)$, meaning
\begin{align}
b_n':=\ip{B_2}{\tilde\chi_n}=-\overline{\tilde\chi_n'(0)}\ne0,
\qquad n\ge1.
\end{align}
\item The target eigenvalues have negative real parts, and the target
shift avoids spectral resonances:
\begin{align}
a>\max_n\operatorname{Re}\mu_n',\qquad
a\notin\mathscr N'_a:=\{\mu_i'-\mu_j':\ i,j\ge1\}.
\end{align}
\end{enumerate}
We write $\mathscr U_\lambda$ for the admissible triples and
$\mathscr U:=\{(\lambda,a,\varepsilon,\theta):(a,\varepsilon,\theta)\in\mathscr U_\lambda\}$.
\end{definition}

For an eigenvector labeled by its eigenvalue $z$, we similarly write
$b'_z:=\ip{B_2}{\tilde\chi_z}$. Thus $b'_n$ uses a mode index,
whereas $b'_z$ uses the eigenvalue itself. Both coefficients refer to the
chosen scaling of the adjoint eigenvector.

In Definition~\ref{def:admissible}, (H1) makes (H2) the Fattorini--Hautus test for $(\Ap,B_2)$:
no adjoint eigenmode is orthogonal to the curvature input direction.
Conditions (H1)--(H2) of Definition~\ref{def:admissible} describe what the displacement pre-feedback must
achieve before the curvature feedback is designed. Condition (H3) of Definition~\ref{def:admissible} then
selects the target shift $a$. These conditions restrict the design
parameters, not the given coefficient $\lambda$. The next theorem shows
that admissible choices exist for every $\lambda\in W^{1,\infty}(0,1)$.

\begin{theorem}[Pre-feedback enables controllability for every $\lambda$]\label{thm:generic}\leavevmode
\begin{enumerate}[label=(\roman*)]
\item Let $\lambda\in W^{1,\infty}(0,1)$, and let $\Ap$ be the operator \eqref{eq:Aprime} of the plant under the controllability-ensuring pre-feedback \eqref{eq:U1}. There is a discrete set $\mathscr T_\lambda\subset(0,\infty)$ such that for every $\theta\notin\mathscr T_\lambda$ there is $\varepsilon_1(\lambda,\theta)>0$ with the following property: for every $\varepsilon$ with $0<|\varepsilon|<\varepsilon_1$, (H1) and (H2) of Definition~\ref{def:admissible} hold, the set $\mathscr N'_a$ is discrete, and $(a,\varepsilon,\theta)\in\mathscr U_\lambda$ for every $a>\max\{0,\max_n\operatorname{Re}\mu_n'\}$ outside $\mathscr N'_a$. Moreover $\max_n\operatorname{Re}\mu_n'=\max_n\mu_n+O(|\varepsilon|)$.
\item For every admissible design, the eigenvectors $\{\chi_n\}$ of $\Ap$ form a Riesz basis of $\LL$ with biorthogonal family the eigenvectors $\{\tilde\chi_n\}$ of $\Ap{}^*$, $\Ap$ generates a $C^0$ semigroup on $\LL$, and, indexing so that $\operatorname{Re}\mu_1'\ge\operatorname{Re}\mu_2'\ge\cdots$, there are $c,C>0$ such that, for all $n,p\ge1$,
\begin{equation}\label{eq:HL-hyp}
c\,n^4\le |\mu_n'|+1\le C n^4,\qquad
|\mu_n'-\mu_p'|\ge c\,n^3|n-p|,\qquad
c\,n\le |b_n'|\le C\,n.
\end{equation}
\end{enumerate}
\end{theorem}

\subsection{Proof of Theorem~\ref{thm:generic}: pre-feedback enables controllability}\label{ssec:prefeedback-proof}

The proof separates the high modes, which already have the required
properties, from finitely many modes that may need pre-feedback.
At high frequencies, the eigenvalues are simple and the curvature input
reaches every mode. We show that the pre-feedback preserves these
properties because its effect is smaller than the spectral gaps.
At a double plant eigenvalue, the pre-feedback retains one eigenvalue
and moves the other. Suitable conditions on the profile make both
eigenvalues simple and both modes reachable. These conditions also
ensure reachability of previously missed simple modes, while modes
already reached remain so.

We first establish the plant estimates in Lemma~\ref{lem:plant-spectral}.
Lemma~\ref{lem:spectrum} then gives
the pre-feedback resolvent, spectral characterization, and curvature-input
coefficients. Lemma~\ref{lem:high} uses these formulas to treat the high
modes. For the remaining modes, we derive finitely many nonvanishing
conditions on the profile and choose $g_\theta$ to satisfy them.
A common small nonzero amplitude $\varepsilon$ and an admissible shift
$a$ then complete part (i). Finally, we prove completeness and use the
high-mode estimates to establish the Riesz basis, semigroup generation,
and spectral bounds in part (ii).

Throughout this proof, $\lambda\in W^{1,\infty}(0,1)$ is
fixed, $g\in\LL$ is a real pre-feedback profile, and $\varepsilon\in\R$.
Recall that $\alpha_n=\varphi_n'''(0)+\lambda(0)\varphi_n'(0)$ and
$\beta_n=\varphi_n'(0)$ are the boundary-trace coefficients associated
with the displacement and curvature inputs, respectively, as defined in
\eqref{eq:modal-coefficients}.
We use $\Ap$ from \eqref{eq:Aprime} with $g_\theta$ replaced by $g$, and write
\begin{align}
g_n&:=\;\ip{g}{\varphi_n},\qquad
\mathfrak g_n:=\min_{m\ne n}|\mu_m-\mu_n|.
\label{eq:high-mode-notation}
\end{align}
Unless stated otherwise, the generic constants in this
subsection may depend on $\lambda$, $g$, $\varepsilon$, and the target shift
$a$ where it is used, but not on the mode indices or spectral variables.

\subsubsection{Spectral properties of the original plant}\label{sssec:plant-spectral}
We begin with the original plant. Its eigenspaces have dimension at most
two, and its high-frequency eigenfunctions approach the sine basis.
The proof also establishes the spectral gap used to control the effect
of pre-feedback on the high modes.

\begin{lemma}[Spectral properties of original operator $\Aop$]\label{lem:plant-spectral}
Let $\lambda\in W^{1,\infty}(0,1)$, set $N_1:=\max\{3,\lceil2\|\lambda\|_{L^\infty}/\pi^2\rceil\}$ and $s_n(x):=\sqrt2\sin(n\pi x)$, and normalize $\varphi_n$ by $\ip{\varphi_n}{s_n}>0$ for $n\ge N_1$. Then, the following hold:
\begin{enumerate}[label=(\roman*)]
\item For every eigenvalue $\mu$ of $\Aop$, the eigenspace $E_\mu$ has dimension at most two.
\item For $\phi \in E_\mu$, the linear map $\phi\mapsto\big(\phi'(0),\qd\phi(0)\big)\in\R^2$ is injective.
In particular, for a double eigenvalue with eigenfunctions $\varphi_j,\varphi_k$ forming a basis of $E_\mu$, the matrix
$\begin{pmatrix}\alpha_j&\alpha_k\\\beta_j&\beta_k\end{pmatrix}$ is nonsingular.
\item With constants depending only on $\|\lambda\|_{W^{1,\infty}}$,
\begin{gather}
\|\varphi_n-s_n\|_{H^s(0,1)}\le C\,n^{s-1}\quad(0\le s\le4),\qquad n\ge N_1,\\
\|\varphi_n^{(j)}-s_n^{(j)}\|_{L^\infty}\le C\,n^{j-1/2}\quad(j=0,1,2,3),\qquad n\ge N_1 ,
\end{gather}
and consequently $\beta_n=\sqrt2\,n\pi+O(n^{1/2})$, $\alpha_n=-\sqrt2(n\pi)^3+O(n^{5/2})$. In particular, after enlarging $N_1$, there are $c_0,C_0>0$ and $N_1$, depending only on $\|\lambda\|_{W^{1,\infty}}$, such that
\begin{equation}\label{eq:trace-bounds}
|\beta_n|\le C_0 n,\quad |\alpha_n|\le C_0n^3\ \ (n\ge1),\qquad c_0n\le|\beta_n|,\quad c_0n^3\le|\alpha_n|\ \ (n\ge N_1).
\end{equation}
\end{enumerate}
\end{lemma}

\begin{proof}
We prove (i) and (ii) together. For $\phi\in E_\mu$, the hinged conditions give
$\phi(0)=\phi''(0)=0$. The ODE uniqueness argument in the proof of
Theorem~\ref{thm:control} shows that the remaining Cauchy data
$\phi'(0)$ and $\qd\phi(0)$ uniquely determine $\phi$. Hence
$\dim E_\mu\le2$, proving (i), and the trace map is injective, proving (ii).

We next prove (iii). Here $C>0$ may vary but depends only on $\|\lambda\|_{W^{1,\infty}(0,1)}$, not on mode indices.

Use the decomposition $\Aop=A_0+P$ from
\eqref{eq:A0P}. Expanding $P\phi=-(\lambda\phi')'$ and applying
H\"older's inequality gives
\begin{align}
\|P\phi\|&\le\;C\|\lambda\|_{W^{1,\infty}}\|\phi\|_{H^2},\label{eq:app-d-display-003}
\end{align}

We will need to bound the difference between eigenvalues.
For $n\ge N_1$ and $m\ne n$, set $p=\min\{m,n\}$ and
$q=\max\{m,n\}$. The min--max lower bound
$\mu_p\ge-p^4\pi^4-\|\lambda\|_{L^\infty}p^2\pi^2$ holds for every $p$.
Combining it with the upper bound in \eqref{eq:mu-asymp} at $q\ge N_1$
and using the eigenvalue ordering gives
\begin{align}
|\mu_m-\mu_n|
&\ge\;\pi^4|m^4-n^4|-\|\lambda\|_{L^\infty}\pi^2(m^2+n^2).\label{eq:app-d-display-080}
\end{align}
Using $n\ge N_1\ge 2\|\lambda\|_{L^\infty}/\pi^2$, we obtain, for $m\ne n$,
\begin{align}
|m^2-n^2|&\ge2n-1\ge2\|\lambda\|_{L^\infty}/\pi^2.\label{eq:app-d-display-001}
\end{align}
Factoring $|m^4-n^4|=|m^2-n^2|(m^2+n^2)$ therefore gives
\begin{align}
|\mu_m-\mu_n|&\ge\;\tfrac12\pi^4|m^4-n^4|,
\qquad n\ge N_1,\quad m\ne n.\label{eq:gap}
\end{align}
Taking the minimum over $m\ne n$ and using
$n^4-(n-1)^4\ge2n^3$ for $n\ge3$ gives
$\mathfrak g_n\ge\pi^4n^3$. In particular, $\mu_n$ is simple for $n\ge N_1$.

Since $A_0$ has eigenpairs $(-n^4\pi^4,s_n)$,
\eqref{eq:mu-asymp} and \eqref{eq:app-d-display-003} give the residual bound
\begin{align}
\|(\Aop-\mu_n)s_n\|
&\le\;\|Ps_n\|+|\mu_n+n^4\pi^4|
\le Cn^2.\label{eq:app-d-display-004}
\end{align}
Expand $s_n=\sum_k\ip{s_n}{\varphi_k}\varphi_k$ in the orthonormal
plant eigenbasis. Applying $\Aop-\mu_n$ removes its component along
$\varphi_n$, and orthogonality with the spectral gap gives
\begin{align}
(\Aop-\mu_n)s_n&=\;\sum_{k\ne n}(\mu_k-\mu_n)
\ip{s_n}{\varphi_k}\varphi_k,\label{eq:sine-residual-expansion}\\
\|(\Aop-\mu_n)s_n\|^2
&=\;\sum_{k\ne n}|\mu_k-\mu_n|^2|\ip{s_n}{\varphi_k}|^2
\ge\;\mathfrak g_n^2\sum_{k\ne n}|\ip{s_n}{\varphi_k}|^2.
\label{eq:sine-residual-gap}
\end{align}
The sum on the right is the squared norm of the component of $s_n$
orthogonal to $\varphi_n$. Taking square roots, using the residual bound
\eqref{eq:app-d-display-004} and $\mathfrak g_n\ge\pi^4n^3$, gives
\begin{align}
\|s_n-\ip{s_n}{\varphi_n}\varphi_n\|
&\le\;\frac{\|(\Aop-\mu_n)s_n\|}{\mathfrak g_n}
\le\frac{Cn^2}{\pi^4n^3}\le Cn^{-1}.
\label{eq:sine-projection-error}
\end{align}
Since $\|s_n\|=1$, orthogonality also gives
\begin{align}
|\ip{s_n}{\varphi_n}|^2
&=\;1-\|s_n-\ip{s_n}{\varphi_n}\varphi_n\|^2
\ge1-C^2n^{-2}.
\label{eq:app-d-display-006}
\end{align}
The chosen sign of $\varphi_n$ therefore gives
\begin{align}
\|\varphi_n-s_n\|&\le\;Cn^{-1}.\label{eq:app-d-display-007}
\end{align}

To estimate the derivatives, set $v=\varphi_n-s_n$. The eigenvalue equations for $\varphi_n$ and $s_n$ are
\begin{align}
(A_0+P)\varphi_n=\mu_n\varphi_n,\qquad
A_0s_n=-n^4\pi^4s_n.\label{eq:eigenvalue-comparison}
\end{align}
Subtracting the second equation from the first and substituting
$\varphi_n=v+s_n$ gives
\begin{align}
A_0v&=\;\mu_nv+(\mu_n+n^4\pi^4)s_n-P\varphi_n.\label{eq:app-d-display-008}
\end{align}
Estimates \eqref{eq:control-elliptic}, \eqref{eq:mu-asymp}, and \eqref{eq:app-d-display-003}, with interpolation and $\|\varphi_n\|=1$, give
\begin{align}
\|\varphi_n\|_{H^4}
&\le\;C\big(\|\Aop\varphi_n\|+\|\varphi_n\|\big)
=C(1+|\mu_n|)\le Cn^4,\label{eq:eigenfunction-H4-bound}\\
\|\varphi_n\|_{H^2}
&\le\;C\|\varphi_n\|^{1/2}\|\varphi_n\|_{H^4}^{1/2}
\le Cn^2,\label{eq:eigenfunction-H2-bound}\\
\|P\varphi_n\|
&\le\;C\|\lambda\|_{W^{1,\infty}}\|\varphi_n\|_{H^2}
\le Cn^2.\label{eq:app-d-display-009}
\end{align}
Combining these bounds with \eqref{eq:app-d-display-007} gives
\begin{align}
\|A_0v\|&\le\;Cn^4n^{-1}+Cn^2\le Cn^3,
\qquad \|v\|_{H^4}\le C(\|A_0v\|+\|v\|)\le Cn^3.\label{eq:app-d-display-010}
\end{align}
Interpolation gives $\|v\|_{H^s}\le Cn^{s-1}$ for $0\le s\le4$.
To obtain pointwise bounds, we apply the one-dimensional
Gagliardo--Nirenberg inequality \cite[p.~233, Eq.~(42)]{Brezis2011}
to $v^{(j)}$. Each derivative $v^{(j)}$, $0\le j\le3$, either vanishes
at both endpoints or has zero mean. Thus Poincar\'e's inequality gives
$\|v^{(j)}\|\le C\|v^{(j+1)}\|$, allowing the additional $L^2$ term
in the Gagliardo--Nirenberg inequality to be absorbed into the product below:
\begin{align}
\|v^{(j)}\|_{L^\infty}
&\le\;C\|v^{(j)}\|^{1/2}\|v^{(j+1)}\|^{1/2}
\le Cn^{j-1/2},\qquad 0\le j\le3.\label{eq:app-d-display-012}
\end{align}

It remains to bound the boundary traces. The trace coefficients can be written as
\begin{align}
\beta_n=s_n'(0)+v'(0),\qquad
\alpha_n=s_n'''(0)+v'''(0)
+\lambda(0)\big(s_n'(0)+v'(0)\big).\label{eq:app-d-display-013}
\end{align}
Since
\begin{align}
s_n'(0)&=\;\sqrt2\,n\pi,
\qquad s_n'''(0)=-\sqrt2(n\pi)^3,\label{eq:app-d-display-014}
\end{align}
\eqref{eq:app-d-display-012} gives the high-mode bounds in
\eqref{eq:trace-bounds}. For the finitely many remaining modes, the upper
bounds follow from the elliptic estimate and the trace inequality after
enlarging the constant. This completes (iii).
\end{proof}

\subsubsection{Resolvent and input coefficients under pre-feedback}\label{sssec:prefeedback-resolvent}
We now give the resolvent and adjoint of the pre-compensated
operator and characterize its spectrum. Importantly, given the simplicity conditions, we showcase that the curvature-input coefficients are non-zero for eigenvalues moved by the pre-feedback
and those retained from double plant eigenvalues.

\begin{lemma}[Pre-feedback resolvent, spectrum, and input coefficients]\label{lem:spectrum}
Let $\lambda\in W^{1,\infty}(0,1)$, $\varepsilon\in\R$, and $g\in\LL$ be real.
For $z\notin\sigma(\Aop)$, define the boundary lifting and secular function by
\begin{align}
\rho(z):=\sum_{n\ge1}\frac{\alpha_n}{\mu_n-z}\varphi_n,\qquad
\Xi(z):=1-\varepsilon\ip{g}{\rho(z)}
=1-\varepsilon\sum_{n\ge1}\frac{g_n\alpha_n}{\mu_n-z}.
\label{eq:secular-function}
\end{align}
Then the following hold.
\begin{enumerate}[label=(\roman*)]
\item $\Ap$ is closed and densely defined, and for $z\notin\sigma(\Aop)$ with $\Xi(z)\ne0$,
\begin{align}
(\Ap-z)^{-1}f&=\;(\Aop-z)^{-1}f+\frac{\varepsilon\ip{g}{(\Aop-z)^{-1}f}}{\Xi(z)}\,\rho(z),\qquad f\in\LL .\label{eq:resolvent}
\end{align}
The plant resolvent is compact, and the correction has rank at most one.
Hence $\Ap$ has compact resolvent.
\item The adjoint has domain $D(\Ap{}^*)=D(\Aop)$ and, for
$\phi\in D(\Aop)$, acts as
\begin{align}
\Ap{}^*\phi&=\;\Aop\phi-\varepsilon\,\qd\phi(0)\,g.\label{eq:app-d-display-016}
\end{align}
\item $z\in\sigma(\Ap)$ if and only if one of the following holds:
\begin{enumerate}[label=(\alph*)]
\item $z\notin\sigma(\Aop)$ and $\Xi(z)=0$.
\item $z\in\sigma(\Aop)$ and $E_z\cap g^\perp\ne\{0\}$.
\item $z\in\sigma(\Aop)$ and $\Xi$ is regular at $z$. For $\varepsilon\ne0$,
this is equivalent to $\sum_{\mu_n=z}g_n\alpha_n=0$.
\end{enumerate}
Here $E_z$ is the eigenspace of $\Aop$ at $z$, and $g^\perp$ consists
of the functions in $\LL$ orthogonal to $g$.
\item Assume $\varepsilon\ne0$. Let $\mu'$ be real, with $\mu'\notin\sigma(\Aop)$ and $\Xi(\mu')=0$.
Then $\mu'$ is algebraically simple if and only if $\Xi'(\mu')\ne0$.
Moreover, when $\Xi'(\mu')\ne0$, the eigenvectors $\chi_{\mu'}$ and $\tilde\chi_{\mu'}$
of the pre-compensated operator $\Ap$ and its adjoint $(\Ap)^*$,
respectively, associated with $\mu'$ and normalized by
$\ip{\chi_{\mu'}}{\tilde\chi_{\mu'}}=1$, are given by
\begin{align}
\chi_{\mu'}&:=\;\rho(\mu'),\qquad
\tilde\chi_{\mu'}:=-\frac{\varepsilon}{\Xi'(\mu')}(\Aop-\mu')^{-1}g.
\label{eq:app-d-display-019}
\end{align}
The effect of the curvature input on this mode is given by the coefficient
\begin{align}
b'_{\mu'}&:=\;\ip{B_2}{\tilde\chi_{\mu'}}=-\tilde\chi_{\mu'}'(0)
=\frac{\varepsilon}{\Xi'(\mu')}\sum_{n\ge1}\frac{g_n\beta_n}{\mu_n-\mu'}.
\label{eq:bprime-moved}
\end{align}
\item Assume $\varepsilon\ne0$. Let $\mu$ be a double eigenvalue of $\Aop$ with orthonormal eigenfunctions $\varphi_j,\varphi_k$, and define
\begin{align}
\mathsf r_\mu&:=\;g_j\alpha_j+g_k\alpha_k.\label{eq:app-d-display-021}
\end{align}
If $\mathsf r_\mu\ne0$, then $\mu$ is an algebraically simple eigenvalue of $\Ap$ with
\begin{align}
\chi_\mu&=\;g_k\varphi_j-g_j\varphi_k,\qquad \tilde\chi_\mu=\frac{\alpha_k\varphi_j-\alpha_j\varphi_k}{\mathsf r_\mu},\qquad b'_\mu=-\frac{\beta_j\alpha_k-\beta_k\alpha_j}{\mathsf r_\mu}\ne0 .\label{eq:app-d-display-083}
\end{align}
If $\mathsf r_\mu=0$, then $\mu$ has algebraic multiplicity at least two.
\end{enumerate}
\end{lemma}

\begin{proof}[Proof of Lemma~\ref{lem:spectrum}]
We first prove (i). We add a boundary lifting to the plant resolvent
to enforce the pre-feedback condition. Fix $z\notin\sigma(\Aop)$.
The lifting $\rho(z)$ is the $H^4$ solution of the boundary-value problem
\begin{align}
-\rho(z)''''-(\lambda\rho(z)')'&=\;z\rho(z),\qquad
\rho(z)(0)=\;1,\qquad \rho(z)(1)=\rho(z)''(0)=\rho(z)''(1)=0.
\label{eq:boundary-lifting-homogeneous}
\end{align}
Since $z\notin\sigma(\Aop)$, this problem has a unique solution.
Multiplying it by a scalar changes only its left displacement.

Lemma~\ref{lem:ibp} gives
$(\mu_n-z)\ip{\rho(z)}{\varphi_n}=\alpha_n$, yielding the series in
\eqref{eq:secular-function}. The series converges in $L^2$, while its
boundary traces are supplied by the boundary-value problem.

For $\phi\in D(\Ap)$, Lemma~\ref{lem:ibp} with test function $\varphi_n$,
the boundary conditions $\phi''(0)=0$ and
$\phi(0)=\varepsilon\ip{g}{\phi}$, and
$\qd{\varphi_n}(0)=\alpha_n$ give
\begin{align}
\ip{\Ap\phi}{\varphi_n}&=\;\mu_n\ip{\phi}{\varphi_n}
-\varepsilon\alpha_n\ip{g}{\phi},\qquad n\ge1.
\label{eq:Aprime-modal}
\end{align}

For the equation $(\Ap-z)\psi=f$, set $p=\ip{g}{\psi}$.
The modal identity \eqref{eq:Aprime-modal} gives
\begin{align}
(\mu_n-z)\psi_n-\varepsilon\alpha_np&=\;f_n.\label{eq:app-d-display-023}
\end{align}
Solving in the plant basis and then applying the observation yields
\begin{align}
\psi&=\;(\Aop-z)^{-1}f+\varepsilon p\rho(z),
\qquad p\Xi(z)=\ip{g}{(\Aop-z)^{-1}f}.\label{eq:app-d-display-024}
\end{align}
This proves \eqref{eq:resolvent} when $\Xi(z)\ne0$. Conversely,
\eqref{eq:boundary-lifting-homogeneous} gives the homogeneous boundary
traces, while $\psi(0)=\varepsilon p=\varepsilon\ip{g}{\psi}$,
so the formula satisfies the domain conditions.

The eigenvalue and trace bounds \eqref{eq:mu-asymp} and
\eqref{eq:trace-bounds} give $\|\rho(z)\|\to0$ as real $z\to+\infty$,
so $\Xi(z)\to1$. Thus the resolvent formula applies for sufficiently
large real $z$.
The resolvent is compact as the sum of a compact operator and a rank-one
operator. Its boundedness and full domain also show that $\Ap$ is closed.

To verify density of the domain, approximate an arbitrary $f\in\LL$ by
$\phi_0\in D(\Aop)$ with error at most $\delta$. Choose a smooth boundary
lifting $\psi_\delta$ supported near zero such that
\begin{align}
\psi_\delta(0)&=\;1,\qquad \psi_\delta''(0)=0,
\qquad \psi_\delta(1)=\psi_\delta''(1)=0,
\qquad \|\psi_\delta\|\le\delta.\label{eq:app-d-display-026}
\end{align}
Such a lifting is obtained by rescaling a fixed smooth bump. Set
\begin{align}
\phi&=\;\phi_0+c\psi_\delta,
\qquad
c=\frac{\varepsilon\ip{g}{\phi_0}}{1-\varepsilon\ip{g}{\psi_\delta}}.
\label{eq:app-d-display-027}
\end{align}
For small $\delta$, the denominator stays away from zero and $c$ remains
bounded. Thus $\phi\in D(\Ap)$ and
$\|f-\phi\|\le\delta+|c|\delta\to0$, proving density of $D(\Ap)$
and completing (i).

We next prove (ii), identifying the adjoint and its domain. For $\phi\in D(\Ap)$ and $\zeta\in D(\Aop)$, the boundary identity gives
\begin{align}
\ip{\Ap\phi}{\zeta}
&=\;\ip{\phi}{\Aop\zeta}
 -\varepsilon\ip{g}{\phi}\sum_n\alpha_n\zeta_n.\label{eq:app-d-display-029}
\end{align}
The eigenfunction expansion of $\zeta$ converges in the graph norm and
hence in $H^4$, giving $\sum_n\alpha_n\zeta_n=\qd\zeta(0)$.
The series also converges absolutely, since the trace coefficients grow
like $n^3$ and $(n^4\zeta_n)_n\in\ell^2$. This proves the adjoint
formula and the inclusion $D(\Aop)\subseteq D(\Ap{}^*)$.
For the reverse domain inclusion, take the adjoint
of \eqref{eq:resolvent}:
\begin{align}
\big((\Ap-z)^{-1}\big)^*
&=\;(\Aop-\bar z)^{-1}
 +\frac{\varepsilon}{\overline{\Xi(z)}}
  (\Aop-\bar z)^{-1}g\,\ip{\cdot}{\rho(z)}.\label{eq:app-d-display-031}
\end{align}
Its range is contained in $D(\Aop)$ and equals $D(\Ap{}^*)$.
This proves (ii).

We next prove (iii). If $\varepsilon=0$, then $\Ap=\Aop$ and
$\Xi\equiv1$, so case (c) gives exactly the plant spectrum.
Assume henceforth that $\varepsilon\ne0$. Since $\Ap$ has compact
resolvent, it suffices to characterize its eigenvalues.
Set $f=0$ in \eqref{eq:app-d-display-023} and recall that
$p=\ip{g}{\psi}$.

For $z\notin\sigma(\Aop)$, \eqref{eq:app-d-display-024} gives
$\psi=\varepsilon p\rho(z)$ and $p\Xi(z)=0$, so a nonzero eigenvector
requires $\Xi(z)=0$. Conversely, since $\rho(z)(0)=1$, the condition
$\Xi(z)=0$ is precisely $1=\varepsilon\ip{g}{\rho(z)}$.
Thus $\rho(z)$ satisfies the pre-feedback boundary condition and is an
eigenfunction of $\Ap$. This proves case (a).

At a plant eigenvalue, every nonzero $\psi\in E_z\cap g^\perp$
satisfies the pre-feedback condition and remains an eigenfunction,
proving case (b). By Lemma~\ref{lem:plant-spectral}(i), the plant
eigenvalue is either simple or double. At a double eigenvalue,
$E_z\cap g^\perp$ is necessarily nontrivial. At a simple eigenvalue
$z=\mu_n$, $g_n=0$ also gives case (b).

It remains to consider a simple plant eigenvalue with $g_n\ne0$.
The $n$th modal equation gives $\varepsilon\alpha_np=0$.
If $\alpha_n\ne0$, then $p=0$, all coefficients except $\psi_n$
vanish, and $p=g_n\psi_n$ forces $\psi=0$.
If $\alpha_n=0$, the coefficient along $\varphi_n$ remains free.
For $m\ne n$, the modal equations give
$\psi_m=\varepsilon\alpha_mp/(\mu_m-\mu_n)$, and $\psi_n$ can be
chosen to satisfy $p=\ip{g}{\psi}$ for any nonzero $p$.
Thus a nonzero eigenvector exists precisely when $\alpha_n=0$,
which is the pole-cancellation condition $g_n\alpha_n=0$ in case (c).
This exhausts the possibilities and proves (iii).

We next prove (iv). At a real root $\mu'$ outside the plant spectrum, the right eigenvector
is proportional to $\rho(\mu')$. Substituting the adjoint formula from
part (ii) into
$(\Ap)^*\tilde\chi_{\mu'}=\mu'\tilde\chi_{\mu'}$ gives
\begin{align}
(\Aop-\mu')\tilde\chi_{\mu'}
&=\;\varepsilon\qd{\tilde\chi_{\mu'}}(0)g,\label{eq:app-d-display-035}
\end{align}
so an unnormalized adjoint eigenvector is
\begin{align}
\zeta&:=\;(\Aop-\mu')^{-1}g
=\sum_n\frac{g_n}{\mu_n-\mu'}\varphi_n.
\label{eq:app-d-display-017}
\end{align}
The consistency condition is precisely $\Xi(\mu')=0$, since
\begin{align}
\qd{[(\Aop-\mu')^{-1}g]}(0)
&=\;\sum_n\frac{g_n\alpha_n}{\mu_n-\mu'}
=\varepsilon^{-1}.\label{eq:app-d-display-036}
\end{align}
The pairing of these two eigenvectors is
\begin{align}
\ip{\rho(\mu')}{(\Aop-\mu')^{-1}g}
&=\;\sum_n\frac{g_n\alpha_n}{(\mu_n-\mu')^2}
=-\varepsilon^{-1}\Xi'(\mu').\label{eq:app-d-display-037}
\end{align}
By \eqref{eq:resolvent}, the resolvent has a simple pole with a rank-one
spectral projection exactly when this pairing is nonzero. Normalization
and the derivative trace give \eqref{eq:app-d-display-019}--\eqref{eq:bprime-moved},
proving (iv).

Finally, we prove (v). Suppose first that $\mathsf r_\mu\ne0$. The vector $\chi_\mu\in E_\mu\cap g^\perp$ remains an eigenvector by
part (iii). Expanding \eqref{eq:secular-function} at $\mu$ gives
\begin{align}
\Xi(z)&=\;-\frac{\varepsilon\mathsf r_\mu}{\mu-z}+O(1),
\qquad \Xi(z)^{-1}=O(\mu-z).\label{eq:app-d-display-038}
\end{align}
Writing $\hat\alpha=\alpha_j\varphi_j+\alpha_k\varphi_k$, the spectral
projection obtained from \eqref{eq:resolvent} is
\begin{align}
&P_\mu-\mathsf r_\mu^{-1}
 (P_\mu\hat\alpha)\ip{P_\mu g}{\cdot}.\label{eq:app-d-display-040}
\end{align}
This projection has rank one, as its trace is $2-1$. Thus the retained
eigenvalue is algebraically simple and has the eigenvectors stated in
(v). Its curvature coefficient is nonzero by the injectivity of the boundary-trace map in
Lemma~\ref{lem:plant-spectral}(ii). When $\mathsf r_\mu=0$, the pole in $\Xi$ is
canceled and the eigenvalue has algebraic multiplicity at least two.
This completes (v).
\end{proof}

\subsubsection{Pre-feedback preserves the simple high modes}\label{sssec:prefeedback-high}
We now show that pre-feedback preserves simplicity and nonzero curvature
coefficients at sufficiently high modes. For a fixed profile $g$, the
cutoff can be chosen uniformly for $|\varepsilon|\le1$. This leaves only
finitely many modes to consider when the amplitude is subsequently reduced.

\begin{lemma}[High modes]\label{lem:high}
Let $\lambda\in W^{1,\infty}$, $g\in\LL$, $\varepsilon\in\R$. There exist an index $N_0$ and a sequence $\vartheta_n\to0$, both depending only on $\|\lambda\|_{W^{1,\infty}}$, $\|g\|$, $|\varepsilon|$ and the tails $\sum_{m\ge N}|g_m|^2$, such that for every $n\ge N_0$ the following holds: the disk $\{|z-\mu_n|\le\mathfrak g_n/2\}$ contains exactly one eigenvalue $\mu_n'$ of $\Ap$; it is real and algebraically simple; and there are eigenvectors $\chi_n$ of $\Ap$ and $\tilde\chi_n$ of $\Ap{}^*$, normalized by $\ip{\chi_n}{\tilde\chi_n}=1$, with
\begin{equation}\label{eq:high-est}
\begin{gathered}
|\mu_n'-\mu_n|\le 2|\varepsilon|\,|g_n|\,|\alpha_n|,\qquad \|\chi_n-\varphi_n\|\le C|\varepsilon|\,|g_n|,\\
b_n':=\ip{B_2}{\tilde\chi_n}=-\beta_n(1+\vartheta_n),\qquad\vartheta_n\to0 .
\end{gathered}
\end{equation}
Moreover $\Ap$ has no eigenvalue with $|z|\ge|\mu_{N_0}|+\mathfrak g_{N_0}/2$ outside these disks, and only finitely many eigenvalues elsewhere.
\end{lemma}

\begin{proof}[Proof of Lemma~\ref{lem:high}]
The gap bound \eqref{eq:gap}, established in the proof of
Lemma~\ref{lem:plant-spectral}(iii), and its trace estimates give simplicity
and nonzero curvature traces for high plant modes. We locate the perturbed eigenvalues,
estimate their eigenvectors and traces, and establish spectral separation.

We first locate the eigenvalue in each high-mode disk.
Fix $n\ge N_1$, set $\mathcal D_n:=\{z\in\C:|z-\mu_n|\le\mathfrak g_n/2\}$, and separate the $n$th pole from the secular function:
\begin{align}
R_n(z)&:=\;\sum_{m\ne n}\frac{g_m\alpha_m}{\mu_m-z},
\qquad
(\mu_n-z)\Xi(z)
=(\mu_n-z)-\varepsilon g_n\alpha_n
 -\varepsilon(\mu_n-z)R_n(z).\label{eq:app-d-display-041}
\end{align}
On $\mathcal D_n$, $|\mu_m-z|\ge\tfrac12|\mu_m-\mu_n|$ for $m\ne n$.
Set $\bar R_n:=\sup_{z\in\mathcal D_n}|R_n(z)|$.
The trace and gap estimates, splitting the sum at $m=n/2$, give
\begin{align}
\bar R_n&\le\;C\sum_{m\ne n}\frac{|g_m|m^3}{|m^4-n^4|}
\le C\left(\sum_{m\ge n/2}|g_m|^2\right)^{1/2}
 \left(\sum_{m\ne n}\frac{m^6}{(m^4-n^4)^2}\right)^{1/2}
 +C\|g\|n^{-1/2}.\label{eq:Rbar}
\end{align}
The square sum is uniformly bounded, with terms comparable to
$(16(m-n)^2)^{-1}$ near $m=n$. Thus the high-mode contribution vanishes
with the $\ell^2$ tail. For $m<n/2$, the bound
$|m^4-n^4|\ge15n^4/16$ gives the remaining $O(\|g\|n^{-1/2})$ term,
so $\bar R_n\to0$. Also, $g_n\to0$ and the cubic gap give
\begin{align}
|\varepsilon g_n\alpha_n|
&\le\;C|\varepsilon||g_n|n^3\le\mathfrak g_n/8\label{eq:app-d-display-043}
\end{align}
for sufficiently large $n$. Choose $N_0$ so that
\begin{align}
|\varepsilon|\bar R_n&\le\;\tfrac14,
\qquad |\varepsilon g_n\alpha_n|\le\mathfrak g_n/8,
\qquad n\ge N_0.\label{eq:app-d-display-044}
\end{align}
On the boundary of the disk, these estimates give
\begin{align}
|\varepsilon g_n\alpha_n+\varepsilon(\mu_n-z)R_n(z)|
&\le\;\mathfrak g_n/4<|\mu_n-z|.\label{eq:app-d-display-045}
\end{align}
Rouch\'e's theorem \cite[p.~70]{Naimark} therefore gives exactly one zero of
$(\mu_n-z)\Xi(z)$ in the disk, counted with multiplicity. This zero is real
by conjugation symmetry. If it differs from $\mu_n$, it satisfies
\begin{align}
\mu_n-\mu_n'
&=\;\frac{\varepsilon g_n\alpha_n}{1-\varepsilon R_n(\mu_n')},\label{eq:app-d-display-046}
\end{align}
which gives the eigenvalue-displacement estimate in \eqref{eq:high-est}.
The multiplicity count also makes the zero simple. For a moved eigenvalue,
$\mu_n-\mu_n'\ne0$, so $\Xi'(\mu_n')\ne0$ and
Lemma~\ref{lem:spectrum}(iv) gives algebraic simplicity.

We next estimate the eigenvectors and curvature coefficient when the eigenvalue is retained.
If $g_n\alpha_n=0$, the high-mode trace bound gives $g_n=0$. The zero is
then the retained eigenvalue $\mu_n$, with $\chi_n=\varphi_n$.
The adjoint eigenvector, normalized by its $n$th coefficient, is
\begin{align}
\tilde\chi_n&=\;\varphi_n+
\frac{\varepsilon\alpha_n}{1-\varepsilon R_n(\mu_n)}
\sum_{m\ne n}\frac{g_m}{\mu_m-\mu_n}\varphi_m.\label{eq:app-d-display-048}
\end{align}
Splitting this sum at $m=n/2$ and applying \eqref{eq:gap} gives
\begin{align}
\|\tilde\chi_n-\varphi_n\|
&\le\;C|\varepsilon|
\left[\left(\sum_{m\ge n/2}|g_m|^2\right)^{1/2}
 +\frac{\|g\|}{n}\right]\longrightarrow0.\label{eq:app-d-display-049}
\end{align}
The $L^2$ estimate does not control the boundary derivative.
Taking the trace in \eqref{eq:app-d-display-048} and using
$|\beta_m|\le Cm$ gives
\begin{align}
b_n'&=\;-\beta_n(1+O(\eta_n)),
\qquad
\eta_n:=n^2\sum_{m\ne n}\frac{|g_m|m}{|m^4-n^4|}.\label{eq:app-d-display-050}
\end{align}
The same low/high splitting shows $\eta_n\to0$. These estimates prove
the asserted bounds for the retained eigenvalue.

For a moved eigenvalue, we obtain the same estimates from the explicit eigenvector formulas.
For $g_n\alpha_n\ne0$, normalize the right eigenvector from
Lemma~\ref{lem:spectrum}(iv) to have $n$th coefficient one,
and denote it by $\chi_n$. Then
\begin{align}
\chi_n&=\;\varphi_n+\sum_{m\ne n}
 \frac{\alpha_m(\mu_n-\mu_n')}
      {\alpha_n(\mu_m-\mu_n')}\varphi_m,\notag\\
\|\chi_n-\varphi_n\|^2
&\le4\varepsilon^2g_n^2
 \sum_{m\ne n}\frac{\alpha_m^2}{(\mu_m-\mu_n')^2}
\le C\varepsilon^2g_n^2.\label{eq:chi-hi}
\end{align}
Here we used the displacement bound, the bounded sum in \eqref{eq:Rbar}, and
\begin{align}
\frac{|\alpha_m|}{|\alpha_n|}&\le\;C(m/n)^3,
\qquad |\mu_m-\mu_n'|\ge\tfrac12|\mu_m-\mu_n|.\label{eq:app-d-display-051}
\end{align}
Write the adjoint eigenvector as
$\tilde\chi_n=\varpi_n\sum_m\frac{g_m}{\mu_m-\mu_n'}\varphi_m$,
where $\varpi_n$ enforces biorthogonality. The corresponding curvature
coefficient $b_n'$ therefore uses a different scaling from
$b'_{\mu_n'}$ in \eqref{eq:bprime-moved}.
Multiplying the numerator and denominator of its trace formula by
$(\mu_n-\mu_n')/g_n$ gives
\begin{align}
\frac{b_n'}{-\beta_n}
&=\;\frac{1+\vartheta_n^{(1)}}{1+\vartheta_n^{(2)}},\label{eq:bratio}
\end{align}
where, with $R_n=R_n(\mu_n')$,
\begin{align}
\vartheta_n^{(1)}
&=\;\frac{\varepsilon\alpha_n}{(1-\varepsilon R_n)\beta_n}
 \sum_{m\ne n}\frac{g_m\beta_m}{\mu_m-\mu_n'},\notag\\
\vartheta_n^{(2)}
&=\;\frac{\varepsilon^2g_n\alpha_n}{(1-\varepsilon R_n)^2}
 \sum_{m\ne n}\frac{g_m\alpha_m}{(\mu_m-\mu_n')^2}.\label{eq:app-d-display-053}
\end{align}
Using $\eta_n$ from \eqref{eq:app-d-display-050} and the splitting used in \eqref{eq:Rbar} gives
\begin{align}
|\vartheta_n^{(1)}|&\le\;C\eta_n\longrightarrow0,\notag\\
|\vartheta_n^{(2)}|
&\le\;C\varepsilon^2|g_n|n^3
 \sum_{m\ne n}\frac{|g_m|m^3}{(m^4-n^4)^2}
\le C\varepsilon^2|g_n|\|g\|\longrightarrow0.\label{eq:app-d-display-054}
\end{align}
Together with the retained-mode estimate \eqref{eq:app-d-display-050}, this gives the sequence $\vartheta_n$ in
\eqref{eq:high-est}.

We now exclude further eigenvalues outside the high-mode disks.
Choose large circles that stay at least $\mathfrak g_n/2$ from each
plant eigenvalue. On such a circle $|z|=R$, split the sum in
\eqref{eq:secular-function} at a fixed index $N$ to obtain
\begin{align}
|\Xi(z)-1|&\le\;|\varepsilon|\sum_{m\le N}
 \frac{|g_m\alpha_m|}{R-|\mu_m|}
 +C|\varepsilon|\left(\sum_{m>N}|g_m|^2\right)^{1/2}.\label{eq:app-d-display-055}
\end{align}
Cauchy--Schwarz and the preceding sum bound give the second term.
Choosing $N$ large, then $R$ large, gives
$\sup_{|z|=R}|\Xi(z)-1|\to0$ along these circles.
Apply Rouch\'e's theorem \cite[p.~70]{Naimark} to the analytic function
$z\mapsto\Xi(z)\prod_{|\mu_m|\le R}(\mu_m-z)$.

Its zeros encode the moved and retained eigenvalues, with algebraic
multiplicities determined by the resolvent formula \eqref{eq:resolvent}
and the moved and retained eigenvalue analysis in
Lemma~\ref{lem:spectrum} and its proof. The total count equals the plant count. With one eigenvalue per high-mode disk, this proves the final assertion.

Finally, we record the spectral separation and uniformity needed below.
Since $g_n\to0$ and $|\alpha_n|\le Cn^3$, the displacement bound
\eqref{eq:high-est} gives
\begin{align}
|\mu_n'-\mu_n|&=\;o(n^3),\qquad n\to\infty.
\label{eq:prefeedback-shift-small}
\end{align}
Combining this with the plant gap \eqref{eq:gap}, for distinct
sufficiently large indices we obtain
\begin{align}
|\mu_i'-\mu_j'|
&\ge\;|\mu_i-\mu_j|-|\mu_i'-\mu_i|-|\mu_j'-\mu_j|
\ge\tfrac14\pi^4|i^4-j^4|.
\label{eq:prefeedback-high-gap}
\end{align}
After enlarging $N_0$, this holds for all distinct $i,j\ge N_0$.
For $|\varepsilon|\le1$, replacing $|\varepsilon|$ by $1$ in these
estimates gives a common $N_0$, uniform convergence of $\vartheta_n$,
and uniform displacement bounds and spectral separation.
\end{proof}

\subsubsection{Selection of the pre-feedback parameters}\label{sssec:prefeedback-low}
We now complete part (i) of Theorem~\ref{thm:generic}.
The gap estimate \eqref{eq:gap} and Lemma~\ref{lem:plant-spectral}(iii)
show that only finitely many plant eigenvalues are double or simple with
$\beta_n=0$. We first give conditions on a real profile $g$ that resolve
these modes for sufficiently small nonzero $\varepsilon$. We then show
that one exponential profile satisfies all the conditions.
\begin{enumerate}[label=(\roman*)]
\item For double $\mu=\mu_j=\mu_k$ ($j\ne k$), take an orthonormal basis $\varphi_j,\varphi_k$ of $E_\mu$ and define
\begin{align}
\hat\alpha_\mu&:=\;\alpha_j\varphi_j+\alpha_k\varphi_k,
\qquad
\hat\beta_\mu:=\beta_j\varphi_j+\beta_k\varphi_k.
\label{eq:app-d-display-069}
\end{align}
Require
\begin{align}
\ip{P_\mu g}{\hat\alpha_\mu}&\ne\;0,
\qquad
\ip{P_\mu g}{\hat\beta_\mu}\ne0.
\label{eq:app-d-display-070}
\end{align}
\item If $\mu_n$ is simple and $\beta_n\ne0$, no additional condition
on $g$ is required.
\item If $\mu_n$ is simple and $\beta_n=0$, require
\begin{align}
\zeta_n&:=\;\sum_{m\ne n}\frac{\beta_m}{\mu_m-\mu_n}\varphi_m,
\qquad \ip{g}{\zeta_n}\ne0.\label{eq:low-mode-condition}
\end{align}
\end{enumerate}

We verify these conditions using the resolvent formula
\eqref{eq:resolvent}, which gives convergence of spectral projections
onto small neighborhoods of the plant eigenvalues as $\varepsilon\to0$.

For a double eigenvalue, the first condition in \eqref{eq:app-d-display-070} gives
$\mathsf r_\mu=g_j\alpha_j+g_k\alpha_k\ne0$.
Lemma~\ref{lem:spectrum}(v) therefore gives simplicity
and a nonzero curvature coefficient for the retained eigenvalue $\mu$.
Near $\mu$, the secular equation can be written as
\begin{align}
(\mu-z)\Xi(z)&=\;(\mu-z)\bigl(1-\varepsilon R_\mu(z)\bigr)
-\varepsilon\mathsf r_\mu,\qquad
R_\mu(z):=\sum_{\mu_n\ne\mu}\frac{g_n\alpha_n}{\mu_n-z}.
\label{eq:low-double-secular}
\end{align}
Since $R_\mu$ is analytic near $\mu$, the implicit function theorem gives
a simple eigenvalue $\mu'(\varepsilon)=\mu-\varepsilon\mathsf r_\mu+O(\varepsilon^2)$,
distinct from $\mu$ for small nonzero $\varepsilon$.

Rescale $\chi_{\mu'(\varepsilon)}$ from
Lemma~\ref{lem:spectrum}(iv) by
$\varepsilon\mathsf r_\mu$. It then converges to $\hat\alpha_\mu$,
and its biorthogonal adjoint converges to
$P_\mu g/\mathsf r_\mu$ in $H^4$. Denote the curvature coefficient
for this scaling by $b'_{\mu'(\varepsilon)}$. Then
\eqref{eq:bprime-moved} gives
\begin{align}
b'_{\mu'(\varepsilon)}&\longrightarrow\;-\frac{\ip{P_\mu g}{\hat\beta_\mu}}{\mathsf r_\mu}
\ne0.\label{eq:app-d-display-071}
\end{align}
Since nonvanishing is independent of scaling, both eigenvalues have
nonzero curvature coefficients.

For a simple plant eigenvalue, simplicity persists for small $\varepsilon$.
Normalize its adjoint eigenvector to have coefficient one along $\varphi_n$.
If $\beta_n\ne0$, continuity gives $b_n'\to-\beta_n\ne0$.

If $\beta_n=0$, the adjoint formula in Lemma~\ref{lem:spectrum}(ii)
gives the first-order expansion
\begin{align}
\tilde\chi_n(\varepsilon)
&=\;\varphi_n+\varepsilon\alpha_n
\sum_{m\ne n}\frac{g_m}{\mu_m-\mu_n}\varphi_m
+O(\varepsilon^2)\quad\text{in }H^4.
\label{eq:low-adjoint-expansion}
\end{align}
Taking the derivative trace at zero yields
\begin{align}
b_n'(\varepsilon)
&=\;-\varepsilon\alpha_n\ip{g}{\zeta_n}+O(\varepsilon^2).
\label{eq:app-d-display-074}
\end{align}
Lemma~\ref{lem:plant-spectral}(ii) ensures $\alpha_n\ne0$, so
\eqref{eq:low-mode-condition} makes this coefficient nonzero for all
sufficiently small nonzero amplitudes. The simple eigenvalues depend
analytically on $\varepsilon$ and hence move by $O(|\varepsilon|)$.

It remains to choose a single profile satisfying all these conditions.
Set $g=g_\theta$. Conditions \eqref{eq:app-d-display-070}
and \eqref{eq:low-mode-condition} are finitely many
inequalities $F(\theta):=\ip{e^{-\theta y}}{v}\ne0$, with fixed $v\in\LL$.
The vectors $v$ are nonzero by Lemma~\ref{lem:plant-spectral}(ii) for
the trace vectors and by the nonzero high-mode coefficients for $\zeta_n$,
whose $L^2$ convergence follows from the eigenvalue and trace bounds.
Each $F$ extends analytically to $\C$ and is not identically
zero: otherwise differentiation at zero would give
$\ip{y^k}{v}=0$ for all $k\ge0$, and density of polynomials would force
$v=0$. The finite union of their zero sets in $(0,\infty)$ is therefore
a discrete set $\mathscr T_\lambda$ depending only on the plant.

Fix $\theta\notin\mathscr T_\lambda$. The proof of Lemma~\ref{lem:high}, with
$g=g_\theta$, gives a cutoff independent of $\varepsilon$ for
$|\varepsilon|\le1$, above which (H1)--(H2) of Definition~\ref{def:admissible} hold.
The preceding finite-mode arguments apply to the remaining eigenvalues,
including simple modes with $\beta_n\ne0$. Choose a common
$\varepsilon_1\le1$ below their amplitude bounds and small enough to keep
eigenvalues arising from distinct plant eigenvalues separated. This gives (H1)--(H2) of Definition~\ref{def:admissible} for every
$0<|\varepsilon|<\varepsilon_1$.

The high-mode separation \eqref{eq:prefeedback-high-gap}, together with
the quartic asymptotics \eqref{eq:mu-asymp} and
\eqref{eq:prefeedback-shift-small}, implies that only finitely many
distinct spectral differences lie in any bounded set. Hence
$\mathscr N'_a$ is discrete, and every
$a>\max\{0,\max_n\operatorname{Re}\mu_n'\}$ outside it gives (H3) of Definition~\ref{def:admissible}.
The same uniform high-mode estimates show that the spectral maximum
is determined by finitely many modes. Their $O(|\varepsilon|)$
displacement established above gives
$\max_n\operatorname{Re}\mu_n'=\max_n\mu_n+O(|\varepsilon|)$.
This completes part (i) of Theorem~\ref{thm:generic}.

\subsubsection{Riesz basis and spectral estimates}

We now prove part (ii) of Theorem~\ref{thm:generic}.
The high-mode estimates give closeness to the plant eigenbasis.
To obtain a Riesz basis, we first establish completeness.
Set $g=g_\theta$ for the given admissible design. Constants may depend on
$\lambda$ and $(a,\varepsilon,\theta)$, but not on mode indices.
We first prove completeness using the resolvent formula in
Lemma~\ref{lem:spectrum}. Choose real $z_0$ in the resolvent sets of both
operators. Lemma~\ref{lem:spectrum}(i) gives
\begin{align}
H&:=\;(\Ap-z_0)^{-1}=(I+R)H_0,
\qquad H_0:=(\Aop-z_0)^{-1},\label{eq:app-d-display-060}
\end{align}
where $R=\varepsilon\Xi(z_0)^{-1}\rho(z_0)\ip{\cdot}{g}$ has rank at most
one. The operator $H_0$ is compact, self-adjoint, and injective.
The plant spectral estimate \eqref{eq:mu-asymp} gives
$cj^{-4}\le s_j(H_0)\le Cj^{-4}$, hence $H_0\in\mathfrak S_1$,
with constants also depending on $z_0$.
Since $H^*=H_0(I+R^*)$ is injective, $\ker(I+R^*)=\{0\}$.
Since $H_0\in\mathfrak S_1$ is self-adjoint and $R^*$ is compact,
Keldysh's theorem \cite[Ch.~V, Thm.~8.1]{GohbergKrein} gives completeness
of the root vectors of $H^*$. The Fredholm alternative makes $I+R^*$,
and hence $I+R$, invertible. Thus $H$ is similar to the injective
product $H_0(I+R)$, to which the same theorem applies. Completeness
transfers to $\Ap$ and its adjoint. By (H1) of Definition~\ref{def:admissible}, their root vectors are
eigenvectors.

We next establish the Riesz-basis property using Lemma~\ref{lem:high}.
Use its high-mode eigenvectors and normalize the finitely many remaining
eigenvectors in $L^2$. The estimate \eqref{eq:high-est} gives
\begin{align}
\sum_{n\ge N_0}\|\chi_n-\varphi_n\|^2
&\le\;C\varepsilon^2\|g\|^2<\infty.\label{eq:app-d-display-059}
\end{align}
Define $S\varphi_n:=\chi_n-\varphi_n$.
Estimate \eqref{eq:app-d-display-059} makes $S$ Hilbert--Schmidt, so
$I+S$ is Fredholm of index zero. Its range is closed and contains the
complete family $\{\chi_n\}$, so it is surjective. Index zero then gives
injectivity. Thus $\chi_n=(I+S)\varphi_n$ is a Riesz basis.
This is the form of Bari's theorem used in
\cite[Ch.~VI, Sec.~2, Bari's theorem]{GohbergKrein} and
\cite[Thm.~6.3]{Guo2001}. The biorthogonal vectors are adjoint eigenvectors.

Finally, the Riesz basis gives the semigroup representation
\begin{align}
e^{t\Ap}f&=\;\sum_n e^{\mu_n't}\ip{f}{\tilde\chi_n}\chi_n.\label{eq:app-d-display-066}
\end{align}
Lemma~\ref{lem:high} and \eqref{eq:mu-asymp} give
$\mu_n'\to-\infty$, so the real parts are bounded above and the series
defines a strongly continuous semigroup.
The high-mode bounds in \eqref{eq:HL-hyp} follow from
\eqref{eq:mu-asymp}, Lemma~\ref{lem:high}, and \eqref{eq:prefeedback-high-gap}.
Quartic eigenvalue growth gives separation for pairs with one low and
one high index. Conditions (H1)--(H2) of Definition~\ref{def:admissible} cover the finitely many remaining
modes and pairs after adjusting the constants. This completes part (ii)
and the proof of Theorem~\ref{thm:generic}.

\section{Fredholm kernel always exists and its transformation is invertible}\label{sec:fredholm}

This section shows that \eqref{eq:modal} is the one-input Fredholm
backstepping problem of \cite{HayatLoko2026}, there called
$F$-equivalence, for the pair $(\Ap,B_2)$. Its operator form is
\eqref{eq:Fequiv-intro}. Theorem~\ref{thm:generic} verifies these hypotheses
for every $\lambda\in W^{1,\infty}(0,1)$ when $\theta$ lies outside a discrete
exceptional set, $\varepsilon$ is sufficiently small and nonzero, and $a>0$
exceeds the largest real part of the pre-compensated eigenvalues and avoids
the locally finite resonance set $\mathscr N'_a$. Using the Riesz basis and spectral estimates of
Theorem~\ref{thm:generic}(ii), we construct the kernel and gain and obtain
invertibility of $I-K$. Throughout this section, $g=g_\theta$ and
$g_n:=\ip{g_\theta}{\varphi_n}$.
Generic constants in this section may depend on $\lambda$ and
$(a,\varepsilon,\theta)$, but not on mode indices.

\subsection{Kernel construction and invertibility}

Theorem~\ref{thm:generic}(ii) supplies the Riesz basis and spectral and control-coefficient bounds needed to apply the Fredholm backstepping theorem of \cite{HayatLoko2026}, yielding an invertible transformation for the pre-compensated plant:

\begin{theorem}[Kernel and gain exist, and the transformation is invertible]\label{thm:kernel}
Let $(a,\varepsilon,\theta)\in\mathscr U_\lambda$, let $\Ap$ be the pre-compensated operator \eqref{eq:Aprime} and $B_2$ the curvature control vector. Then:
\begin{enumerate}[label=(\roman*)]
\item There exist a unique sequence $(K_n)_{n\ge1}$ with $(K_nb_n'n^{-1/2-\epsilon})\in\ell^2$ for all $\epsilon>0$ and a bounded operator $T$ on $\LL$ such that, with $h:=\sum_nK_n\overline{\tilde\chi_n}$ (so that $\ip{\chi_n}{h}=K_n$),
\begin{equation}\label{eq:Fequiv}
\begin{gathered}
T(\Ap+B_2\ip{h}{\cdot})=(\Ap-a)T\quad\text{in }\mathcal L(\mathcal H^{1+s},\mathcal H^{-3+s})\ \text{ for }|s|<\tfrac32,\\
TB_2=B_2\quad\text{in }\mathcal H^{-3/2-\epsilon}\ \text{ for every }\epsilon>0.
\end{gathered}
\end{equation}
Moreover $\sum_nn^{1-2\epsilon}|K_n|^2<\infty$ for every $\epsilon>0$, so that $h\in\LL$, and $h\in H^{1/2-\epsilon}(0,1)$ for every $\epsilon\in(0,\tfrac12]$. Also $\Ap+B_2\ip{h}{\cdot}$ generates a $C^0$ semigroup on $\mathcal H^r$ for $r\in(-9/2,5/2)$.

\item \emph{(Invertibility.)} $T$ is an isomorphism of $\mathcal H^r$ for every $r\in(-9/2,5/2)$. In particular $T=I-K$ is boundedly invertible on $\LL$, so that
\begin{equation}\label{eq:transformation-bounds}
c_1=\|I-K\|_{\mathcal L(L^2)}<\infty,\qquad c_2=\|(I-K)^{-1}\|_{\mathcal L(L^2)}<\infty .
\end{equation}
\item Explicitly, $T\chi_n=-K_n\sum_{p\ge1}\dfrac{b_p'}{\mu_n'-\mu_p'+a}\chi_p$, and $K_nb_n'=-(a+\kappa_n)$ where $(\kappa_n)$ is the unique solution in $\ell^\infty$ of
\begin{equation}\label{eq:kappa}
\kappa_m=-a\sum_{n\ne m}\frac{a+\kappa_n}{\mu_n'-\mu_m'+a},\qquad m\ge1,
\end{equation}
and $(\kappa_nn^{\epsilon})\in\ell^2$ for every $\epsilon\in(0,5/2)$. The feedback and kernel are
\begin{equation}\label{eq:hk-explicit}
h(y)=\sum_{n\ge1}K_n\overline{\tilde\chi_n(y)},\qquad
k(x,y)=\sum_{n,p\ge1}\Big(\delta_{np}+\frac{K_n\,b_p'}{\mu_n'-\mu_p'+a}\Big)\chi_p(x)\,\overline{\tilde\chi_n(y)}.
\end{equation}
the second series converging in $L^2((0,1)^2)$. The operator $K:=I-T$ is Hilbert--Schmidt with kernel $k$.
\item The coefficients $K_{mj}=\ip{K\varphi_j}{\varphi_m}$ and $h_j=\ip{h}{\varphi_j}$ solve the modal system \eqref{eq:modal}. Conversely any solution of \eqref{eq:modal} in the class of (i) coincides with it.
\end{enumerate}
\end{theorem}

\begin{proof}
We first prove (i). We apply \cite[Prop.~4.1 and Thm.~3.3]{HayatLoko2026} to
$(\mathcal A,B)=(\Ap,B_2)$ with $m=1$, using Theorem~\ref{thm:generic}(ii)
to verify the hypotheses. Condition (H1) of Definition~\ref{def:admissible} gives algebraic simplicity,
and Theorem~\ref{thm:generic}(ii) gives semigroup generation, a Riesz basis,
and the bounds \eqref{eq:HL-hyp}. These verify conditions (6)--(7)
and (13) of that reference with $\alpha=4$, $\beta=-1$, and $\gamma=0$.
Condition (H3) of Definition~\ref{def:admissible} supplies the admissible target shift $a$.

Proposition~4.1 of that reference gives a feedback functional
$\mathcal F$, with coefficients $\mathcal F\chi_n=K_n$, satisfying
\begin{align}
\mathcal F&\in\mathcal L(\mathcal H^{-1/2+\epsilon};\C),
\qquad (n^{1/2-\epsilon}K_n)_n\in\ell^2,
\qquad \epsilon>0,\label{eq:app-e-display-003}
\end{align}
and a transformation satisfying \eqref{eq:Fequiv-intro}, with
\begin{align}
T:\mathcal H^r&\longrightarrow\;\mathcal H^r
\quad\text{an isomorphism},\qquad -\frac92<r<\frac52.\label{eq:app-e-display-004}
\end{align}
This is the cited range for $\mathcal H^{\beta+r}$, with $\beta=-1$
and $1/2-\alpha<r<\alpha-1/2$. Its equations (50) and (90) give
$TB_2=B_2$ and the intertwining identity in
$\mathcal L(\mathcal H^{1+s},\mathcal H^{-3+s})$ for $|s|<3/2$.
Its Lemma~5.19 gives semigroup generation for the closed-loop operator
$\Ap+B_2\mathcal F$.

Lemma~5.7 of \cite{HayatLoko2026} gives uniqueness of the feedback
coefficients. It remains to verify the gain regularity.
By Lemma~\ref{lem:spectrum}(ii), $D(\Ap{}^*)=D(\Aop)$, with graph norm
equivalent to the $H^4$ norm. Interpolating the adjoint eigenvector
synthesis map between $\ell^2\to L^2$ and $\ell^2_4\to H^4$, using the
Riesz and eigenvalue bounds of Theorem~\ref{thm:generic}(ii), gives boundedness
$\ell^2_s\to H^s$ for $0\le s\le4$. Thus
\eqref{eq:app-e-display-003} gives $h\in H^s(0,1)$ for $0\le s<1/2$,
with its series converging in $L^2$.

Uniqueness and the real plant data give a real-valued gain. Together with
the operator identities and semigroup generation above, this proves (i).

We next prove (ii). The isomorphism assertion is \eqref{eq:app-e-display-004}.
Taking $r=0$ in \eqref{eq:app-e-display-004} and using the Riesz-basis
norm equivalence gives bounded invertibility on $\LL$, proving (ii).

We next prove (iii). Equations (55), (77), and (83) of
\cite{HayatLoko2026} give the stated spectral formulas. The proof of its Lemma~5.8 gives
$(\kappa_n n^\epsilon)_n\in\ell^2$ for $0<\epsilon<5/2$.

In the Riesz basis, the matrix of $K=I-T$ is
\begin{align}
\hat K_{pn}&=\;\delta_{np}
 +\frac{K_nb_p'}{\mu_n'-\mu_p'+a},
\qquad \hat K_{nn}=-\frac{\kappa_n}{a},\label{eq:app-e-display-007}
\end{align}
where the diagonal identity follows from $K_nb_n'=-(a+\kappa_n)$.
Boundedness of $(\kappa_n)$ and Theorem~\ref{thm:generic}(ii) give
$|K_n|\le Cn^{-1}$ and $|b_p'|\le Cp$. Together with the spectral gaps,
(H3) of Definition~\ref{def:admissible}, and $(\kappa_n)\in\ell^2$, these imply
\begin{align}
\sum_{p,n}|\hat K_{pn}|^2&<\;\infty.\label{eq:app-e-display-009}
\end{align}
Thus $K$ is Hilbert--Schmidt, and the kernel series
\eqref{eq:hk-explicit} converges in $L^2((0,1)^2)$.
Uniqueness and the real plant data make the coefficients conjugate-symmetric,
so conjugate eigenvalue pairs combine to give real $h$ and $k$.
This proves (iii).

Finally, to prove (iv), project \eqref{eq:Fequiv-intro} onto the original plant
eigenbasis, using \eqref{eq:Aprime-modal} and $T=I-K$.
Reversing the computation in Proposition~\ref{prop:modal}, $TB_2=B_2$
gives \eqref{eq:modal-b}, and the intertwining gives \eqref{eq:modal-a}.
Conversely, a solution in the stated class gives the same operator identities,
so uniqueness of the feedback and transformation gives the same modal coefficients.
This proves (iv).

\end{proof}

The eigenvalues and eigenfunctions in \eqref{eq:intro-eig} are $\tilde\mu_n=\overline{\mu_n'}$ and $\tilde\chi_n$, the eigenpairs of $\Ap{}^*$ in \eqref{eq:adjoint-eigenvector-normalization}, and $\tilde\kappa_n=\overline{\kappa_n}$ is the bounded solution of \eqref{eq:intro-kappa}. Since $h$ is real, \eqref{eq:hk-explicit} with $K_n=-(a+\kappa_n)/b_n'$ and $b_n'=-\overline{\tilde\chi_n'(0)}$ coincides with \eqref{eq:intro-gain}. Theorems~\ref{thm:generic} and~\ref{thm:kernel} therefore prove the assertions of Theorem~\ref{thm:intro} on the eigenvalues, the gain and the coefficients $\tilde\kappa_n$.

For each admissible design, the following lemma establishes smoothing of $K$
and bounded invertibility of $T$ and its modal adjoint on the stated modal
spaces. Uniform bounds over the admissible class are established in
Section~\ref{ssec:uniform}. Here $T^\sharp$ denotes the modal adjoint,
defined by $\hp{Tv}{w}=\hp{v}{T^\sharp w}$.

\begin{lemma}[Smoothing of $K$ and invertibility on modal spaces]\label{lem:smoothing}
Let $(a,\varepsilon,\theta)\in\mathscr U_\lambda$ and let $T=I-K$ be the Fredholm transformation \eqref{eq:transform} with the kernel of Theorem~\ref{thm:kernel}.
\begin{enumerate}[label=(\roman*)]
\item $K\in\mathcal L(\mathcal H^{0},\mathcal H^{2})\cap\mathcal L(\mathcal H^{-1},\mathcal H^{1})$. In particular $K:\LL\to\mathcal H^2\subset H^2(0,1)$ is compact and $I-K$ is Fredholm of index zero on $\LL$.
\item $T$, $T^{-1}$ are isomorphisms of $\mathcal H^r$ for $r\in(-\tfrac92,\tfrac52)$, and $T^\sharp$, $(T^\sharp)^{-1}$ are isomorphisms of $\mathcal H^{-r}$ for the same $r$; in particular $T^\sharp\in\mathcal L(\mathcal H^1)\cap\mathcal L(\mathcal H^2)$, with $\|T^\sharp\|_{\mathcal L(\mathcal H^{-r})}=\|T\|_{\mathcal L(\mathcal H^r)}$.
\end{enumerate}
\end{lemma}

\begin{proof}[Proof of Lemma~\ref{lem:smoothing}]
We first prove (i) by estimating the kernel matrix
\eqref{eq:app-e-display-007} in weighted sequence spaces.
Theorem~\ref{thm:kernel}(iii), Theorem~\ref{thm:generic}(ii), and (H3) of Definition~\ref{def:admissible} give
\begin{align}
|K_n|&\le\;Cn^{-1},\qquad |b_p'|\le Cp,
\qquad |\mu_n'-\mu_p'+a|\ge c(1+|n^4-p^4|)\quad(p\ne n),\label{eq:app-e-display-012}
\end{align}
with constants depending on $\lambda$ and $(a,\varepsilon,\theta)$,
but not on the mode indices. Applying these estimates to the matrix
formula \eqref{eq:app-e-display-007} and the coefficient equation
\eqref{eq:kappa}, using boundedness of $(\kappa_n)$, gives
\begin{align}
|\hat K_{pn}|&\le\;\frac{Cn^{-1}p}{1+|n^4-p^4|}\quad(p\ne n),
\qquad |\hat K_{nn}|\le Cn^{-3}\log(n+1).\label{eq:app-e-display-014}
\end{align}
The logarithm comes from summing $1/|m-n|$ over indices comparable to $n$.
To prove square summability of the weighted matrix, split the
 off-diagonal indices into three regions:
\begin{align}
p^rn^{-r}|\hat K_{pn}|&\le\;C
\begin{cases}
p^{r-3}n^{-r-1},&p>2n,\\
(1+n^3|n-p|)^{-1},&n/2\le p\le2n,\ p\ne n,\\
p^{r+1}n^{-r-5},&p<n/2.
\end{cases}\label{eq:app-g-display-025}
\end{align}
Squaring and summing these bounds, with an additional factor $p^2$
in the matrix entries for smoothing, and using
\eqref{eq:app-e-display-014} on the diagonal gives
\begin{align}
\sum_{p,n}\big(p^rn^{-r}|\hat K_{pn}|\big)^2
 &\le\;C,&&r\in\{-1,0,1,2\},\notag\\
\sum_{p,n}\big(p^{r+2}n^{-r}|\hat K_{pn}|\big)^2
 &\le\;C,&&r\in\{-1,0\}.\label{eq:HSsums}
\end{align}
The second bound proves the smoothing assertions and makes
$K:\LL\to\mathcal H^2$ Hilbert--Schmidt, hence compact.
The Riesz-basis property gives the continuous embedding
$\mathcal H^2\hookrightarrow\LL$, so $K$ is compact on $\LL$ and
$I-K$ is Fredholm of index zero. This proves (i).

To prove (ii), use the isomorphism range \eqref{eq:app-e-display-004}
and duality in the modal pairing. Since
$(\mathcal H^r)^*=\mathcal H^{-r}$, duality gives the norm identity
in Lemma~\ref{lem:smoothing}(ii). Applying the same argument to $T^{-1}$
proves (ii).

\end{proof}

\section{Gain operator, continuity and neural-operator approximation}\label{sec:continuity}

\subsection{Gain operator and admissible class}

\begin{definition}[Gain operator]\label{def:designmap}
For $\lambda\in W^{1,\infty}(0,1)$ and $(a,\varepsilon,\theta)$ admissible for
$\lambda$, let $h$ be the curvature gain in \eqref{eq:hk-explicit}.
We define the \emph{gain operator} as the mapping
\begin{equation}\label{eq:G}
\mathcal G:\ \mathscr U\to\LL,\qquad \mathcal G(\lambda,a,\varepsilon,\theta)=h .
\end{equation}
\end{definition}

This is the operator to be approximated, and the only one challenging to
solve for, since the displacement gain $\varepsilon g_\theta$ is explicit.
However, to approximate this operator, we will do so over a class of
functions $\lambda$ and admissible parameters $(a,\varepsilon,\theta)$.
This class requires two key properties. First, it must be
compact, as neural network approximation theorems rely on compact domains \cite{LanthalerLiStuart2025}.
Second, it must be bounded \emph{away} from inadmissibility to tolerate
any error $\epsilon$ in the approximation.

With the right eigenvectors normalized by $\|\chi_n\|_{L^2}=1$ and $\tilde\chi_n$ chosen biorthogonally, $\|\tilde\chi_n\|\ge1$, so the margin on $\tilde\chi_n'(0)$ in \eqref{eq:margins-lower} gives $|b_n'|\ge dn$. The class $\mathscr U_{B,d}$ coincides with \eqref{eq:intro-class}.

Hence, we define the approximation class as follows. Fix a \emph{size}
$B>0$ and a \emph{margin} $d\in(0,1]$, and let $\mathscr U_{B,d}$ be the
set of parameters $\mathbf p=(\lambda,a,\varepsilon,\theta)$ with
$(a,\varepsilon,\theta)$ admissible for $\lambda$ such that the following holds:
\begin{subequations}\label{eq:margins}
\begin{align}
&\|\lambda\|_{W^{2,\infty}}\le B,\qquad
 a+|\varepsilon|+\theta+\theta^{-1}\le B,\label{eq:margins-size}\\
&\min_{i\ne j}\frac{|\mu_i'-\mu_j'|}{\max(i,j)^3}\ge d,\quad
\min_n\frac{|\tilde\chi_n'(0)|}{n\,\|\tilde\chi_n\|}\ge d,\quad
\operatorname{dist}(a,\mathscr N_a')\ge d,\quad
a-\max_n\operatorname{Re}\mu_n'\ge d.\label{eq:margins-lower}
\end{align}
\end{subequations}
Notice that the bound on $\lambda$ is placed in $W^{2,\infty}$ rather than
in the $W^{1,\infty}$ assumed in Section~\ref{sec:control}. This ensures that the class of
coefficients $\lambda$ is compact in $W^{1,\infty}$ and is needed only for
the training set. The design of Section~\ref{sec:kernel} requires nothing beyond
$\lambda\in W^{1,\infty}$.

\subsection{Continuity of the kernel and gain}

The following theorem establishes continuity of \eqref{eq:G} on the compact admissible class.

\begin{theorem}[Design map is continuous, with uniformly bounded design constants]\label{thm:lipschitz}
Let $B>0$ and $d\in(0,1]$. Then $\mathscr U_{B,d}$ is compact in $W^{1,\infty}(0,1)\times\R^3$, the gain operator $\mathcal G:\mathscr U_{B,d}\to\LL$ is continuous, and
\begin{equation}\label{eq:unif-const}
\sup_{\mathscr U_{B,d}}\Big(\|h\|_{\LL}+\|k\|_{L^2((0,1)^2)}+c_1+c_2+m^{-1}+M+c_\gamma^{-1}+\Lambda+C_{\rm N}\Big)<\infty ,
\end{equation}
where $c_1,c_2$ are defined in \eqref{eq:transformation-bounds},
$m,M$ in \eqref{eq:riesz-const}, and $c_\gamma,\Lambda,C_{\rm N}$ in
\eqref{eq:dissipation-coercivity}, \eqref{eq:Lambda}, and
\eqref{eq:nonlinear-constant}, with $\nu_0=d/2$.
\end{theorem}

Theorem~\ref{thm:lipschitz} is not a routine continuity statement, for two
reasons. First, the gain is not given by a formula in $\lambda$ but through
the spectral data $(\mu_n',\chi_n,\tilde\chi_n,b_n')$ of the
pre-compensated operator $\Ap$, which is non-self-adjoint and whose domain
\eqref{eq:Aprime} moves with $\varepsilon$ and $\theta$. Hence,
perturbation results for a fixed operator do not apply directly. Second,
the eigenfunctions of $\Aop$ need not depend continuously on $\lambda$: at
a double eigenvalue of $\Aop$, which is exactly the configuration the
two-input design exists to handle, the eigenspace varies continuously but
no continuous choice of individual eigenfunctions exists. A
coefficient-by-coefficient argument therefore fails. The proof instead
reads all low-frequency spectral data of $\Ap$ off
its Riesz projections, which are continuous, and controls the
high-frequency data by asymptotics that hold uniformly on the class. In
this way, the gain and the transformation are shown to be continuous where
no continuous eigenbasis exists.

\subsection{Proof of Theorem~\ref{thm:lipschitz}: continuity of the kernel and gain}\label{ssec:continuity-proof}

We follow how a change in the plant coefficient and design parameters
passes through the construction of the gain and kernel. We first show
that the pre-compensated resolvent varies continuously and that limits
preserve the admissibility margins, proving compactness of
$\mathscr U_{B,d}$. Spectral projections then allow us to choose
eigenvectors continuously near each admissible design. Although the
eigenbasis changes with the design, its modal coefficients belong to the
same sequence spaces. We use these spaces to compare the reconstruction
maps and the infinite coefficient system, proving that its solution
varies continuously. Combining the coefficients with the eigenvectors
then gives continuity of the gain in $L^2(0,1)$ and the kernel in
$L^2((0,1)^2)$. At each stage, we handle finitely many modes by continuity
and control the remaining modes by uniform summable estimates, including
the kernel bounds from Lemma~\ref{lem:smoothing}(i). Finally, compactness
and continuity of inversion bound the transformations and their inverses.
The spectral estimates and Sobolev bounds from Appendix~\ref{app:modal-sobolev}
then give uniform bounds on the remaining stability constants.

\subsubsection{Resolvent continuity and uniform spectral estimates}\label{sssec:continuity-resolvent}
We first establish resolvent convergence under the size bounds
\eqref{eq:margins-size}, without assuming that the limiting parameters
are admissible. Subscript $0$ denotes the limiting parameters and
associated objects. The plant operators have a fixed domain. Since
$A_\lambda-A_{\lambda_0}=-\partial_x((\lambda-\lambda_0)\partial_x\,\cdot\,)$,
its $H^4$-to-$\LL$ norm tends to zero. The resolvent identity gives
\begin{align}
(z-A_\lambda)^{-1}&\longrightarrow\;(z-A_{\lambda_0})^{-1}
\quad\text{in }\mathcal L(\LL)
\text{ and }\mathcal L(\LL,H^4),\label{eq:app-g-display-010}
\end{align}
uniformly on compact subsets of the common resolvent set. Subtracting
from $\rho(z)$ a fixed smooth function with the same boundary data reduces
its boundary-value problem to one with homogeneous boundary conditions.
Thus \eqref{eq:app-g-display-010} gives $\rho(z)\to\rho_0(z)$ in $H^4$.
Continuity of $g_\theta$ then gives $\Xi(z)\to\Xi_0(z)$, with $\rho$ and
$\Xi$ defined in Lemma~\ref{lem:spectrum}.

To obtain a common resolvent point without assuming admissibility,
choose a sufficiently large positive $z_0$. The plant spectral and trace
bounds in \eqref{eq:mu-asymp} and Lemma~\ref{lem:plant-spectral}(iii) give
\begin{align}
\|\rho(z_0)\|^2
&=\;\sum_n\frac{|\alpha_n|^2}{|\mu_n-z_0|^2}
\le C(B)\sum_n\frac{n^6}{(z_0+n^4)^2}
\le C(B)z_0^{-1/4}.
\label{eq:common-resolvent-lifting}
\end{align}
Since $|\varepsilon|\|g_\theta\|$ is bounded by the size constraints,
we may choose $z_0=z_0(B)$ so that
\begin{align}
|\Xi(z_0)-1|
&\le\;|\varepsilon|\|g_\theta\|\|\rho(z_0)\|
\le\tfrac12.
\label{eq:common-resolvent-secular}
\end{align}
Formula \eqref{eq:resolvent} therefore gives convergence of the
pre-compensated resolvents at $z_0$ in both operator norms above,
including at a possibly inadmissible limit. The resolvent identity
extends this convergence to compact subsets of the limiting resolvent set.
The adjoint operators have the common domain $D(\Aop)$ by
Lemma~\ref{lem:spectrum}(ii). Their coefficients and boundary-trace term
vary continuously from $H^4$ to $\LL$, so the same resolvent convergence
holds for the adjoints.

We next obtain estimates that control the high modes uniformly.
The eigenvectors supplied by Section~\ref{sec:existence} are rescaled
to the unit-norm convention above before applying the admissibility margins.
The constants in Lemma~\ref{lem:high} depend only on the coefficient bound,
the size of the pre-feedback, and the modal tails of $g_\theta$.
Integration by parts and Lemma~\ref{lem:plant-spectral}(iii) give
$|g_n|\le C(B)/n$ uniformly over $\mathscr U_{B,d}$.

Lemma~\ref{lem:high} thus has a uniform threshold $N_0=N_0(B)$. For every design and $n\ge N_0$, the disk of radius $\mathfrak g_n/2$ about $\mu_n$ contains exactly one real, simple eigenvalue $\mu_n'$. With $C=C(B)$,
\begin{align}
|\mu_n'-\mu_n|&\le\;Cn^2,\quad
\tfrac12|\beta_n|\le|b_n'|\le2|\beta_n|,\quad
\|\chi_n-\varphi_n\|+\|\tilde\chi_n-\varphi_n\|\le C/n.\label{eq:unif-high}
\end{align}
The first estimate follows from
$|\mu_n'-\mu_n|\le2|\varepsilon||g_n||\alpha_n|\le Cn^2$
in Lemma~\ref{lem:high}.

This perturbation is smaller than the cubic gaps and the fourth-order
spectral growth. After enlarging the threshold, the off-diagonal
nonresonance estimate is
\begin{align}
|\mu_n'-\mu_m'+a|
&\ge\;c(B,d)(1+|n^4-m^4|),\qquad n\ne m.\label{eq:app-g-display-003}
\end{align}
These estimates control the high-mode constraints in \eqref{eq:margins}. We treat the finitely many low modes separately.

\subsubsection{Compactness of the admissible class}\label{sssec:continuity-compactness}
We now show that limits preserve the admissibility margins.
Arzel\`a--Ascoli applied to $\lambda$ and $\lambda'$ makes the bounded
$W^{2,\infty}$ coefficient class relatively compact in $W^{1,\infty}$.
The limit remains in the same ball by weak-$*$ closedness, and the scalar
parameters range over a compact set. Consider a sequence
$\mathbf p_j\in\mathscr U_{B,d}$ with
$\mathbf p_j\to\mathbf p$ in $W^{1,\infty}\times\R^3$.

Section~\ref{sssec:continuity-resolvent} gives norm convergence of the compact resolvents at $z_0$ without
requiring admissibility of $\mathbf p$. Thus isolated eigenvalues and
their Riesz projections converge. The uniform high-mode estimates leave
only finitely many eigenvalues in a common bounded region to consider.
After taking a subsequence and relabeling, their limits
$\zeta_1,\ldots,\zeta_{N_0-1}$ satisfy
\begin{align}
|\zeta_i-\zeta_l|&\ge\;d\max(i,l)^3,\qquad i\ne l.\label{eq:app-g-display-006}
\end{align}
These limits are distinct eigenvalues of the limiting operator. Norm
convergence of the compact resolvents excludes additional eigenvalues
in this bounded region. Choose a small circle $\Gamma_n$ around each
$\zeta_n$, enclosing no other limiting eigenvalue. By Section~\ref{sssec:continuity-resolvent}, the
resolvents converge uniformly on $\Gamma_n$, so the associated Riesz
projections converge and retain rank one. Thus each $\zeta_n$ is
algebraically simple. Applying the same argument
to the adjoints gives locally normalized eigenvectors converging in $H^4$.
The boundary traces therefore satisfy
\begin{align}
b_n'(\mathbf p_j)
&=\;-\overline{\tilde\chi_n'(0;\mathbf p_j)}
\longrightarrow b_n'(\mathbf p).\label{eq:app-g-display-008}
\end{align}
The separation and curvature-input margins pass to these limits.
Section~\ref{sssec:continuity-resolvent} controls the remaining high modes uniformly, and the nonresonance
and target-stability margins are preserved by the same spectral convergence.
The limiting design therefore belongs to $\mathscr U_{B,d}$, proving compactness.

\subsubsection{Continuity of the spectral coordinate maps}\label{sssec:continuity-coordinates}
Fix an admissible design $\mathbf p_0$. We choose eigenvectors continuously
near this design, then compare their modal reconstruction maps on $\ell^2$.

For each fixed $n$, surround $\mu_n'(\mathbf p_0)$ by an isolating
contour $\Gamma$. Choose a fixed $v$ with $P_n(\mathbf p_0)v\ne0$ and set
\begin{align}
P_n(\mathbf p)&=\;\frac1{2\pi i}\oint_\Gamma(z-\Ap)^{-1}\dd z,
\qquad \chi_n(\mathbf p)=\frac{P_n(\mathbf p)v}{\|P_n(\mathbf p)v\|}.
\label{eq:app-g-display-012}
\end{align}
Section~\ref{sssec:continuity-resolvent} makes $P_n$ continuous in both operator norms. Its rank remains
one because projections at operator-norm distance less than one have
equal rank. The eigenvalue and normalized eigenvector therefore vary
continuously. The biorthogonal adjoint eigenvector can also be chosen
continuously in $H^4$, so the trace inequality gives continuity of $b_n'$.
Local eigenvector choices suffice because the reconstructed gain and
kernel are independent of their phases.

Define the synthesis maps by $Je_n=\chi_n$,
$\tilde Je_n=\tilde\chi_n$, and $J_0e_n=s_n$, with $J_0$ unitary.
Comparing with the fixed sine basis avoids following plant eigenfunctions
through a double eigenvalue. For high modes, choose the phases so that
$\ip{\chi_n}{s_n}>0$. Then \eqref{eq:unif-high} and
Lemma~\ref{lem:plant-spectral}(iii) give
\begin{align}
\|\chi_n-s_n\|+\|\tilde\chi_n-s_n\|
&\le\;\frac{C(B)}{n},\qquad n\ge N_0(B).\label{eq:JJ0}
\end{align}
Continuity of each eigenvector and the uniform square-summable tail in
\eqref{eq:JJ0} give operator-norm continuity of $J$ and $\tilde J$.
Let $m,M>0$ be the optimal Riesz constants, defined by
\begin{align}
m\|w\|^2\le\|w\|_{\mathcal H^0}^2\le M\|w\|^2,
\qquad w\in\LL.\label{eq:riesz-const}
\end{align}
Thus $m^{-1}=\|J\|^2$ and $M=\|J^{-1}\|^2$.
Compactness from Section~\ref{sssec:continuity-compactness} and
invertibility of the synthesis maps from Theorem~\ref{thm:generic}(ii) give
\begin{align}
\|J\|+\|\tilde J\|+\|J^{-1}\|&\le\;C(B,d),
\qquad m^{-1},M\le C(B,d).\label{eq:app-g-display-022}
\end{align}
\subsubsection{Continuity of the feedback coefficients}\label{sssec:continuity-coefficients}
We solve the coefficient system in $\ell^\infty$, since the rescaled
coefficients $c_n$ below tend to one. Define
\begin{align}
\mathsf C_{mn}&=\;\begin{cases}
(\mu_n'-\mu_m'+a)^{-1},&n\ne m,\\
0,&n=m,
\end{cases}
\qquad c_n=\frac{a+\kappa_n}{a}.\label{eq:app-g-display-016}
\end{align}
Then \eqref{eq:kappa} becomes
\begin{align}
(I+a\mathsf C)c&=\;\mathbf1.\label{eq:app-g-display-017}
\end{align}
Estimate \eqref{eq:app-g-display-003} bounds the row and column sums:
\begin{align}
\max\left\{\sum_n|\mathsf C_{mn}|,\sum_n|\mathsf C_{nm}|\right\}
&\le\;C(B,d)m^{-3}\log(m+1).\label{eq:app-g-display-018}
\end{align}
These bounds give uniform boundedness on $\ell^\infty$ and compactness
by matrix truncation.
Each entry is continuous by Section~\ref{sssec:continuity-coordinates}, and the uniform tails upgrade this
to operator-norm continuity.

Uniqueness of the bounded solution in Theorem~\ref{thm:kernel}(iii)
implies that $I+a\mathsf C$ has trivial kernel on $\ell^\infty$.
The Fredholm alternative therefore gives invertibility. Continuity of
inversion gives continuity of $c=(I+a\mathsf C)^{-1}\mathbf1$ in
$\ell^\infty$, and compactness of the design class yields
\begin{align}
\sup_{\mathscr U_{B,d}}
\|(I+a\mathsf C)^{-1}\|_{\mathcal L(\ell^\infty)}&<\;\infty.\label{eq:c-cont}
\end{align}
The gain and kernel use two sequences with summable tails. First,
\begin{align}
K_n&=\;-\frac{ac_n}{b_n'},\qquad |K_n|\le C(B,d)n^{-1}.\label{eq:app-g-display-019}
\end{align}
Continuity and the square-summable bound give continuity of $(K_n)$ in
$\ell^2$. Second, the coefficient equation gives
$\kappa=-a^2\mathsf Cc$, so the row-sum estimate yields
\begin{align}
|\kappa_n|
&\le\;a^2\|c\|_{\ell^\infty}\sum_{m\ne n}|\mathsf C_{nm}|
\le C(B,d)n^{-3}\log(n+1).\label{eq:kappa-tail}
\end{align}
Dominated convergence therefore also gives continuity of $(\kappa_n)$
in $\ell^2$.

\subsubsection{Reconstruction of the gain and kernel}\label{sssec:continuity-reconstruction}

We now combine the continuous coefficients and coordinate maps to recover
the gain and kernel as functions.
The gain is reconstructed as
\begin{align}
h&=\;\overline{\tilde J\,(\overline{K_n})_n}.\label{eq:app-g-display-023}
\end{align}
Conjugation is an isometry of $\ell^2$ and of $\LL$, so Sections~\ref{sssec:continuity-coordinates}--\ref{sssec:continuity-coefficients} prove
continuity and uniform boundedness of $h$ in $\LL$.

The kernel matrix is given by \eqref{eq:app-e-display-007}.
Sections~\ref{sssec:continuity-coordinates}--\ref{sssec:continuity-coefficients} give uniform constants in the bounds for $c_n$, $K_n$, and $b_n'$, while
\eqref{eq:app-g-display-003} controls the off-diagonal denominators.
The coefficient equation \eqref{eq:kappa} and \eqref{eq:kappa-tail}
control $\kappa_n$, giving a uniform diagonal bound. Thus the bounds
\eqref{eq:app-e-display-014}--\eqref{eq:app-g-display-025} from the proof of
Lemma~\ref{lem:smoothing}(i) hold uniformly with constants depending only
on $B,d$. They provide a summable majorant for the squared matrix entries.
Entrywise continuity and dominated convergence therefore give
Hilbert--Schmidt continuity of $\hat K$. Since $K=J\hat KJ^{-1}$,
continuity of $J$, $J^{-1}$, and $\hat K$ gives Hilbert--Schmidt continuity
of $K$, equivalently continuity of $k$ in $L^2((0,1)^2)$.

\subsubsection{Uniform bounds on the design constants}\label{ssec:uniform}

The continuity established above bounds $\|h\|$ and $\|k\|$ on
the compact class $\mathscr U_{B,d}$. The coordinate estimate
\eqref{eq:app-g-display-022} bounds $m^{-1}$ and $M$.
We therefore need only bound the transformation and forcing constants.

We now establish uniform bounds for the transformations. Since the
modal spaces depend on the design, we compare the transformations on
the fixed spaces $\ell^2_r$. The matrix bounds
\eqref{eq:app-e-display-014}--\eqref{eq:app-g-display-025}, made uniform
in Section~\ref{sssec:continuity-reconstruction}, and the weighted
summability in \eqref{eq:HSsums} give operator-norm continuity of
$\hat K$ on $\ell^2_r$ for $r\in\{-1,0,1,2\}$.
Lemma~\ref{lem:smoothing}(ii) makes $\hat T=I-\hat K$ invertible at
every design. Continuity of inversion and compactness then give
\begin{align}
\sup_{\mathscr U_{B,d}}
\|\hat T^{\pm1}\|_{\mathcal L(\ell^2_r)}&<\;\infty,
\qquad r\in\{-1,0,1,2\}.\label{eq:app-g-display-026}
\end{align}
Returning to the modal spaces and interpolating gives the bounds for
$T^{\pm1}$. The norm identity in Lemma~\ref{lem:smoothing}(ii),
also applied to $T^{-1}$, gives the bounds for $(T^\sharp)^{\pm1}$.
Together with the smoothing estimate in \eqref{eq:HSsums}, this gives
\begin{align}
\|T^{\pm1}\|_{\mathcal L(\mathcal H^r)}
+\|(T^\sharp)^{\pm1}\|_{\mathcal L(\mathcal H^{-r})}
&\le\;C(B,d),\qquad r\in[-1,2],\label{eq:uniform-transform-scale}\\
\|K\|_{\mathcal L(\mathcal H^r,\mathcal H^{r+2})}
&\le\;C(B,d),\qquad r\in\{-1,0\}.\label{eq:uniform-smoothing-scale}
\end{align}
In particular, \eqref{eq:uniform-transform-scale} controls $T^\sharp$ on
$\mathcal H^1$, as needed for the nonlinear estimate.
Taking $r=0$ and using the uniform $L^2$ norm equivalence from
\eqref{eq:app-g-display-022} yields $c_1,c_2\le C(B,d)$.

To define the nonlinear constant, we first relate the modal inner product
to the usual $L^2$ pairing through
\begin{equation}\label{eq:Gamma}
\Gamma w:=\sum_n w_n\tilde\chi_n=\tilde JJ^{-1}w.
\end{equation}
Then $\hp{f}{w}=\ip{f}{\Gamma w}$ for $f,w\in\LL$.
Lemma~\ref{lem:gamma} in Appendix~\ref{app:modal-sobolev} gives
$\|\Gamma\|_{\mathcal L(\LL)}\le M$ and
$\|\Gamma^{-1}\|_{\mathcal L(\LL)}\le m^{-1}$.
It also shows that $\Gamma:\mathcal H^s\to H^s(0,1)$ is bounded
for $0\le s\le4$, uniformly on $\mathscr U_{B,d}$, and its images
vanish at both endpoints for $1/2<s\le4$.

We next define the dissipation weights and boundary-forcing constant.
For an admissible design, let
\begin{align}
\nu:=a-\max_n\operatorname{Re}\mu_n'=a-\operatorname{Re}\mu_1'>0.
\label{eq:nu}
\end{align}
For $0<\nu_0<\nu$, define
\begin{align}
\gamma_n:=2\big(a-\operatorname{Re}\mu_n'-\nu_0\big).
\label{eq:dissipation-weights}
\end{align}
Every weight satisfies $\gamma_n\ge2(\nu-\nu_0)>0$.
The plant estimate \eqref{eq:mu-asymp} and the displacement estimate
\eqref{eq:prefeedback-shift-small} give
\begin{align}
a-\operatorname{Re}\mu_n'-\nu_0
&=\;\pi^4n^4+o(n^3)\ge\tfrac12\pi^4n^4,
\qquad n\text{ sufficiently large}.\label{eq:app-f-display-012}
\end{align}
Enlarge $N_0$ so that this holds for every $n\ge N_0$, and set
\begin{align}
c_\gamma:=\min\{\pi^4,2(\nu-\nu_0)N_0^{-4}\}>0.
\label{eq:dissipation-coercivity}
\end{align}
Then $\gamma_n\ge c_\gamma n^4$ for all $n\ge1$. Consequently,
\begin{align}\label{eq:dissipation-constants}
\Sigma(w):=\sum_n\gamma_n|w_n|^2\ge c_\gamma\|w\|_{\mathcal H^2}^2,
\qquad w\in\mathcal H^2.
\end{align}
Define
\begin{align}
\Lambda:=\left(\sum_{n\ge1}\frac{|b_n'|^2}{\gamma_n}\right)^{1/2}<\infty.
\label{eq:Lambda}
\end{align}
The series converges because Theorem~\ref{thm:generic}(ii) gives
$|b_n'|\le Cn$, while $\gamma_n\ge c_\gamma n^4$.

For uniform bounds on $\mathscr U_{B,d}$, take $\nu_0=d/2$.
The admissibility margin \eqref{eq:margins-lower} gives $\nu\ge d$.
With $\nu_0=d/2$, this and the uniform high-mode estimates in
Section~\ref{sssec:continuity-resolvent} give
\begin{align}
\gamma_n&\ge\;d\quad(n\ge1),
\qquad \gamma_n\ge\pi^4n^4\quad(n\ge N_0),\label{eq:app-g-display-028}
\end{align}
where $N_0$ is uniform over $\mathscr U_{B,d}$.
Thus $c_\gamma\ge\min\{d/N_0^4,\pi^4\}>0$.
The spectral and coordinate estimates in
Sections~\ref{sssec:continuity-resolvent}--\ref{sssec:continuity-coordinates} give $|b_n'|\le C(B,d)n$ uniformly.
Splitting the series for $\Lambda^2$ at $N_0$ therefore bounds the finite
sum and the summable tail, yielding
\begin{align}
c_\gamma^{-1}&\le\;C(B,d),
\qquad \Lambda^2=\sum_n\frac{|b_n'|^2}{\gamma_n}\le C(B,d).\label{eq:app-g-display-029}
\end{align}

Finally, let $C_{\rm GN}$ be a constant in the Gagliardo--Nirenberg
inequality $\|v\|_{L^4}^2\le C_{\rm GN}\|v\|^{3/2}\|v\|_{H^1}^{1/2}$
\cite[p.~233, Eq.~(42)]{Brezis2011}, and let $C_{\rm emb}$ be the norm
of the embedding $\mathcal H^1\hookrightarrow H^1(0,1)$.
Define the nonlinear constant by
\begin{align}
C_{\rm N}:=C_{\rm GN}\,\|\Gamma\|_{\mathcal L(\mathcal H^1,H^1)}
\,\|T^\sharp\|_{\mathcal L(\mathcal H^1)}
\,\|T^{-1}\|_{\mathcal L(\mathcal H^1)}^{1/2}\,C_{\rm emb}^{1/2}.
\label{eq:nonlinear-constant}
\end{align}
The transformation bounds \eqref{eq:uniform-transform-scale} and the
Sobolev bounds proved in Lemma~\ref{lem:gamma} uniformly bound every
factor. Hence $C_{\rm N}\le C(B,d)$, completing the proof of
\eqref{eq:unif-const}.
\qed

\subsection{Neural-operator approximation of the Fredholm gain}\label{sec:neural-approximation}

We approximate $\mathcal G$ by a neural operator
\cite{LuEtAl2021DeepONet,LanthalerLiStuart2025}, as defined in
\cite[Def.~3]{BhanKrsticShi2026Horizons}, mapping $\lambda$ and
$(a,\varepsilon,\theta)$ to the curvature gain.
Theorem~\ref{thm:lipschitz} and neural-operator universality yield:

\begin{theorem}[Neural-operator approximation {\cite[Thm.~2]{LanthalerLiStuart2025}}]\label{thm:NO}
Let $B>0$, $d\in(0,1]$ and $\epsilon>0$. There exists a neural operator
$\hat{\mathcal G}$, with function input $\lambda$, scalar inputs
$(a,\varepsilon,\theta)$, and output in $\LL$, such that
\begin{equation}\label{eq:NO}
\sup_{\mathbf p\in\mathscr U_{B,d}}
\|\mathcal G(\mathbf p)-\hat{\mathcal G}(\mathbf p)\|\le\epsilon,
\qquad \mathbf p=(\lambda,a,\varepsilon,\theta).
\end{equation}
\end{theorem}

This theorem establishes existence without bounding network size,
training cost, or sample complexity. These questions are studied separately
through parametric complexity bounds \cite{LanthalerStuart2026Complexity},
data complexity estimates \cite{KovachkiLanthalerMhaskar2024Data}, and
empirical cost--accuracy comparisons \cite{DeHoopEtAl2022CostAccuracy}.
Here, the result supplies the approximation required in the stability
estimate: a single map, accurate to a prescribed tolerance
across the whole design class.

\section{Closed-loop stability}\label{sec:target}
We first establish exponential stability of the target system.
We then control the nonlinear forcing and use the bounded operators
$T$ and $T^{-1}$ to obtain local rapid stabilization of the plant
under the exact curvature gain. We then extend this result to the case where the 
gain is approximated by perscribing a maximum error tolerance for which
the system remains locally stable. 

\subsection{Stability of the target system}\label{subsec:stability-target}
We study the linear part of the target \eqref{eq:target} together with its
inhomogeneous version
\begin{align}
w_t=(\Ap-a)w+f,\qquad w(0)=w_0\in\LL .\label{eq:target-lin}
\end{align}
The forcing $f$ collects everything beyond the linear design. In the
closed-loop variable $w=Tu$, it is
$f=-T(uu_x)+B_2\ip{\hat h-h}{u}$, where $\hat h$ is the implemented
curvature gain. The first term is the transformed Burgers nonlinearity,
and the second is the effect of approximating the gain. Thus the
forced-target estimates below apply to both the exact-gain analysis in
Section~\ref{ssec:exact-stability} and the gain-error analysis in
Section~\ref{sec:approx-stability}.

We use the decay margin $\nu$ defined in \eqref{eq:nu}.
Using the modal inner product and norms from Section~\ref{ssec:notation}, we define the Lyapunov function by
\begin{align}
V(w)&:=\sum_{n=1}^{\infty}|w_n|^2=\|w\|_{\mathcal H}^2.
\label{eq:lyapunov}
\end{align}
The Riesz bounds \eqref{eq:riesz-const} give
$m\|w\|^2\le V(w)\le M\|w\|^2$.

\begin{proposition}[Exponential stability of forced target system]\label{prop:target-lin}
Let $(a,\varepsilon,\theta)\in\mathscr U_\lambda$, let $\Ap$ be the pre-compensated operator \eqref{eq:Aprime}, let $\nu=a-\max_n\operatorname{Re}\mu_n'$ be as in \eqref{eq:nu}, and consider the forced target system \eqref{eq:target-lin}, $w_t=(\Ap-a)w+f$, $w(0)=w_0\in\LL$. Then:
\begin{enumerate}[label=(\roman*)]
\item $\Ap-a$ generates an analytic semigroup on $\LL$, namely $e^{t(\Ap-a)}w=\sum_ne^{(\mu_n'-a)t}w_n\chi_n$, and for every $w_0\in\LL$ and $f\in L^2_{\rm loc}([0,\infty);\mathcal H^{-2})$ problem \eqref{eq:target-lin} has a unique mild solution, and $w\in C([0,\infty);\LL)\cap L^2_{\rm loc}([0,\infty);\mathcal H^2)$ with $w_t\in L^2_{\rm loc}([0,\infty);\mathcal H^{-2})$.
\item Along this solution, $t\mapsto V(w(t))$ is absolutely continuous and, for every $\nu_0\in(0,\nu)$,
\begin{equation}\label{eq:Vdot}
\begin{aligned}
\frac{d}{dt}V(w)&=-2\sum_{n\ge1}\big(a-\operatorname{Re}\mu_n'\big)|w_n|^2+2\operatorname{Re}\hp{f}{w}\\
&=-2\nu_0\,V(w)-\sum_{n\ge1}\gamma_n|w_n|^2+2\operatorname{Re}\hp{f}{w},
\end{aligned}
\end{equation}
where $\gamma_n$ and $c_\gamma$ are defined in
\eqref{eq:dissipation-weights} and \eqref{eq:dissipation-coercivity}.
In particular, $\gamma_n\ge2(\nu-\nu_0)>0$ and $\gamma_n\ge c_\gamma n^4$.
\item In particular, for $f\equiv0$,
\begin{equation}\label{eq:target-decay}
V(w(t))\le e^{-2\nu t}V(w_0),\qquad \|w(t)\|\le\sqrt{M/m}\;e^{-\nu t}\|w_0\|,\qquad t\ge0 ,
\end{equation}
and, for general $f$ and every $\delta\in(0,1)$,
\begin{equation}\label{eq:target-forced}
\begin{gathered}
\frac{d}{dt}V(w)\le-2\nu_0 V(w)-(1-\delta)\sum_n\gamma_n|w_n|^2+\frac{1}{\delta}\,\|f\|_{\mathcal H^{-2}_\gamma}^2,\\
\|f\|_{\mathcal H^{-2}_\gamma}^2:=\sum_n\frac{|f_n|^2}{\gamma_n}\le\frac{\|f\|_{\mathcal H^{-2}}^2}{c_\gamma} .
\end{gathered}
\end{equation}
\end{enumerate}
\end{proposition}

\begin{proof}
Theorem~\ref{thm:generic}(ii) supplies the Riesz basis and quartic spectral growth,
while Lemma~\ref{lem:high} shows that sufficiently high eigenvalues are
real and negative. We use these properties to prove (i), then establish
the forced energy identity in (ii) and the decay and forcing estimates in (iii).

We first prove (i). In the Riesz coordinates of Theorem~\ref{thm:generic}(ii),
$J^{-1}(\Ap-a)J=\operatorname{diag}(\mu_n'-a)$, and the semigroup is
obtained by shifting \eqref{eq:app-d-display-066}. By \eqref{eq:nu},
$\operatorname{Re}(\mu_n'-a)\le-\nu$. Since only finitely many eigenvalues
are nonreal, Lemma~\ref{lem:high} also gives
$|\arg(a-\mu_n')|\le\pi/2-\delta'$ for some $\delta'>0$.
The diagonal semigroup therefore extends analytically to a sector.

To establish existence and regularity in (i), fix a finite interval
$[0,\tau]$ and let
$f\in L^2(0,\tau;\mathcal H^{-2})$. The variation-of-constants formula
gives, in each coordinate,
\begin{align}
w_n(t)&=\;e^{(\mu_n'-a)t}w_{0,n}
 +\int_0^t e^{(\mu_n'-a)(t-s)}f_n(s)\dd s.\label{eq:app-f-display-005}
\end{align}
Taking the real part of the coordinate equation after multiplication by
$\overline{w_n}$, applying Young's inequality, and integrating gives
\begin{align}
\int_0^\tau(a-\operatorname{Re}\mu_n')|w_n|^2\dd t
&\le\;|w_{0,n}|^2+
\int_0^\tau\frac{|f_n|^2}{a-\operatorname{Re}\mu_n'}\dd t.\label{eq:app-f-display-006}
\end{align}
Theorem~\ref{thm:generic}(ii) and the positive margin \eqref{eq:nu} give
$cn^4\le a-\operatorname{Re}\mu_n'\le Cn^4$. Summing therefore gives
$w\in L^2(0,\tau;\mathcal H^2)$ and
$w_t\in L^2(0,\tau;\mathcal H^{-2})$.
Retaining the terminal energy in the same scalar estimate and summing
only over $n>N$ gives
\begin{align}
\sup_{0\le t\le\tau}\sum_{n>N}|w_n(t)|^2
&\le\;\sum_{n>N}|w_{0,n}|^2
+\int_0^\tau\sum_{n>N}\frac{|f_n(s)|^2}{a-\operatorname{Re}\mu_n'}\dd s
\longrightarrow0\quad(N\to\infty).\label{eq:target-continuous-tail}
\end{align}
Thus the finite modal sums converge in $C([0,\tau];\LL)$ to the mild
solution. Uniqueness follows from \eqref{eq:app-f-display-005}.
Since $\tau$ is arbitrary, this proves (i).

We next prove (ii). Each coordinate is absolutely continuous and satisfies, almost everywhere,
\begin{align}
\frac{\dd}{\dd t}|w_n|^2
&=\;2\operatorname{Re}\big((\mu_n'-a)|w_n|^2+f_n\bar w_n\big).\label{eq:app-f-display-008}
\end{align}
By \eqref{eq:app-f-display-006}, the summed dissipation is integrable.
Moreover, weighted Cauchy--Schwarz gives
\begin{align}
\sum_n|f_n\bar w_n|
&\le\;\|f\|_{\mathcal H^{-2}}\|w\|_{\mathcal H^2}
\in L^1(0,\tau).\label{eq:app-f-display-010}
\end{align}
The integrated identities for the partial sums converge uniformly, and
their derivatives converge in $L^1(0,\tau)$. Thus $V(w)$ is absolutely
continuous and satisfies \eqref{eq:Vdot}.

The weight bounds and the choice of $c_\gamma$ were established in
Section~\ref{ssec:uniform}. Together with \eqref{eq:Vdot}, they complete (ii).

Finally, we prove (iii). For $f=0$, \eqref{eq:app-f-display-005} and
\eqref{eq:nu} give $V(w(t))\le e^{-2\nu t}V(w_0)$.
The norm equivalence \eqref{eq:riesz-const} yields the second estimate
in \eqref{eq:target-decay}. For general $f$, weighted Cauchy--Schwarz
and Young's inequality give
\begin{align}
2\operatorname{Re}\hp fw
&\le\;\delta\sum_n\gamma_n|w_n|^2
 +\delta^{-1}\sum_n\frac{|f_n|^2}{\gamma_n},
\qquad \delta>0.\label{eq:app-f-display-014}
\end{align}

Substituting into \eqref{eq:Vdot} for $0<\delta<1$ gives
\eqref{eq:target-forced}. Its norm bound follows from
$\gamma_n\ge c_\gamma n^4$, completing (iii).
\end{proof}

This proposition establishes exponential decay of the linear target system
and provides the Lyapunov identity used below to control the nonlinear
forcing $f=-T(uu_x)$.

\subsection{Local rapid stabilization of the nonlinear plant with the exact gain}\label{ssec:exact-stability}
For the exact curvature gain, the transformed nonlinear system is
\begin{equation}\label{eq:target-nl}
w_t=(\Ap-a)w-T(uu_x),\qquad w(0)=Tu_0.
\end{equation}

\begin{definition}[Solutions of the closed loop]\label{def:solution}
Let $u_0\in\LL$ and set
$Y_\tau:=C([0,\tau);\LL)\cap L^2_{\rm loc}([0,\tau);\mathcal H^2)$.
A function $u\in Y_\tau$ is a \emph{solution} of the closed loop
\eqref{eq:plant}, \eqref{eq:U1}, \eqref{eq:U2hat} on the interval
$[0,\tau)$ if the function $w:=Tu$ is the mild solution of
\begin{equation}\label{eq:general-mild-solution}
w_t=(\Ap-a)w+f,\qquad w(0)=Tu_0,
\end{equation}
where $f\in L^2_{\rm loc}([0,\tau);\mathcal H^{-2})$ is the corresponding
target forcing.
\end{definition}
For the exact gain, this definition applies with $f=-T(uu_x)$. Proposition~\ref{prop:wp} in Appendix~\ref{app:well-posedness},
applied with $\hat h=h$, supplies local existence, uniqueness, and the
continuation criterion used below.

The identity $\hp{f}{w}=\ip{f}{\Gamma w}$, with $\Gamma$ defined in
\eqref{eq:Gamma}, lets us estimate the nonlinear term using integration
by parts. Lemma~\ref{lem:gamma} gives the Sobolev bounds and zero endpoint
values needed below.

\begin{lemma}[Estimate of the nonlinear forcing term]\label{lem:forcing}
Let $(a,\varepsilon,\theta)\in\mathscr U_\lambda$ and let $T$ be the Fredholm transformation \eqref{eq:transform}. Let $u\in\mathcal H^2$ and $w=Tu$. Then
\begin{equation}
\big|\hp{T(uu_x)}{w}\big|\le C_{\rm N}\,\|u\|^{3/2}\,\|w\|_{\mathcal H^1}^{3/2},\label{eq:nl-est}
\end{equation}
where $C_{\rm N}$ is defined in \eqref{eq:nonlinear-constant} and
satisfies $C_{\rm N}\le C(B,d)$ on $\mathscr U_{B,d}$.
\end{lemma}

\begin{proof}
To estimate the modal pairing in physical coordinates, set
$\phi=\Gamma T^\sharp w$. Lemma~\ref{lem:smoothing}(ii) and
Lemma~\ref{lem:gamma} give
$\phi\in H^1_0(0,1)$.
The definitions of $T^\sharp$ and $\Gamma$, followed by integration by
parts, give
\begin{align}
\hp{T(uu_x)}{w}
&=\;\ip{uu_x}{\phi}
=-\frac12\int_0^1u^2\overline{\phi'}\dd x.\label{eq:app-h-display-009}
\end{align}
The boundary term vanishes because $\phi$ vanishes at both endpoints,
although $u(0)$ need not vanish.
The Gagliardo--Nirenberg inequality \cite[p.~233, Eq.~(42)]{Brezis2011} gives
\begin{align}
|\hp{T(uu_x)}{w}|
&\le\;\tfrac12\|u\|_{L^4}^2\|\phi'\|
\le C_{\rm GN}\|u\|^{3/2}\|u\|_{H^1}^{1/2}\|\phi\|_{H^1}.\label{eq:app-h-display-011}
\end{align}
Since $u=T^{-1}w$, Lemma~\ref{lem:smoothing}(ii) and the Sobolev bounds
established in Lemma~\ref{lem:gamma} give
\begin{align}
\|u\|_{H^1}
&\le\;C_{\rm emb}\|T^{-1}\|_{\mathcal L(\mathcal H^1)}
   \|w\|_{\mathcal H^1},\notag\\
\|\phi\|_{H^1}
&\le\;\|\Gamma\|_{\mathcal L(\mathcal H^1,H^1)}
   \|T^\sharp\|_{\mathcal L(\mathcal H^1)}
   \|w\|_{\mathcal H^1}.\label{eq:app-h-display-012}
\end{align}
Here $C_{\rm emb}$ is the norm of the embedding from $\mathcal H^1$ into
$H^1$. Both $H^1$ factors are therefore controlled by
$\|w\|_{\mathcal H^1}$. Substitution into \eqref{eq:app-h-display-011}
gives \eqref{eq:nl-est} with the constant \eqref{eq:nonlinear-constant}.
The uniform transformation bounds \eqref{eq:uniform-transform-scale}
and the uniform Sobolev bounds in Lemma~\ref{lem:gamma} give
$C_{\rm N}\le C(B,d)$.
\end{proof}

The nonlinear estimate in Lemma~\ref{lem:forcing} gives the following result.

\begin{theorem}[Closed loop is locally exponentially stable]\label{thm:nominal}
Consider the plant \eqref{eq:plant} with the pre-feedback \eqref{eq:U1} and the
exact curvature feedback \eqref{eq:U2}. Let
$(a,\varepsilon,\theta)\in\mathscr U_\lambda$, fix rates
$0<\omega<\bar\nu<\nu$ and a fraction $\varrho\in(0,1)$, and form
$\gamma_n,c_\gamma,\Sigma$ and $C_{\rm N}$ as in
Section~\ref{ssec:uniform}, with $\nu_0:=\bar\nu$. Set
\begin{align}\label{eq:basin}
C_*:=\tfrac58\Big(\tfrac{8\varrho}{3}\Big)^{-3/5}
\Big(\frac{2C_{\rm N}c_2^{3/2}}{m^{3/4}c_\gamma^{3/8}}\Big)^{8/5},\quad
r_0:=\Big(\frac{\bar\nu-\omega}{C_*}\Big)^{5/4},\quad
\rho_0:=\frac{\sqrt{r_0/M}}{c_1}.
\end{align}
Then for every $u_0\in\LL$ with $\|u_0\|\le\rho_0$ the closed loop has a unique
solution on $[0,\infty)$ and
\begin{align}
\|u(t)\|\le c_1c_2\sqrt{M/m}\;e^{-\omega t}\,\|u_0\|,\qquad t\ge0,\label{eq:nominal-est}
\end{align}
together with the dissipation bound
\begin{align}
(1-\varrho)\int_0^\infty e^{2\omega t}\,\Sigma(Tu(t))\dd t\le M c_1^2\|u_0\|^2 .\label{eq:nominal-H2}
\end{align}
\end{theorem}

\begin{proof}
Proposition~\ref{prop:wp} gives a local solution with $\hat h=h$
and its Lyapunov identity. We work on its maximal
interval of existence. Write $w=Tu$, $V=V(w)$, and $\Sigma=\Sigma(w)$,
with $\nu_0=\bar\nu$.

We first estimate the nonlinear contribution. Substituting the exact-gain
forcing from \eqref{eq:target-nl} into \eqref{eq:Vdot} gives
\begin{align}
\dot V&=\;-2\bar\nu V-\Sigma
 -2\operatorname{Re}\hp{T(uu_x)}{w}.\label{eq:Vdot-nominal}
\end{align}
The transformation bound \eqref{eq:transformation-bounds}, norm
equivalence \eqref{eq:riesz-const}, and spectral interpolation with
Proposition~\ref{prop:target-lin}(ii) give
\begin{align}
\|u\|&\le\;c_2m^{-1/2}V^{1/2},
\qquad
\|w\|_{\mathcal H^1}\le V^{1/4}(\Sigma/c_\gamma)^{1/4}.\label{eq:app-j-display-002}
\end{align}
Set $A:=2C_{\rm N}c_2^{3/2}/(m^{3/4}c_\gamma^{3/8})$.
Lemma~\ref{lem:forcing} and Young's inequality
\cite[App.~B.2]{Evans}, with exponents $8/3$ and $8/5$, yield
\begin{align}
2\big|\hp{T(uu_x)}{w}\big|
&\le\;A V^{9/8}\Sigma^{3/8}
\le\varrho\Sigma+C_*V^{9/5},\label{eq:nl-young1}
\end{align}
where $C_*=\tfrac58(\tfrac{8\varrho}{3})^{-3/5}A^{8/5}$,
exactly as in \eqref{eq:basin}. Since $\varrho<1$, positive dissipation
remains. Substitution into \eqref{eq:Vdot-nominal} gives
\begin{align}
\dot V
&\le\;-\big(2\bar\nu-C_*V^{4/5}\big)V-(1-\varrho)\Sigma.\label{eq:app-j-display-005}
\end{align}

We next use the initial-state bound to close the energy estimate.
By \eqref{eq:transformation-bounds}, \eqref{eq:riesz-const}, and the
definition of $\rho_0$ in \eqref{eq:basin},
\begin{align}
V(0)&\le\;Mc_1^2\|u_0\|^2\le r_0.\label{eq:app-j-display-008}
\end{align}
For $V\le r_0$, \eqref{eq:basin} gives
$C_*V^{4/5}\le\bar\nu-\omega$, hence
$2\bar\nu-C_*V^{4/5}\ge\bar\nu+\omega>2\omega$.
Scalar comparison applied to \eqref{eq:app-j-display-005} therefore
preserves $V\le r_0$. Along the solution,
\begin{align}
\dot V&\le\;-2\omega V-(1-\varrho)\Sigma.\label{eq:app-j-display-007}
\end{align}
Multiplying by $e^{2\omega t}$ and integrating gives
\begin{align}
e^{2\omega t}V(t)
+(1-\varrho)\int_0^t e^{2\omega s}\Sigma(w(s))\dd s
&\le\;V(0).\label{eq:app-j-display-010}
\end{align}
In particular, $V(t)\le e^{-2\omega t}V(0)$, so
\begin{align}
\|u(t)\|
&\le\;c_2m^{-1/2}e^{-\omega t}V(0)^{1/2}
\le c_1c_2\sqrt{M/m}\,e^{-\omega t}\|u_0\|.\label{eq:app-j-display-009}
\end{align}
This bound and the blow-up alternative in Proposition~\ref{prop:wp}
imply global existence and prove \eqref{eq:nominal-est}.

Finally, discarding the nonnegative terminal energy in
\eqref{eq:app-j-display-010} and letting $t\to\infty$ gives
\eqref{eq:nominal-H2}, using \eqref{eq:app-j-display-008}.
The physical $H^2$ estimate stated after the theorem follows from
Lemma~\ref{lem:smoothing}(ii), the embedding established in
Lemma~\ref{lem:gamma}, and Proposition~\ref{prop:target-lin}(ii):
\begin{align}
\|u\|_{H^2}
&\le\;C\|T^{-1}\|_{\mathcal L(\mathcal H^2)}
 (\Sigma/c_\gamma)^{1/2}.\label{eq:app-j-display-011}
\end{align}
\end{proof}

Theorem~\ref{thm:nominal}, applied with $\bar\nu\in(\omega,\nu)$ proves the closed-loop assertion of Theorem~\ref{thm:intro} for $\hat h=h$. Moreover, since $\Sigma\ge c_\gamma\|Tu\|_{\mathcal H^2}^2$, \eqref{eq:nominal-H2} also gives
$\int_0^\infty e^{2\omega t}\|u(t)\|_{H^2}^2\dd t\le C\|u_0\|^2$,
where $C$ depends on $\lambda$, $(a,\varepsilon,\theta)$, $\bar\nu$, and
$\varrho$, but not on $u_0$. This estimate bounds the remaining dissipation
after controlling the nonlinear term and will be used to absorb
the additional forcing due to gain approximation.

\subsection{Local rapid stabilization under gain approximation}\label{sec:approx-stability}\label{ssec:inexact-stability}
We now quantify the curvature-gain accuracy required to preserve stability.
For an admissible design $\mathbf p=(\lambda,a,\varepsilon,\theta)$, the
implemented curvature feedback is
\begin{equation}\label{eq:U2hat}
\hat U_2(t)=\ip{\hat h}{u(\cdot,t)},\qquad \hat h=\hat{\mathcal G}(\mathbf p),
\end{equation}
applied to the plant \eqref{eq:plant} with pre-feedback \eqref{eq:U1}.

Writing $e:=\hat h-h$ for the gain error, the exact transformation
$w=Tu$ gives
\begin{equation}\label{eq:target-full}
\begin{gathered}
w_t=(\Ap-a)w+f,\qquad w(\cdot,0)=Tu_0, \qquad
f:=B_2\ip{e}{u}-T(uu_x),\qquad u=T^{-1}w.
\end{gathered}
\end{equation}
Since $TB_2=B_2$, gain approximation adds the boundary forcing
$B_2\ip{e}{u}$ to the exact-gain dynamics \eqref{eq:target-nl}.
Hence, we use Definition~\ref{def:solution} to define solutions with the forcing  $f = B_2\ip{e}{u}-T(uu_x)$.

The nonlinear term is already bounded by Lemma~\ref{lem:forcing}.
It remains to bound the gain error as follows.

\begin{lemma}[Gain-error bound]\label{lem:gain-error}
For $u\in\mathcal H^2$, $w=Tu$, and $e\in\LL$,
\begin{align}
\big|\hp{B_2\ip{e}{u}}{w}\big|
&\le\Lambda\,\|e\|\,\|u\|\,\Sigma(w)^{1/2}.
\label{eq:err-est}
\end{align}
where $\Sigma$ and $\Lambda$ are defined in
\eqref{eq:dissipation-constants} and \eqref{eq:Lambda}, respectively.
\end{lemma}

\begin{proof}
Since $w\in\mathcal H^2$ by Lemma~\ref{lem:smoothing}(ii), weighted
Cauchy--Schwarz gives
\begin{align}
|\hp{B_2}{w}|
&=\;\left|\sum_n b_n'\bar w_n\right|
\le\left(\sum_n\frac{|b_n'|^2}{\gamma_n}\right)^{1/2}
\Sigma(w)^{1/2}=\Lambda\Sigma(w)^{1/2}.\label{eq:app-h-display-006}
\end{align}
This also shows that the modal dual pairing is well-defined.
Multiplying by $|\ip{e}{u}|\le\|e\|\|u\|$ proves \eqref{eq:err-est}.
\end{proof}

Combining this bound with the nonlinear estimate gives an explicit
gain-error tolerance and initial-state radius.

\subsubsection{Local exponential stability under gain error}\label{sec:main}

\begin{theorem}[Closed loop is locally exponentially stable under gain error]\label{thm:main}
Consider the Kuramoto--Sivashinsky plant \eqref{eq:plant} with the pre-feedback \eqref{eq:U1} and the perturbed curvature feedback \eqref{eq:U2hat}. Let $(a,\varepsilon,\theta)\in\mathscr U_\lambda$, let $h$ be the exact curvature gain of Theorem~\ref{thm:kernel}, and let $\nu=a-\max_n\operatorname{Re}\mu_n'$ be the rate of the target system. Fix a target rate $\omega$ and an auxiliary rate $\bar\nu$ with $0<\omega<\bar\nu<\nu$, and fractions $\varkappa_e,\varkappa_{\rm N}\in(0,1)$ with $\varkappa_e+\varkappa_{\rm N}<1$. Let $m,M$ be the Riesz constants \eqref{eq:riesz-const}, $c_1,c_2$ the norms of $T,T^{-1}$ on $\LL$, and $\gamma_n,c_\gamma,\Sigma,\Lambda,C_{\rm N}$ as defined in Section~\ref{ssec:uniform} with $\nu_0:=\bar\nu$. Set
\begin{equation}\label{eq:tolerance}
\epsilon^*:=\frac{\sqrt{\varkappa_e\,m\,(\bar\nu-\omega)}}{c_2\,\Lambda},\qquad
r_0:=\Big(\frac{\bar\nu-\omega}{C_*}\Big)^{5/4},\qquad \rho_0:=\frac{\sqrt{r_0/M}}{c_1},\qquad 
C_*:=\tfrac58\Big(\tfrac{8\varkappa_{\rm N}}{3}\Big)^{-3/5}\Big(\frac{2C_{\rm N}c_2^{3/2}}{m^{3/4}c_\gamma^{3/8}}\Big)^{8/5}.
\end{equation}
Let $\hat h=\hat{\mathcal G}(\lambda,a,\varepsilon,\theta)$ be the implemented gain, with $\|\hat h-h\|\le\epsilon^*$. Then for every $u_0\in\LL$ with $\|u_0\|\le\rho_0$ the closed loop \eqref{eq:plant}, \eqref{eq:U1}, \eqref{eq:U2hat} has a unique solution on $[0,\infty)$ in the sense of Definition~\ref{def:solution}, and
\begin{equation}\label{eq:main-est}
\|u(t)\|\le c_1c_2\sqrt{M/m}\;e^{-\omega t}\,\|u_0\|,\qquad t\ge0,
\end{equation}
and
\begin{equation}\label{eq:main-H2}
(1-\varkappa_e-\varkappa_{\rm N})\int_0^\infty e^{2\omega t}\,\Sigma(Tu(t))\,dt\le Mc_1^2\|u_0\|^2 .
\end{equation}
\end{theorem}

\begin{proof}
Use the local solution and Lyapunov identity from
Proposition~\ref{prop:wp}, with $\nu_0=\bar\nu$.
Writing $w=Tu$, $V=V(w)$, and $\Sigma=\Sigma(w)$, the additional
boundary forcing changes \eqref{eq:Vdot-nominal} to
\begin{align}
\dot V&=\;-2\bar\nu V-\Sigma
 +2\operatorname{Re}\hp{B_2\ip{e}{u}}{w}
 -2\operatorname{Re}\hp{T(uu_x)}{w}.\label{eq:Vdot-full}
\end{align}
The nonlinear estimate \eqref{eq:nl-young1} applies with
$\varrho=\varkappa_{\rm N}$, giving the constant $C_*$ in
\eqref{eq:tolerance}. By $\|e\|\le\epsilon^*$,
Lemma~\ref{lem:gain-error}, and the plant-norm estimate
\eqref{eq:app-j-display-002}, Young's inequality with exponents $2,2$ gives
\begin{align}
2\big|\hp{B_2\ip{e}{u}}{w}\big|
&\le\;2\epsilon^*\Lambda c_2m^{-1/2}V^{1/2}\Sigma^{1/2}\notag\\
&\le\;\varkappa_e\Sigma+(\bar\nu-\omega)V.\label{eq:err-young}
\end{align}
Here $(\epsilon^*)^2\Lambda^2c_2^2/(\varkappa_e m)=\bar\nu-\omega$
by \eqref{eq:tolerance}.
Combining the two forcing estimates leaves the positive dissipation
fraction $1-\varkappa_e-\varkappa_{\rm N}$ and gives
\begin{align}
\dot V
&\le\;-\big(\bar\nu+\omega-C_*V^{4/5}\big)V
 -(1-\varkappa_e-\varkappa_{\rm N})\Sigma.\label{eq:app-i-display-015}
\end{align}

As in the proof of Theorem~\ref{thm:nominal}, the initial bound gives
$V(0)\le r_0$, and $C_*V^{4/5}\le\bar\nu-\omega$ while $V\le r_0$.
The same comparison argument preserves this bound. Multiplying by
$e^{2\omega t}$ and integrating yields
\begin{align}
e^{2\omega t}V(t)
+(1-\varkappa_e-\varkappa_{\rm N})
\int_0^t e^{2\omega s}\Sigma(w(s))\dd s&\le\;V(0).
\label{eq:app-i-display-020}
\end{align}
The norm conversion and continuation argument of Theorem~\ref{thm:nominal}
now give global existence and \eqref{eq:main-est}. Letting $t\to\infty$
and using $V(0)\le Mc_1^2\|u_0\|^2$ gives \eqref{eq:main-H2}.
The physical $H^2$ estimate follows from \eqref{eq:app-j-display-011}.
\end{proof}

The tolerance scales with $\sqrt{\bar\nu-\omega}$ for fixed design and
$\bar\nu$. It applies to any approximate gain, including gains computed
by numerical schemes. Applied with $\bar\nu\in(\omega,\nu)$ to a fixed admissible design, Theorem~\ref{thm:main} proves the approximate-gain assertion of Theorem~\ref{thm:intro}. Applied uniformly over $\mathscr U_{B,d}$ with $\bar\nu=d/2$, $\omega=d/4$, and
$\varkappa_e=\varkappa_{\rm N}=1/4$, it proves Theorem~\ref{thm:intro-NO}. In particular, the uniform bounds of Theorem~\ref{thm:lipschitz} then give positive
constants $\epsilon^*_{B,d},\rho_{B,d},C_{B,d}$, depending only on $B,d$,
such that
\begin{equation}\label{eq:uniform-stability-choices}
\epsilon^*\ge\epsilon^*_{B,d},\qquad
\rho_0\ge\rho_{B,d},\qquad c_1c_2\sqrt{M/m}\le C_{B,d}
\end{equation}
throughout the class. Hence every map with uniform gain error at most
$\epsilon^*_{B,d}$ and every initial state with $\|u_0\|\le\rho_{B,d}$
satisfy the hypotheses of Theorem~\ref{thm:main}, which gives the unique
global solution and \eqref{eq:intro-NO-decay}.

\section{Numerical examples}\label{sec:numerics}
We first illustrate prescribed decay rates for a plant with an unstable
double eigenvalue. We then define a family of coefficients, train a neural
operator approximation of the curvature gain, and evaluate its nonlinear
closed-loop performance and computational cost.

\subsection{Prescribed rate stabilization with the modal gain}\label{sec:numerical}

\begin{figure}[!htbp]
\centering
\includegraphics[width=0.85\textwidth]{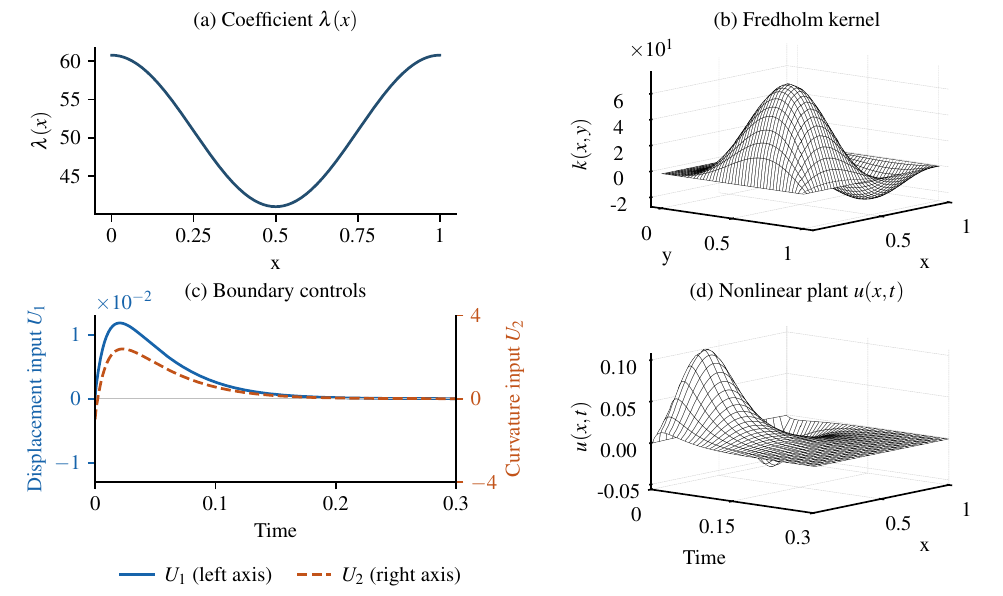}
\caption{Stabilization of the nonlinear plant with coefficient
\eqref{eq:numerical-c}, initial condition
$u_0(x)=10^{-3}\sqrt{2}\sin(\pi x)$, and
$(a,\theta,\varepsilon)=(480,1,0.5)$.
(a) The coefficient $\lambda(x)$. (b) The Fredholm kernel computed with
$32$ modes. (c) The two boundary controls. (d) The nonlinear plant state.}
\label{fig:numerical}
\end{figure}
We illustrate the two-input design on a plant with a repeated positive
eigenvalue. Consider \eqref{eq:plant} with the spatially varying coefficient
\begin{align}
\lambda(x)&=\pi^2\bigl(c_*+\cos(2\pi x)\bigr), \qquad c_*\approx5.156.\label{eq:numerical-c}
\end{align}
Here $c_*$ is selected so that the two largest eigenvalues coincide:
$\mu_1=\mu_2=\mu_*\approx459.114$.
By Proposition~\ref{prop:obstruction}, this unstable double eigenvalue
precludes exponential stabilization through the curvature input $U_2$
alone. This example therefore illustrates the role of the displacement
pre-feedback $U_1$ in the two-input construction.

For this experiment, we choose the admissible design parameters
$(a,\theta,\varepsilon)=(480,1,0.5)$ and the initial condition
$u_0(x)=10^{-3}\sqrt{2}\sin(\pi x)$, with
$\|u_0\|_{L^2(0,1)}=10^{-3}$.
The gain \eqref{eq:intro-gain} is computed by truncating the modal
kernel construction to $32$ kernel modes.
Figure~\ref{fig:numerical} shows the coefficient, the computed Fredholm
kernel, and the resulting boundary controls and nonlinear plant response.
Following an initial transient, the state and both controls decay toward
zero for this initial condition.

\begin{figure}[!htbp]
\centering
\includegraphics[width=0.85\textwidth]{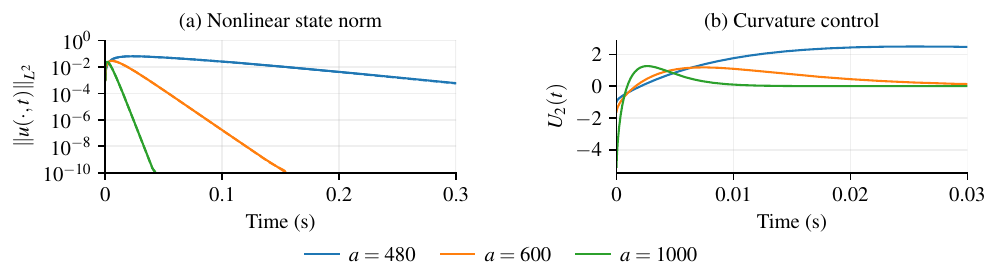}
\caption{Selectable decay rates for the fixed plant
\eqref{eq:numerical-c}, with $(\theta,\varepsilon)=(1,0.5)$ and
$u_0(x)=10^{-3}\sqrt{2}\sin(\pi x)$.
(a) Nonlinear state norms for $a=480,600,1000$.
(b) The corresponding curvature controls, with the
first $0.03$ seconds displayed to resolve the transient.
Each curvature gain is recomputed using $1536$ modes.}
\label{fig:rate-study}
\end{figure}

\begin{figure}[!b]
\centering
\begin{minipage}[t]{0.48\textwidth}
\vspace{0pt}
\centering
\includegraphics[width=\linewidth]{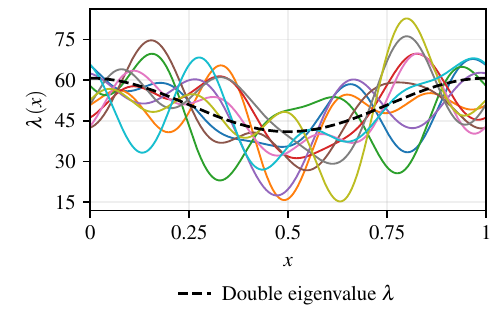}
\caption{Ten randomly selected training coefficients from the
Fourier family \eqref{eq:data-family}. The dashed
curve is the coefficient \eqref{eq:numerical-c} with an unstable double eigenvalue.}
\label{fig:part2-coefficients}
\end{minipage}\hfill
\begin{minipage}[t]{0.50\textwidth}
\vspace{0pt}
\setlength{\tabcolsep}{3pt}
\renewcommand{\arraystretch}{1.45}
\centering
\normalsize
\setlength{\captionindent}{0pt}
\begin{tabular}{lrr}
\hline
Method & Time (ms) & Error (\%)\\
\hline
Modal ($N=4$) & $0.229\pm0.007$ & 26.388 \\
Modal ($N=64$) & $1.143\pm0.019$ & 3.680 \\
Modal ($N=128$) & $4.137\pm0.022$ & 2.417 \\
Modal ($N=512$) & $152.894\pm2.649$ & 0.885 \\
Modal ($N=1024$) & $814.620\pm5.503$ & 0.000 \\
FNO (CPU) & $25.925\pm0.467$ & 0.128 \\
FNO (GPU) & $2.622\pm0.023$ & 0.129 \\
\hline
\end{tabular}
\par\smallskip
\captionof{table}{\normalfont Gain-generation time (mean $\pm$ standard deviation)
and mean relative $L^2$ error over $200$ test profiles, using the
$1024$-mode reference. Modal timings use one CPU thread.}
\label{tab:part2-timing}
\end{minipage}
\end{figure}

\begin{figure}[!t]
\centering
\includegraphics[width=0.9\textwidth]{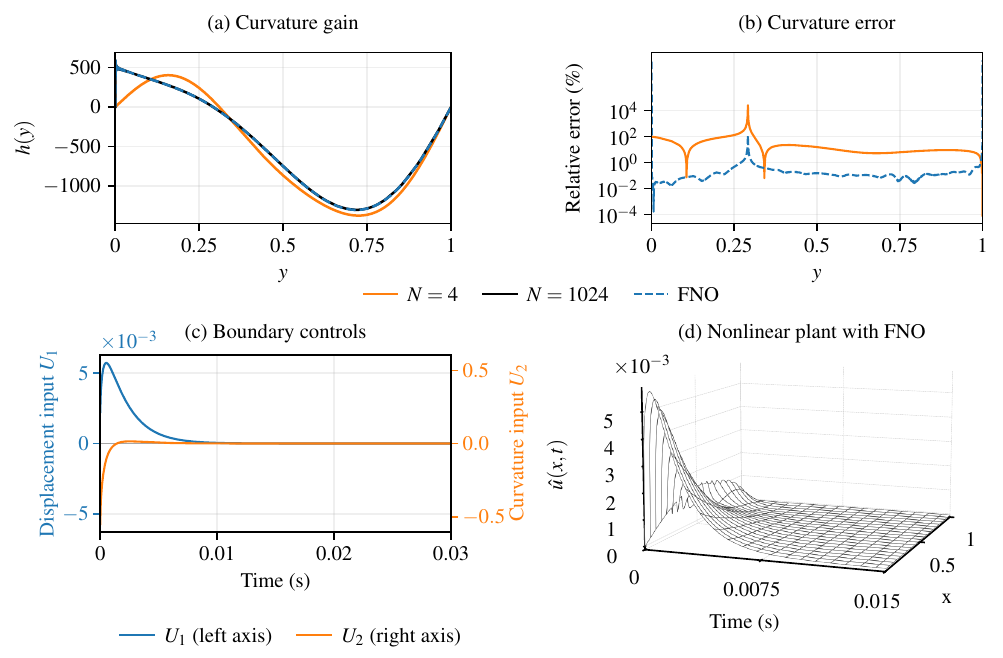}
\caption{Gain approximation and nonlinear stabilization for the coefficient \eqref{eq:numerical-c}, with $(a,\theta,\varepsilon)=(1000,1,4)$ and
$u_0(x)=10^{-3}\sqrt{2}\sin(\pi x)$.
(a) Gains from $N=4,1024$ and FNO.
(b) Relative pointwise error $|\hat h(y)-h_{\rm ref}(y)|/(|h_{\rm ref}(y)|+10^{-6})$,
expressed as a percentage.
The FNO relative $L^2$ error is $0.1107\%$.
(c) FNO controls $U_1$ (blue, left axis) and $U_2$ (orange, right axis).
(d) Nonlinear state under neural feedback up to $t=0.015$.}
\label{fig:part2-results}
\end{figure}

We next examine how the target shift $a$ sets the decay rate. For an
admissible design, Theorem~\ref{thm:intro} guarantees local exponential
stability at any rate $0<\omega<\nu$, where
$\nu=a-\max_n\operatorname{Re}\mu_n'$ is the target spectral margin
defined in \eqref{eq:nu}, also for every gain $\hat h$ with $\|\hat h-h\|\le\epsilon^*$.
A truncated modal gain is covered provided its $L^2$ error is below $\epsilon^*$; this
inequality is not verified here, so the simulations below are empirical.
With the plant and pre-feedback fixed, increasing $a$ increases this
margin, while the sufficient initial-state radius depends on the design
and chosen rate. Keeping $\lambda$, $(\theta,\varepsilon)=(1,0.5)$, and $u_0$ fixed, we
compare $a=480,600,1000$ in Figure~\ref{fig:rate-study}.
For this comparison, each curvature gain is recomputed using $1536$
modes. Panel~(a) shows faster decay after the initial transient as $a$
increases. Panel~(b) shows the corresponding curvature-control
transients over the first $0.03$ seconds. Note that, despite an increase in 
the gain, the control magnitudes do not necessarily increase and yet
rapid stabilization is attained. 
\subsection{Stabilization with neural operator approximated gain}
We now present experiments using a gain approximation. We first specify the
set of design classes for which will we train a neural operator.

\subsubsection{Dataset construction}

We consider smooth spatially varying coefficients generated by
\begin{align}
\lambda(x)=\pi^2\Bigl(c+\sum_{\ell=1}^{4}
\bigl[\alpha_\ell\cos(2\pi\ell x)+\beta_\ell\sin(2\pi\ell x)\bigr]\Bigr).
\label{eq:data-family}
\end{align}
The nine Fourier parameters are sampled independently and uniformly from
\begin{align}
c&\in[4.8,5.3],\qquad \alpha_1\in[0.8,1.2], \qquad 
\beta_1\in[-0.2,0.2],\qquad
\alpha_\ell,\beta_\ell\in[-1,1],\qquad \ell=2,3,4. \notag
\end{align}
The higher harmonics allow substantial changes in spatial shape, as
illustrated by the ten randomly selected training profiles in
Fig.~\ref{fig:part2-coefficients}.
These bounded Fourier coefficients vary the mean, harmonic amplitudes,
and phases but remain a uniformly bounded
$W^{2,\infty}$ family. 

We fix $(a,\theta,\varepsilon)=(1000,1,4)$ to learn the dependence on
$\lambda$, improving conditioning and decay relative to Section~\ref{sec:numerical}.
We generate $1800$ gain labels using at least $1024$ modes on a uniform
$3073$-point grid, split into $1400$ training, $200$ validation, and $200$
test samples.\newline
Dataset generation takes approximately $226$ seconds with a standard CPU.

\subsubsection{Neural operator training}

We train a Fourier neural operator (FNO) \cite{LiEtAl2021FNO} for the map
$\lambda\mapsto h$ at the fixed design parameters above. Its inputs are
samples of $\lambda$ and the spatial coordinates and its output is the
curvature gain $\hat h(y)$ used in \eqref{eq:U2hat}. To improve numerical conditioning during training, we center and rescale
the input coefficients using the mean and variance of the training data.
The same rescaling is used for validation and testing. This normalization is standard 
in machine learning \cite{Bishop2006}.

We minimize the squared $L^2(0,1)$ error between $\hat h$ and the modal
label $h_N$. The FNO has four $64$-neuron layers, $32$ Fourier modes, and
$329281$ trainable parameters. We train for $400$ epochs using AdamW
\cite{LoshchilovHutter2019}, initial learning rate $10^{-4}$, batch size
$32$, and gradient clipping at norm $1$, selecting the checkpoint with
lowest validation loss. Training takes $406$ seconds on an RTX~3090~Ti.

Against the reference gain $h_{\rm ref}=h_{1024}$, the median relative
$L^2$ errors are $0.0916\%$ on the training set and $0.1041\%$ on the test
set.
The closed-loop tests below show that this accuracy is sufficient for
stabilization.

\subsubsection{Neural operator closed-loop stability}

\emph{Nonlinear closed-loop stability.}
We simulate the nonlinear closed loop for all $200$ held-out test profiles
using the FNO and each modal approximation in Table~\ref{tab:part2-timing}.
The simulations use $384$ plant modes and the initial condition
$u_0(x)=10^{-3}\sqrt{2}\sin(\pi x)$, with $\|u_0\|_{L^2(0,1)}=10^{-3}$.
Every method produces decay over the simulation horizon $0\le t\le0.8$,
with $\|u(0.8)\|_{L^2}<\|u_0\|_{L^2}$ in all $200$ cases. The largest
state norm under FNO feedback is $3.30\times10^{-3}$.

Figure~\ref{fig:part2-results}(c)--(d) illustrates the nonlinear response
for the coefficient \eqref{eq:numerical-c} with an unstable double eigenvalue, using
$384$ plant modes and the same initial condition. Panel (c) shows the
displacement and curvature inputs, and panel (d) shows the state under
FNO feedback. The FNO and $1024$-mode reference gain produce nearly
identical decay, and all five modal approximations also stabilize this
example. These simulations demonstrate the effectiveness of the learned
gain on the held-out cases.

\emph{Computation time.}
Table~\ref{tab:part2-timing} compares gain-generation time and relative
error over $200$ test profiles on the $3073$-point grid. Modal timings
include assembly, spectral computation, the coefficient solve, and gain
reconstruction. FNO timings include normalization and GPU transfers for
one profile. Offline training and dataset generation are excluded.
Timings use a Ryzen~9~5900X CPU (one thread) or RTX~3090~Ti.

Low-mode approximations are inexpensive but have larger gain errors.
The FNO achieves approximately $0.13\%$ mean relative error in $25.9$~ms
on the CPU or $2.62$~ms on the GPU, compared with $815$~ms for the
$1024$-mode solve. For a fixed plant, the gain is computed once and reused.
The FNO's speed is particularly valuable in adaptive control, where gains
are updated as the coefficient estimate changes
\cite{BhanShiKrstic2025}.

\section{Conclusion}\label{sec:conclusion}

We have given a Fredholm backstepping design for the KS equation with a
spatially varying anti-diffusion coefficient. It is valid for every
$\lambda\in W^{1,\infty}(0,1)$, with no further restriction on the plant.
What makes this possible is the second boundary input. At an unstable
repeated eigenvalue of the plant operator, no linear feedback through a
single channel can stabilize the system
(Proposition~\ref{prop:obstruction}). The displacement input is spent
entirely on a controllability-ensuring pre-feedback, in the sense of
Heymann's Lemma, so that every remaining restriction falls on the design
parameters $(a,\varepsilon,\theta)$ and none on the plant
(Theorem~\ref{thm:generic}). We then show that the Fredholm transformation is invertible
(Theorem~\ref{thm:kernel}(ii)), and its invertibility is a consequence of
controllability of the pre-compensated plant.

Under this two-input design, the Fredholm gain is still challenging to implement numerically. Consequently, we show that the operator mapping the design
parameters and anti-diffusion coefficient $\lambda$ to the gain $h$ is
continuous (Theorem~\ref{thm:lipschitz}). This yields the existence of a
single neural operator approximation of the gain to any prescribed $L^2$
accuracy on a compact admissible design class (Theorem~\ref{thm:NO}).
We then establish stability for both the exact gain and sufficiently
accurate neural-operator approximations
(Theorems~\ref{thm:nominal} and~\ref{thm:main}). Explicitly, we show that local rapid
stabilization is achievable in both cases
(Theorems~\ref{thm:intro} and~\ref{thm:intro-NO}). We conclude with
numerical results illustrating prescribed decay rates and stabilization
using the learned gain.

Hence, this work opens the door to several significant extensions. First,
it opens the way to adaptive and observer-based control via Fredholm
backstepping, a completely unexplored topic. Second, the idea of using two
controls can apply to numerous problems beyond the KS equation, as already
shown in the companion papers \cite{Krstic2026NavierStokesHeymann,Krstic2026YihInterface}. By making controllability part
of the feedback design, this approach extends the reach of Fredholm
backstepping to plants previously excluded by their spectral structure,
enabling new directions in PDE backstepping.

\begin{acknowledgement}
\textbf{Acknowledgment.}
The authors' problems, ideas, and results were developed with the assistance
of Claude and ChatGPT in final theorem formulation, proofs, simulations,
and drafting throughout the paper, under the authors' correction and
complete verification.
\end{acknowledgement}

\begin{appendices}
\renewcommand{\sectionname}{\appendixname}
\renewcommand{\theequation}{\thesection\arabic{equation}}

\section{Additional proofs for Section~\ref{sec:kernel}: verification of the kernel and modal equations}\label{app:kernel-equations}

\subsection{Proof of Proposition~\ref{prop:kernel}: verification of the plant-to-target transformation}\label{app:kernel-transformation-proof}

\begin{proof}
Write $L_y=\partial_y^4+\partial_y(\lambda(y)\partial_y)$.
Integration by parts, using the plant and kernel boundary conditions, gives
\begin{align}
\int_0^1k(x,y)L_yu(y,t)\dd y
&=\;\big\langle L_yk(x,\cdot),u(\cdot,t)\big\rangle
  +k_y(x,0)U_2(t)  +\big(k_{yyy}(x,0)+\lambda(0)k_y(x,0)\big)U_1(t).\label{eq:app-c-display-001}
\end{align}
The first term on the right is the distributional pairing in $y$.
Since $k_y(x,0)=0$, only the displacement boundary term remains.
Substituting $U_1=\varepsilon\ip{g_\theta}{u}$, this term cancels the
kernel source $\varepsilon k_{yyy}(x,0)g_\theta(y)$ after integration
against $u$, giving the target equation.

At the left endpoint, \eqref{eq:kpde-e} and $w=u-Ku$ give
\begin{align}
w(0)&=\;\varepsilon\ip{g_\theta}{u}-\ip{k(0,\cdot)}{u}
=\varepsilon\ip{g_\theta}{u-Ku}
=\varepsilon\ip{g_\theta}{w}.\label{eq:app-c-display-004}
\end{align}
The remaining target boundary conditions follow from the kernel traces.
\end{proof}

\subsection{Proof of Proposition~\ref{prop:modal}: verification of the modal equivalence of the kernel}\label{app:modal-equivalence-proof}

\begin{proof}
We first project in the $y$ variable and then in the $x$ variable.
The section $k(x,\cdot)$ need not belong to $D(\Aop)$ because its third
derivative jumps at $y=x$. The first projection must therefore use the
distributional kernel equation, rather than an operator-domain identity.
Define
\begin{align}
k_j(x)&:=\;\ip{k(x,\cdot)}{\varphi_j},
\qquad \ell(x):=k_{yyy}(x,0).\label{eq:app-c-display-005}
\end{align}
Integration by parts has no endpoint terms because both functions and
their second derivatives vanish at $y=0,1$, by \eqref{eq:kpde-b} and
$\varphi_j\in D(\Aop)$. Since $-L_y\varphi_j=\mu_j\varphi_j$,
\begin{align}
\ip{-L_yk(x,\cdot)}{\varphi_j}
&=\;\ip{k(x,\cdot)}{-L_y\varphi_j}
=\mu_jk_j(x).\label{eq:app-c-display-006}
\end{align}
Here $L_yk$ is interpreted distributionally, so it includes the terms
generated by derivative jumps at $y=x$. Thus the identity does not
require $k(x,\cdot)\in D(\Aop)$.
Hence \eqref{eq:kpde-a} becomes
\begin{align}
k_j''''+(\lambda k_j')'+(\mu_j+a)k_j
= a\varphi_j+\varepsilon\ell g_j,\qquad
\ell(x)=\sum_i k_i(x)\alpha_i.\label{eq:kj}
\end{align}
The second identity uses $k_y(x,0)=0$, which also gives
$\sum_i k_i(x)\beta_i=0$ and hence \eqref{eq:modal-b} after projection in $x$.

For the remaining modal equation, the boundary data for $k_j$ are
\begin{align}
k_j(1)&=\;k_j''(1)=0,
\qquad k_j(0)=\varepsilon\sum_pK_{pj}g_p,
\qquad k_j''(0)=h_j.\label{eq:app-c-display-008}
\end{align}
Pairing \eqref{eq:kj} with $\varphi_m$ and applying the boundary identity
of Lemma~\ref{lem:ibp} gives \eqref{eq:modal-a}.
\end{proof}

\section{Sobolev bounds for the modal coordinate map}\label{app:modal-sobolev}

Here generic constants may depend on $\lambda$, $(a,\varepsilon,\theta)$, and $\nu_0$, but not on the state, mode indices, or time. Bounds $C(B,d)$ are uniform over $\mathscr U_{B,d}$ for the stated common rate.

We first establish the Sobolev bounds and boundary properties of $\Gamma$
needed in the proof of Lemma~\ref{lem:forcing}. These properties justify
integration by parts in the modal pairing. Their uniform bounds are also
used in Theorem~\ref{thm:lipschitz}.

\begin{lemma}[The operator $\Gamma$]\label{lem:gamma}
Let $(a,\varepsilon,\theta)\in\mathscr U_\lambda$ and let $\Gamma$ be the map \eqref{eq:Gamma} relating the modal inner product of the eigenbasis of $\Ap$ to the $L^2$ inner product. Then $\Gamma=\tilde JJ^{-1}$ is an isomorphism of $\LL$ with $\|\Gamma\|\le M$ and $\|\Gamma^{-1}\|\le m^{-1}$. It is bounded from $\mathcal H^s$ into $H^s(0,1)$ for every $s\in[0,4]$, with 
\begin{align}
\Gamma(\mathcal H^s)&\subset\;\{v\in H^s(0,1):v(0)=v(1)=0\}\label{eq:app-h-display-001}
\end{align}
 for $\frac12<s\le4$, in particular $\Gamma(\mathcal H^1)\subset H^1_0(0,1)$. On $\mathscr U_{B,d}$ the norms $\|\Gamma\|_{\mathcal L(\mathcal H^s,H^s)}$ are bounded by $C(B,d)$.
\end{lemma}

\begin{proof}[Proof of Lemma~\ref{lem:gamma}]
The coordinate maps give the $L^2$ bounds directly. Indeed,
\begin{align}
\Gamma J&=\;\tilde J,
\qquad \Gamma=\tilde JJ^{-1},
\qquad \tilde J=(J^{-1})^*.\label{eq:app-h-display-002}
\end{align}
Consequently,
\begin{align}
\|\Gamma\|&\le\;\|J^{-1}\|^2=M,
\qquad \|\Gamma^{-1}\|\le\|J\|^2=m^{-1}.\label{eq:app-h-display-003}
\end{align}

For Sobolev regularity, we use graph-norm estimates for $\Ap$ and
its adjoint. The following constants depend on $\lambda$ and
$(a,\varepsilon,\theta)$ for a fixed design and are uniform on
$\mathscr U_{B,d}$ as indicated.
For $\phi\in D(\Ap)$, the equation and boundary condition give
\begin{align}
\phi''''&=\;-\Ap\phi-(\lambda\phi')',
\qquad |\phi(0)|=|\varepsilon\ip{g_\theta}{\phi}|
\le C(B)\|\phi\|.\label{eq:app-g-display-030}
\end{align}
The elliptic estimate, with interpolation to absorb the $H^2$ term, gives
\begin{align}
\|\phi\|_{H^4}&\le\;C(B)(\|\Ap\phi\|+\|\phi\|).\label{eq:app-g-display-031}
\end{align}
For the adjoint, the additional trace term is absorbed using
\begin{align}
|\qd\psi(0)|&\le\;\delta\|\psi\|_{H^4}
 +C_{\delta,B}\|\psi\|,\label{eq:app-g-display-032}
\end{align}
with $\delta$ small. This is uniform because
$|\varepsilon|\|g_\theta\|$ is bounded on the class. Hence both operators
have uniform equivalence of graph and $H^4$ norms. For $\phi\in\mathcal H^4$,
the series
\begin{align}
\Gamma\phi&=\;\sum_n\phi_n\tilde\chi_n,
\qquad \sum_n\bar\mu_n'\phi_n\tilde\chi_n
\label{eq:app-h-display-004}
\end{align}
converge in $\LL$ by the spectral growth and Riesz bounds.
Closedness of $\Ap{}^*$ puts the first series in its domain.
The uniform graph-norm estimate then gives
\begin{align}
\|\Gamma\|_{\mathcal L(\mathcal H^4,H^4)}&\le\;C(B,d).\label{eq:app-g-display-033}
\end{align}
Interpolating with the $L^2$ bound above gives the uniform bound for every
$s\in[0,4]$, and the same argument applies to the spectral embedding.
Each finite adjoint eigenfunction sum vanishes at both endpoints.
Convergence in $H^s$ and continuity of the trace for $s>1/2$ pass these
boundary conditions to the limit.
\end{proof}

\section{Closed-loop well-posedness}\label{app:well-posedness}

Throughout this appendix, fix an admissible design. Generic constants may
depend on $\lambda$, $(a,\varepsilon,\theta)$, and $\nu_0$, but not on
the state, mode indices, or time.

Proposition~\ref{prop:target-lin} treats a prescribed forcing $f$.
For the closed loop, $f$ depends on the unknown state, so existence and
uniqueness require a fixed-point argument. The following result uses
the general forced-target estimate together with
Lemmas~\ref{lem:forcing} and~\ref{lem:gain-error}.

\begin{proposition}[Local existence and uniqueness for closed-loop plant]\label{prop:wp}
Consider the Kuramoto--Sivashinsky plant \eqref{eq:plant} with the pre-feedback \eqref{eq:U1} and the perturbed curvature feedback \eqref{eq:U2hat}, that is $U_2=\ip{\hat h}{u}$. For every $u_0\in\LL$ and every $\hat h\in\LL$ there is $\tau>0$, depending only on $\|u_0\|$, $\|\hat h-h\|$ and the design, and a unique solution of the closed loop on $[0,\tau)$ in the sense of Definition~\ref{def:solution}. Any two solutions coincide on their common interval of existence. If $[0,\tau^*)$ is the maximal interval of existence and $\tau^*<\infty$, then $\|u(t)\|\to\infty$ as $t\uparrow\tau^*$. Moreover $t\mapsto V(Tu(t))$ is absolutely continuous on compact subsets of $[0,\tau^*)$ and satisfies \eqref{eq:Vdot} with $f$ as in \eqref{eq:target-full}.
\end{proposition}

\begin{proof}[Proof of Proposition~\ref{prop:wp}]
We first prove local existence and uniqueness by contraction in target coordinates. Fix
$0<\nu_0<\nu$ and use the weighted norm
$\|f\|_{\mathcal H^{-2}_\gamma}$ defined in \eqref{eq:target-forced}. Equivalence follows from
$cn^4\le\gamma_n\le Cn^4$ by Proposition~\ref{prop:target-lin} and
Theorem~\ref{thm:generic}(ii). Write $w=Tu$, $u=T^{-1}w$, $e=\hat h-h$, and work in
\begin{align}
X_\tau&=\;C([0,\tau];\mathcal H^0)
       \cap L^2(0,\tau;\mathcal H^2),
\qquad
\|w\|_{X_\tau}
=\sup_{0\le t\le\tau}\|w(t)\|_{\mathcal H^0}
 +\left(\int_0^\tau\Sigma(w(t))\dd t\right)^{1/2}.\label{eq:app-i-display-003}
\end{align}
Define $\Phi(w)$ as the mild solution of the linear target system with
initial condition $Tu_0$ and forcing $f(w)=B_2\ip{e}{u}-T(uu_x)$ from \eqref{eq:target-full}.
A fixed point gives a solution in the sense of Definition~\ref{def:solution}.
Integrating the identity of Proposition~\ref{prop:target-lin}(ii) with
\eqref{eq:app-f-display-014} at $\delta=1/2$ gives
\begin{align}
\|\Phi(w)\|_{X_\tau}
&\le\;3\|Tu_0\|_{\mathcal H^0}
 +4\|f(w)\|_{L^2(0,\tau;\mathcal H^{-2}_\gamma)}.\label{eq:app-i-display-009}
\end{align}

In the proof of Lemma~\ref{lem:forcing}, replace $w$ in the second slot
by an arbitrary $v\in\mathcal H^2$. Lemma~\ref{lem:smoothing}(ii) and
Lemma~\ref{lem:gamma} control $\Gamma T^\sharp v$ in $H^1_0(0,1)$.
Using $\|v\|_{\mathcal H^1}\le c_\gamma^{-1/2}\Sigma(v)^{1/2}$ from
\eqref{eq:dissipation-constants} and taking the supremum over
$\Sigma(v)\le1$ gives the first bound below. Spectral interpolation
and the $L^2$ norm equivalence give the second:
\begin{align}
\|T(uu_x)\|_{\mathcal H^{-2}_\gamma}
&\le\;C\|u\|^{3/2}\|u\|_{\mathcal H^1}^{1/2}
\le C\|u\|^{7/4}\|u\|_{\mathcal H^2}^{1/4}.\label{eq:nlH2}
\end{align}
H\"older's inequality in time and boundedness of $T^{-1}$ on
$\mathcal H^0$ and $\mathcal H^2$, from Lemma~\ref{lem:smoothing}(ii), give
\begin{align}
\|T(uu_x)\|_{L^2(0,\tau;\mathcal H^{-2}_\gamma)}
&\le\;C\tau^{3/8}\|w\|_{X_\tau}^2.\label{eq:app-i-display-006}
\end{align}
For the boundary forcing, \eqref{eq:Lambda} gives
$\|B_2\|_{\mathcal H^{-2}_\gamma}=\Lambda$, so
\begin{align}
\|B_2\ip{e}{u}\|_{L^2(0,\tau;\mathcal H^{-2}_\gamma)}
&\le\;C\Lambda\|e\|\tau^{1/2}\|w\|_{X_\tau}.\label{eq:app-i-display-007}
\end{align}
For differences, use
$uu_x-vv_x=\tfrac12\partial_x((u-v)(u+v))$.
The same product estimates and the linear estimate with zero initial data give
\begin{align}
\|\Phi(w)-\Phi(z)\|_{X_\tau}
&\le\;C\big[\tau^{3/8}(\|w\|_{X_\tau}+\|z\|_{X_\tau})
 +\Lambda\|e\|\tau^{1/2}\big]\|w-z\|_{X_\tau}.
\label{eq:local-contraction}
\end{align}

Choose $R=6\|Tu_0\|_{\mathcal H^0}+1$. Estimates
\eqref{eq:app-i-display-009}, \eqref{eq:app-i-display-006}, and
\eqref{eq:app-i-display-007} show that $\Phi$ maps the closed radius-$R$
ball of $X_\tau$ into itself for sufficiently small $\tau$.
Estimate \eqref{eq:local-contraction} makes it a strict contraction.
Banach's fixed-point theorem \cite[Thm.~5.7]{Brezis2011} gives a unique
fixed point. The choice of $\tau$ depends only on the stated data and
the design. Applying the same difference estimate on short intervals
gives uniqueness on overlapping intervals.

We next prove continuation and the blow-up alternative. Restarting the local construction gives a maximal
interval of existence. The existence time is uniform when the initial
norm is bounded. If the solution remained bounded along times approaching
a finite $\tau^*$, restarting at one of those times would extend it
beyond $\tau^*$, proving the blow-up alternative.

Finally, we prove the Lyapunov identity. The forcing estimates place $f$ in
$L^2(0,\tau;\mathcal H^{-2})$ on every compact subinterval of existence.
Proposition~\ref{prop:target-lin}(ii) therefore gives absolute continuity
of $V(Tu)$ and identity \eqref{eq:Vdot}.
\end{proof}

\end{appendices}

\bibliographystyle{plain}
\bibliography{references}
\end{document}

%% file: diagrams/styles.tex
\definecolor{ksblue}{HTML}{1764AB}
\definecolor{ksorange}{HTML}{C25319}
\tikzset{
  ksbox/.style={draw=black!55,rounded corners=2pt,fill=black!2,
    align=center,inner sep=5pt,font=\fontsize{9}{10.5}\selectfont,minimum height=8mm},
  ksdesign/.style={ksbox,draw=ksblue,fill=ksblue!5},
  ksgain/.style={ksbox,draw=ksorange,fill=ksorange!5},
  ksarrow/.style={-{Latex[length=2mm]},line width=.7pt},
  kssignal/.style={font=\footnotesize,fill=white,inner sep=1.5pt},
  ksheading/.style={font=\small\bfseries,anchor=west},
  ksnote/.style={font=\footnotesize,align=center}
}

%% file: diagrams/part1-overview.tex
\begin{tikzpicture}[x=1cm,y=1cm]
\node[ksheading] at (0,0) {(a) Design sequence};
\node[ksheading] at (8.8,0) {(b) Closed-loop implementation};

\node[ksdesign,text width=3.0cm,minimum height=1.9cm] (pre) at (1.65,-1.4)
  {\textcolor{ksblue}{\textbf{1. Pre-feedback}}\\[4pt]
   Choose $(\varepsilon,\theta)$\\
   Make every mode\\
   controllable from $U_2$};
\node[ksgain,text width=3.4cm,minimum height=1.9cm] (fredholm) at (5.65,-1.4)
  {\textcolor{ksorange}{\textbf{2. Fredholm gain}}\\[4pt]
   Choose the shift $a$\\
   \mbox{to assign the decay rate}\\[2pt]
   Construct $h$ and $T$};
\draw[ksarrow] (pre.east)--(fredholm.west);
\node[ksnote,text width=7.4cm] at (3.65,-2.9)
  {Displacement enables controllability.\\
   Curvature feedback stabilizes the plant.};

\node[ksdesign,text width=2.35cm,minimum height=.85cm] (uone) at (10.35,-.85)
  {\textcolor{ksblue}{\textbf{Displacement}}\\
   $U_1=\varepsilon\langle g_\theta,u\rangle$};
\node[ksgain,text width=2.35cm,minimum height=.85cm] (utwo) at (10.35,-2.05)
  {\textcolor{ksorange}{\textbf{Curvature}}\\
   $U_2=\langle h,u\rangle$};
\node[ksbox,text width=2.2cm,minimum height=2.05cm] (plant) at (14.1,-1.45)
  {\textbf{KS plant}\\[4pt]
   $u(0)=U_1$\\
   $u_{xx}(0)=U_2$};
\draw[ksarrow,draw=ksblue] (uone.east)--node[kssignal,above] {$U_1$} (plant.west |- uone.east);
\draw[ksarrow,draw=ksorange] (utwo.east)--node[kssignal,above] {$U_2$} (plant.west |- utwo.east);
\draw[ksarrow] (plant.east)--(16.35,-1.45);
\node[kssignal,anchor=south east] at (16.4,-1.4) {$u(\cdot,t)$};
\draw (16.0,-1.45)--(16.0,-3.0)--(8.8,-3.0)--(8.8,-.85);
\draw[ksarrow] (8.8,-.85)--(uone.west);
\draw[ksarrow] (8.8,-2.05)--(utwo.west);
\fill (8.8,-2.05) circle (1pt);
\fill (16.0,-1.45) circle (1pt);
\node[kssignal] at (12.3,-3.0) {full state $u(\cdot,t)$};
\end{tikzpicture}